\documentclass[11pt]{article}
\pdfoutput=1 % if your are submitting a pdflatex (i.e. if you have
\usepackage{jheppub} % for details on the use of the package, please
\usepackage[T1]{fontenc} % if needed

\usepackage{verbatim}

\usepackage{amsmath}
\usepackage{amssymb}
\usepackage{amsthm}
\usepackage{mathrsfs}
\usepackage{physics}
\usepackage{bbm}
\usepackage{enumitem}

\usepackage[dvipsnames]{xcolor}

\DeclareSymbolFont{extraup}{U}{jkpsyc}{m}{n}
\DeclareMathSymbol{\vardiamondsuits}{\mathalpha}{extraup}{113} 

\def\vardiamondsuit{{\Large \vardiamondsuits}}

\usepackage{tikz-cd}

\newtheorem{theorem}{Theorem}
\newtheorem{theorem*}{Theorem}[section]
\newtheorem{lemma}{Lemma}
\newtheorem*{lemma*}{Lemma}

\def\H{H}
\def\varPhi{\Omega}

\newtheorem*{corollary*}{Corollary}
\newtheorem{corollary}{Corollary}

\theoremstyle{definition}
\newtheorem{definition}{Definition}
\newtheorem*{conventions}{Conventions}

\theoremstyle{remark}
\newtheorem{remark}{Remark}

\DeclareMathOperator{\Hil}{\mathscr{H}}

\DeclareMathOperator{\non}{\nonumber}

\DeclareMathOperator{\B}{\mathcal{B}}

\DeclareMathOperator{\LL}{\mathcal{L}}

\DeclareMathOperator{\N}{{\mathcal{N}}}

\author{Thomas Faulkner,}
\emailAdd{tomf@illinois.edu}
\author{Min Li}
\emailAdd{minl2@illinois.edu}
\affiliation{Department of Physics, University of Illinois,\\ 1110 W. Green St., Urbana, IL 61801-3080, USA}

\title{\boldmath  Asymptotically isometric codes for holography}
\abstract{The holographic principle suggests that the low energy effective field theory of gravity, as used to describe perturbative quantum fields
about some background has far too many states. 
It is then natural that any quantum error correcting code with such a quantum field theory as the code subspace
is not isometric. We discuss how this framework can naturally arise in an algebraic QFT treatment of a family of CFT with a large-$N$ limit described by the single trace sector. 
We show that an isometric
code can be recovered in the  $N \rightarrow \infty$ limit when acting on  fixed states in the code Hilbert space. 
Asymptotically isometric codes come equipped with the notion of simple operators and nets of causal wedges. While the causal wedges are additive, they need not satisfy Haag duality, thus leading to the possibility of non-trivial entanglement wedge reconstructions. Codes with complementary recovery are
defined as having extensions to Haag dual nets, where entanglement wedges are well defined for all causal boundary regions. 
We prove an asymptotic version of the information disturbance trade-off theorem and use this to show that boundary theory causality is maintained by net extensions. We give a characterization of the existence of an entanglement wedge extension via the asymptotic equality of bulk and boundary relative entropy or modular flow.  
While these codes are asymptotically exact, at fixed $N$ they can have large errors on states that do not survive the large-$N$ limit. This allows us to fix well known issues that arise when modeling gravity as an exact codes, while maintaining the nice features expected of gravity, including, among other things, the emergence
of non-trivial von Neumann algebras of various types. 
 }

\begin{document}

\maketitle

\section{Introduction}

Quantum error correction (QEC) is a useful paradigm for studying AdS/CFT  and holography \cite{Dong:2016eik, pastawski2015holographic, Almheiri:2014lwa, Harlow:2016vwg, Kang:2018xqy, Faulkner:2020hzi}. Small (i.e. finite dimensional) code subspaces
embedded into the effective field theory of gravity are protected from the action of localized operators in the microscopic CFT, allowing them
to be decoded on multiple boundary subregions. The encoding map from bulk to boundary, in this case can be taken as an isometry.

An important open question is to what extent bulk locality, as described by local quantum fields about some background geometry, can be understood in this QEC paradigm? Since dynamics is built into relativistic QFT via local spacetime symmetries this question is related to the challenge of building dynamics into holography codes.  Dynamical constraints on gravity such as the quantum focusing conjecture \cite{bousso2016quantum} reduce to causality constraints on local QFT algebras \cite{Balakrishnan:2017bjg,Ceyhan:2018zfg}, and we would like to see these naturally arise in the QEC paradigm. 
 In order to achieve this one might expect that the code subspace must be large enough to encode the quantum field theory Hilbert space and it is not entirely clear how this should work.  In fact there is a certain sense in which local quantum fields are incompatible with the traditional QEC paradigm - even with a short distance cutoff, there are  more states in QFT than are allowed by holographic bounds \cite{Bekenstein:1973ur,'tHooft:1993gx, susskind1994black,Susskind:1998dq, Bousso:2002ju} and many of these states should not map to anything sensible in the boundary/microscopic theory. Instead these holographic
bounds are enforced via backreaction effects and the appearance of black holes. 

From an operator algebra point of view, it is known that the collection of operators describing the low energy bulk theory are in fact not algebras at all \cite{El-Showk:2011yvt, Papadodimas:2013jku}.
On the other hand if the bulk was well described by QFT in some local spacetime region then the general expectation is that the algebra of operators should be that of a type-III$_1$ von Neumann algebra \cite{haag2012local}. In contrast to type-I von Neumann algebras,  more typically studied in quantum information theory, these algebras give the correct
mathematical structure that explains the necessarily infinite
entanglement present in the QFT adiabatic vacuum near the edge of the bulk region \cite{Witten:2018lha}. In a gravitational theory such infinite entanglement is regularized to a finite quantity \cite{Bekenstein:1973ur,solodukhin2008entanglement}, as measured by the generalized entropy, yet the resulting entropy
scales as $1/G_N$ so, working perturbatively about $G_N = 0$, one expects that such QFT like algebras could still arise.

 Indeed recent papers have argued that exactly such
a von Neumann algebra description is emergent in the dual of the thermofield double state in the large-$N$ QFT \cite{Leutheusser:2021qhd,Leutheusser:2021frk, Witten:2021unn}.
 The goal of this paper is to extend the QEC description of holography in a way that allows for such non-trivial QFT algebras acting on the code subspace. In so doing we will extend
 these recent papers, studying large $N$ von Neumann algebras, beyond causal wedges to general entanglement wedges.
 Such a result for QEC codes has remained elusive, due to the necessary approximate nature of these codes. For exact QEC codes, it
 is easy to construct examples with bulk QFT algebras \cite{Kang:2018xqy, Faulkner:2020hzi}. However they fail to genuinely describe holographic QFTs due to the additivity issue discussed in \cite{Kelly:2016edc, Faulkner:2020hzi}. In the approximate case the focus has been on models with finite dimensional Hilbert spaces \cite{Cotler:2017erl,Hayden:2018khn,Akers:2020pmf}. Some infinite dimensional results exists, however the type-III$_1$ case has remained elusive \cite{Gesteau:2021jzp}.

Here we directly confront the incompatibility of local QFT and error correcting codes: since effective field theory
over counts the degrees of freedom in the microscopic theory it is natural that the code is non-isometrically embedded, using a bounded operator $V_N$, into the microscopic theory. In particular we expect that certain bulk states map to the same boundary state implying the existence of a non-trivial kernel for $V_N$. In this way we can preserve the local predictions of the bulk effective theory.

We discuss a version of this paradigm that arises from a careful study of the large-$N$ limit of some boundary QFT, with the expectation
that $G_N^{-1} = \mathcal{O}( N^{\#})$ for some positive power. Should it exist, the gravitational effective field theory arises
as an expansion about the $N=\infty$ point as an algebra generated by the single-trace operators. These are the operators with a smooth large $N$ limit.  These algebras are
state dependent in the weak subspace dependent sense. The boundary/microscopic theory is then a fixed-$N$ CFT. Our code is built from the single trace algebra. However, at any fixed-$N$ the single trace operators along with their multi-trace composites do not actually form an algebra. Hence
the decoding maps $\gamma_N$ from the single trace description to the finite $N$ microscopic theory are not necessarily algebra preserving homomorphisms - indeed it need not be a unital completely positive map, i.e. a quantum channel.  Instead these nice features are only recovered in a limiting procedure. This leads to the non-isometric approximate codes, albeit ones that become isometric in the limit.

 By approximate we do not mean the standard notion of approximate error correction, where
reconstruction errors are controlled by some function of $N$ which vanishes as $N \rightarrow \infty$. Instead we only demand the code works exactly in the limit $N \rightarrow \infty$. In particular this code does poorly (not even approximately well) for certain ``black hole'' states of the finite $N$ theory. These are the states that do not have a smooth large-$N$ limit.
 The distinction between our notion of approximate and the more standard notion arises from the distinction between uniform convergence and pointwise convergence. In particular we find these codes become isometric only when probing $V_N^\dagger V_N$ in matrix elements using fixed states on the code subspace ({\bf pointwise}): from the theory of operator algebras this is the notion of weak operator convergence.  This should be compared with norm convergence ({\bf uniform}). 
 We will refer to codes that are at least weakly isometric, in the limit $N \rightarrow \infty$, simply
 as \emph{asymptotically isometric codes.}

The importance of non-isometric codes for studying gravity was recently discussed in \cite{Akers:2022qdl, Akers:2021fut, Penington:2019npb}. The context here was an attempt
to give a Hilbert space interpretation of the Page curve computations \cite{Almheiri:2019psf,Penington:2019npb}. In these cases the island rule \cite{Almheiri:2019hni} imposes a different
kind of holographic bound, screening the large Hilbert space inside an evaporating Black Hole that is entangled with the radiation.
This large Hilbert space arises in the effective description and is non-isometrically encoded in the smaller black hole Hilbert space. 
We expect the non-isometric codes we discuss here can be used to study the Page curve computations and that in this case the non-isometric nature
of our codes will be related that of \cite{Akers:2022qdl}. 
In the discussion Section~\ref{sec:d} we make some preliminary comments on how to
use asymptotic codes in this case, in particular we emphasize the need for two codes before and after the Page time. We then discuss how for these codes the Hawking paradox arises due to a non-commutativity of the large-$N$ limit and the $\mathcal{O}(N^\#)$ time evolution required to see the Page transition. 

In this paper we will study several aspects of asymptotically isometric codes. In Section~\ref{sec:sta} we show how such codes might arise from the single trace algebra of a matrix like theory with a large $N$ limit. The construction is based on a sequence of states that has a well defined large-$N$ limit when computing correlation functions of the single trace fields. To give a more rigorous discussion we describe the single trace fields using a Weyl algebra. These algebras are also sometimes called generalized free fields.
Under certain technical assumptions on the resulting von Neumann algebras and on the mode of convergence for the correlation function, we extract our asymptotic code. 
In particular we highlight the important assumption that the von Neumann algebras under study are \emph{hyperfinite} meaning they can be generated by an increasing family of finite dimensional von Neumann sub-algebras. Most physically relevant von Neumann algebras are hyperfinite - in particular the type-III$_1$ von Neumann algebras of local regions in QFTs that have reasonable thermodynamic behavior, are hyperfinite. 

In the absence of a derivation of these specific technical assumption we still claim that these asymptotic codes are of independent interest. In particular the rest of the paper, Section~\ref{sec:aic}-\ref{sec:p} is devoted to a study of these codes. We find that many properties of these codes are related to the expected properties of holographic QFTs.

In Section~\ref{sec:aic} we give our main definition of an asymptotically isometric codes, which corresponds to the triple $(\mathcal{C},V_N, \gamma_N)$ where $\mathcal{C}$
defines a net (in the algebraic QFT sense) of von Neumann algebras associated to boundary causal regions, $V_N$ are the non-isometric maps from bulk to boundary and $\gamma_N$ are operator maps that send the single trace fields to a definite operator in the finite $N$ theory. These can be interpreted as simple QEC decoding/reconstruction maps and for this
reason it is natural to think of $\mathcal{C}$ as the causal wedges. We make the observation the local algebras in $\mathcal{C}$ need not satisfy Haag duality even if the underlying theory has this property. 
Entanglement wedges exist if such operator reconstructions can be extended beyond the causal wedge to Haag dual nets. In Section~\ref{sec:crhd} we prove that any extensions of the causal wedge must satisfy basic constraints of boundary theory causality, such as entanglement wedge nesting. We also show that the additivity issue identified in \cite{Kelly:2016edc,Faulkner:2020hzi} is not present. 
These results are based on Theorem~\ref{thm:qec} an asymptotic form of information-disturbance tradeoff that plays an important role
in the theory of QEC \cite{kretschmann2008information,kretschmann2008continuity,crann2016private} and in turn AdS/CFT \cite{Dong:2016eik,Harlow:2016vwg,Kang:2018xqy,Faulkner:2020hzi}.
We study the conditions under which the extension is maximal, culminating with the statement of the JLMS condition \cite{Jafferis:2015del} 
and the equality of bulk and boundary modular flow \cite{Faulkner:2017vdd}
 in Theorem~\ref{thm:JLMS}. We defer the proofs of these results to Section~\ref{sec:p}. The JLMS proof notably uses Petz's powerful work on sequences (or nets) of quantum channels \cite{petz1994discrimination,ohya2004quantum} and in particular we find that the standard Petz map can be used as the recovery map with no extra twirling \cite{junge2018universal,Cotler:2017erl,Faulkner:2020iou} required.\footnote{For finite dimensional codes this was established in \cite{Chen:2019gbt}.}

State-dependent entanglement wedges and back-reaction effects can be understood as arising from different sectors associated to different large-$N$ limits.  These limits correspond to different sequences of states.  In a paper under preparation \cite{toappear} (and briefly summarized in the discussion section) motivated by the hypothesis that the full large-$N$ limit of some boundary theory requires many different asymptotic codes,
we study possible ways to patch the codes together.  We then study the vacuum sector and thermal sectors in some detail. In particular we prove conditions under which the entanglement wedge must equal the causal wedge for such codes.
We then use sector theory to discuss the appearance of the crossed product and type-II$_\infty$ \cite{Witten:2021unn, Chandrasekaran:2022eqq} algebras acting on thermal codes using so called time-shifted thermal codes.

While we do not prove a general condition under which an entanglement wedge exists we make some speculative comments in the discussion section connecting:
1. the survival of the split property on a sufficiently complete code subspace, 2. saturation of the modular chaos bound, 3. the appearance of a gravitational effective theory and  4. the existence of an entanglement wedge for all boundary regions deformable to a sphere. 

Sections~\ref{sec:aic}-\ref{sec:p}, which defines and studies our codes, can be read independently of Section~\ref{sec:sta} which extracts them from
the large-$N$ limits of certain QFTs.
In the next Section~\ref{sec:aqft}, to get us started, we review some aspects of algebraic QFT that we will need throughout the paper.

\section{Algebraic QFT}
\label{sec:aqft}

\begin{conventions}
We denote vectors in a Hilbert space with Dirac notation
$\left| \psi \right>, \left| \eta \right>, \ldots $, and we will also sometimes use an explicit inner product notation $( \left| \psi \right>, \left| \eta \right>)$ where the \emph{first} entry is anti-linear. Define the adjoint action ${\rm Ad}_u(\cdot) = u^\dagger \cdot u$.
We will conflate a closed subspace of a Hilbert space with the projection to this subspace. 
Complex Hilbert spaces will typically 
be denoted with mathscr, $\mathscr{H}, \mathscr{K}, \ldots$, algebras with mathcal $\mathcal{M},\mathcal{N}, \mathcal{W} \ldots$, and spacetime regions with capital letters $K, O, \ldots$.  The notation $\omega_{\psi}(\cdot) \equiv \left< \psi \right| \cdot \left| \psi \right>$ will be used a lot. By quantum channel we mean a normal unital completely positive map
between von Neumann algebras. Normal states and positive maps are ultraweakly continuous. Limits on various spaces will, by default, refer to the natural norm on that space.
So for example $\lim_i a_i = 0$ for bounded operators on $\mathscr{H}$: $a_i \in \mathcal{B}(\mathscr{H})$ means that $\lim_i \| a_i \| = 0$.
 $so-\lim_i a_i$ means that $\lim_i \| a_i \left| \psi \right> \| = 0$ for all vectors $\psi \in \mathscr{H}$ and $wo-\lim_i a_i$ means that $\lim_i \left< \psi_1 \right| a_i \left| \psi_2 \right> = 0$ for all vectors $\psi_{1,2} \in \mathscr{H}$. The pointwise limit of linear functionals $w-\lim_i \rho_i = 0$ simply means that $\lim_i \rho_i(a) = 0$ for each $a$.
Occasionally we will need some other limits/operator topologies. 
We review these in Appendix~\ref{app:convergence}. 
The name \emph{state} will either refer to vectors on a Hilbert space, or positive linear functionals on some operator algebra. The later will typically be normal states (ultraweakly continuous) if not otherwise specified. 
Quantum channels are defined to be normal unital completely positive maps between two von Neumann algebras. 
\end{conventions}

Consider a $d$-dimensional CFT, defined on the Lorentzian cylinder $K = S_{d-1} \times \mathbb{R}$.\footnote{By working on the cylinder instead of the conformally related Minkowski space, we avoid having to discuss various singular cases in the assignment of algebras to regions \cite{Duetsch:2002hc}. For the so called conformal nets (corresponding to chiral CFTs) this is the distinction between
a net over $\mathbb{R}$ and a net over $\mathbb{S}^1$. The relation between algebras in the two different nets is complicated by the addition of the point at $\infty$
to the $\mathbb{R}$ net.} 
We will assign local algebras to causally complete subregions $\{ O \subset K : O'' = O\}$ of $K$.\footnote{For some of the properties stated below a more restrictive
set of regions might be necessary due to issues of regularity. For example one might stick to regions causally generated by finite unions of causal diamonds $\diamondsuit$ defined below. This would be sufficient for our purposes, yet the details of this restriction will not be important in this paper.}
The $'$ notation, applied to open subregions of $O \subset K$, refers to the largest open set that is space-like separated to $O$ and we use the causal structure of the Lorentzian cylinder.  The entire cylinder $K$ also counts as a causally complete region with $K' = \emptyset$. 
A causally complete region $O$ is called local if the closure of $O$ is a non-empty compact subset of $K$. Generally a subset labelled $O \subset K$ will be assumed causally complete if not otherwise stated. There is a special class of such local regions which are called double cones $\diamondsuit = I^+(p) \cap I^-(q)$ where $p,q \in K$ with $p\in I^-(q)$ and where $I^\pm(p)$ denote the set of points in $K$ that are to the chronological future/past of $p$. Only for $p,q$ sufficiently close are these diamonds causally complete and
we only use the label $\diamondsuit = \diamondsuit''$ when this is the case. 
The conformal group $g\in\widetilde{SO}(d,2)$ (where the tilde denotes the universal cover) acts on causally complete regions
in $K$ via $O \rightarrow g(O)$ in the obvious manner. 

When working with both boundary and bulk operator algebras we will consider the following axiomatic framework, based
on ideas developed in algebraic QFT \cite{Haag:1963dh}. See \cite{haag2012local} for a textbook treatment. We have a Hilbert space $\mathscr{K}$ and
a net of von Neumann algebras $O \rightarrow \mathcal{M}(O)$ that acts on $\mathscr{K}$ and which is indexed by causally complete regions $O$ as defined above. 
The local net over $K$ will be denoted in full as $K \supset O \rightarrow \mathcal{M}(O) \subset \mathcal{B}(\mathscr{K})$ and we will simply use $\mathcal{M}
\equiv \mathcal{M}(K)$ as a shorthand when
the rest of the net structure is understood.  Generally we will take $\mathcal{M}(\emptyset) = \mathbb{C} 1$. We might include some of the following axioms:
\begin{enumerate}
\item \emph{Isotony}:
\begin{equation}
\label{isotony}
\mathcal{M}(O_1) \subset \mathcal{M}(O_2) \qquad \qquad \forall \, O_1 \subset O_2 \,\,,\,\, O_{1,2} \subset K
\end{equation}
\item \emph{Causality}:
\begin{equation}
\mathcal{M}(O') \subset \mathcal{M}(O)'  \qquad \forall\,O \subset K
\end{equation}
where the prime acting on an algebra gives the commutant. 
\item \emph{Standard and factorial}:
 There exists a vector $\eta \in \mathscr{K}$ for which $\mathcal{M}(O)$ is cyclic and separating whenever $O$
 is local. The local algebras $\mathcal{M}(O)$ are factors. 
  \item  \emph{Additivity}: For any collection of double cones $\{ \diamondsuit^i \subset O \}_{i \in I}$ such that $\bigcup_i \diamondsuit^i = O$ then:
\begin{equation}
\mathcal{M}(O) = 
\left( \bigcup_{i \in I} \mathcal{M}(\diamondsuit^i) \right)''
\end{equation}
see Figure~\ref{fig:causal-regions}
 \item  \emph{Completeness}:
\begin{equation}
 \mathcal{M}( S_{d-1} \times \mathbb{R}) = \mathcal{B}(\mathscr{K})
 \end{equation}
 \item \emph{Haag duality}: 
\begin{equation}
\mathcal{M}(O') = \mathcal{M}(O)' \, \qquad \forall  \,\, O \subset K
\end{equation}
\item \emph{Hyperfinite}: The local von Neumann algebras $\mathcal{M}(O)$ are hyperfinite 
type-III$_1$ factors. 
 \item \emph{Split property}: For two local regions $O_{1,2} \subset K$ such that $O_1 \subset O_2$ is a proper subset there exists a type-I factor $\mathcal{N}$:
\begin{equation}
\mathcal{M}(O_1) \subset \mathcal{N}  \subset \mathcal{M}(O_2) 
\end{equation}
\item \emph{Conformal symmetry}: The Hilbert space $\mathscr{K}$ furnishes a strongly continuous representation $U(g)$ of the conformal group $g \in \widetilde{SO}(d,2)$, with
non-negative spectrum for the energy operator corresponding to time translations along $\mathbb{R}$. There is a vacuum vector $\Omega$ that is invariant $U(g) \left| \Omega \right> = \left| \Omega\right>$ and which is locally standard for $\mathcal{M}$. Also the $U(g)$ act covariantly on the net:
 \begin{equation}
U(g) \mathcal{M}(O) U(g)^\dagger 
= \mathcal{M}( g( O))
\label{ref:vNnet}
 \end{equation}

 \end{enumerate}
 
 \begin{figure}[!ht]
\centering
\includegraphics[scale=.4]{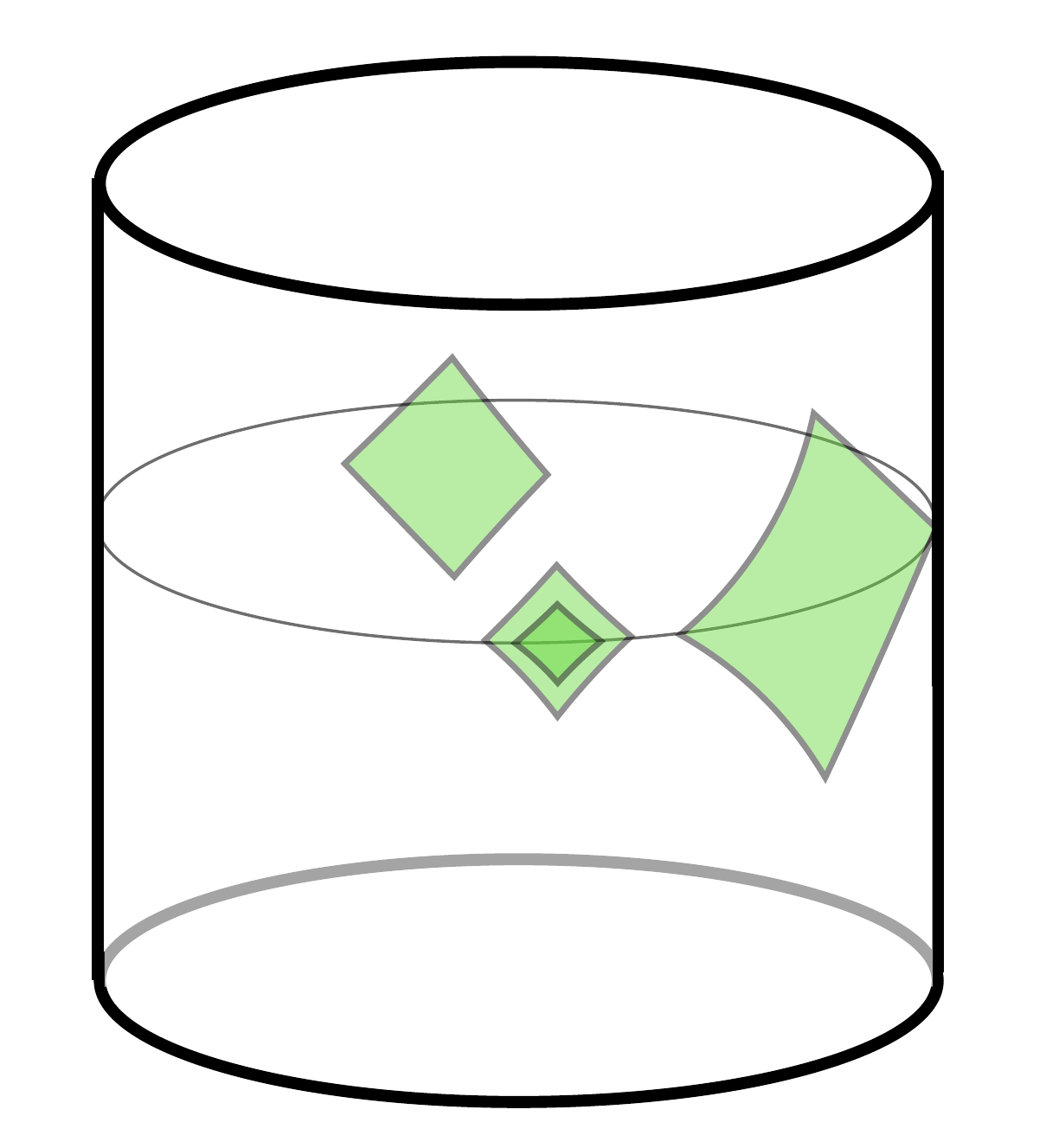} \hspace{2cm} \includegraphics[scale=.4]{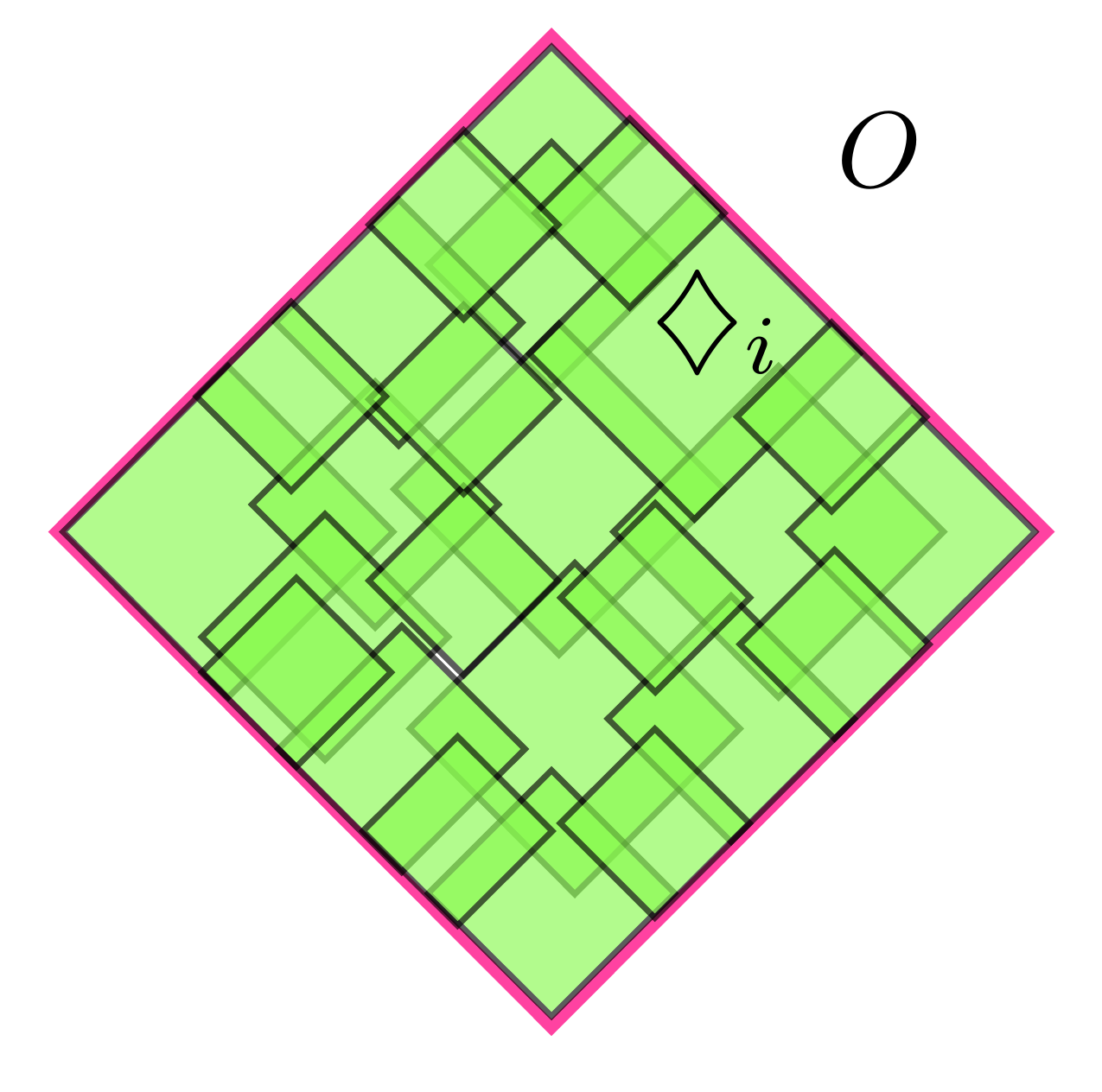}
\caption{ (\emph{left}) Some examples of causal regions on $S^{d-1} \times \mathbb{R}$. (\emph{right}) The additivity property we use in this paper considers
spacetime covers of $O$ by causal diamonds.
\label{fig:causal-regions}}
\end{figure}

Rather than quote these axioms individually every time we consider a new net, we will instead collect some axioms together and give these names:
\begin{definition}
\label{def:vN}
We refer to an \emph{additive net} as a local net $\mathcal{M}$ that satisfies axioms (1-4). 
We refer to a \emph{complete net} as a local net $\mathcal{M}$ that satisfies axioms (1-6). A net is additionally \emph{hyperfinite/split} if it also satisfies (7)/ (7-8) [split implies hyperfinite \cite{Buchholz:1986bg,Longo:1982zz}]. In general we may take $K$, in the definition of the net, to be different from $S_{d-1} \times \mathbb{R}$ (such as by including a reference/bath system) and in this case
the axioms $(1-9)$ as stated could still be applied, as long as we substitute in a new definition of ``local''
and specify the geometric action of the conformal group on this new $g : K \rightarrow K$.  We say that the net is in the vacuum sector if it also satisfies (9) for some pair $(U(g),\Omega)$. 
\end{definition}

It is well known that some of these conditions follow from others. In this paper we will not concern ourselves too much with the question of what are the minimal set of axioms.
At the same time some of these conditions are too restrictive in general. For example fermions should be included by grading the algebra and allowing for anti-commutation relations, see for example \cite{Pontello:2020csg}. 
 And we have  ignored the possibility of superselection sectors \cite{Doplicher:1971wk, Doplicher:1973at, haag2012local,Brunetti:1992zf, fredenhagen1990generalizations,Pontello:2020csg}, by working in a fixed representation.
 However these axioms will be, for the most part, sufficient for our purposes. For the fixed-$N$ boundary theories, local excitations of the vacuum on $S^{d-1} \times \mathbb{R}$ already contains all the interesting
 physics of AdS/CFT and in the presence of global charges we can work with the field algebra and this will give us a complete net with no superselection sectors. 
It is important to generalize our work to the fermion/anticommutator case, but we lose very little of the physics that we plan to expose by doing so. For the bulk algebras, as we will explain below, the theory of (superselection) sectors seems to work quite differently to the standard QFT case. This is because sector physics is tied to the large-$N$ limit. Instead the error correcting codes that we construct should be thought of as describing a single sector of the theory for which the above axioms are sufficient.

While we will always take the boundary theory to be described by a net of von Neumann algebras acting concretely on some $\mathscr{K}$, we will sometimes describe the bulk theory using $^\star$-algebras and $C^\star$-algebras. These are more abstract notions of operator algebras not requiring a Hilbert space representation.
A $^\star$-algebra is an associative algebra with an involution confusingly denoted here by $^\dagger$ and not $^\star$. $C^\star$-algebras additionally have a norm $\| \cdot \|$ satisfying the condition
$\| a^\dagger a \| = \| a \| \| a^\dagger \|$ and are closed under the norm topology induced by $\| \cdot \|$. We will always take such algebras to be unital, containing a unit $1$. 
Compared to von Neumann algebras
this is a stronger closure requirement, allowing for richer possibilities. Similar axioms can be given for nets of $^\star$-algebras and $C^\star$-algebras. A
net of $^\star$-algebras over $K$ is denoted $ K \supset O  \rightarrow \mathcal{W}(O)$, with shorthand $\mathcal{W} \equiv \mathcal{W}(K)$. In 
this paper such $^\star$-algebra nets are always assumed to satisfy:
\begin{enumerate}
\item \emph{Isotony}:
\begin{equation}
\mathcal{W}(O_1) \subset \mathcal{W}(O_2) \qquad O_1 \subset O_2
\end{equation}
\item \emph{Causality:}
\begin{equation}
[ \mathcal{W}(O_1), \mathcal{W}(O_2) ] = \{ 0 \} \, \qquad O_1 \subset O_2'
\end{equation}
\end{enumerate}
And for $C^\star$ algebras we always demand:
\begin{enumerate}
\item[3.] \emph{Additivity}: for any collection of double cones $\{ O^i \subset O \}$ such that $\bigcup_i O^i = O$ then:
\begin{equation}
\label{ref:Cnet}
\mathcal{W}(O) = 
\bigvee_i \mathcal{W}(O^i)
\end{equation}
where  for $C^\star$ algebras the $\vee_i$ operation refers to the smallest $C^\star$ algebra that
contains as sub-algebras each of the $\mathcal{W}(O^i)$. 
\end{enumerate}
\begin{definition}
\label{def:C}
We will refer to such a net satisfying (1-3), given immediately above, as an \emph{additive net} of $C^\star$ algebras. A \emph{subnet} $\mathcal{A} \subset \mathcal{W}$ is a net
of $^\star$-algebras for which $\mathcal{A}(O) \subset \mathcal{W}(O)$ for all $O \subset K$. 
\end{definition}

\section{Single trace algebra}

\label{sec:sta}

We now consider a family of CFT's labelled by $N \in \mathbb{N}$. We assume each CFT is described by a complete split net
in the vacuum sector $K \supset O \rightarrow \mathcal{M}_N(O) \subset \mathcal{B}(\mathscr{K}_N)$. To discuss a large $N$ limit we need to relate different theories. We now discuss a special class of operators that exists in each of the nets for any $N$ that will help us take the large-$N$ limit.  These are the single trace fields. This discussion will be necessarily semi-informal at points since there does not seem to be
a rigorous algebraic QFT definition of the relevant theories that we wish to work with. In any case we have the goal of abstracting a precise axiomatically defined mathematical structure that we will introduce in Section~\ref{sec:aic}, so this does not concern us too much.
Our approach follows the outlined procedure given in \cite{Leutheusser:2021qhd,Leutheusser:2021frk,Witten:2021jzq, Witten:2021unn}, albeit with a slightly different focus.

The single trace self-adjoint fields $\{ \phi_\alpha(x) \} $
generate an algebra that can be thought of as acting in the $N=\infty$ theory.  The number of such fields that we include could be finite or infinite depending on the theory we are studying.\footnote{$\mathcal{N}=4$ SYM, at large $N$ and large but fixed $\lambda$ (the `t Hooft coupling), will have an infinite number of stringy modes. Scaling
$\lambda \sim N^p$ for sufficiently small $p$ will truncate this, but there will still be an infinite number of KK modes from reduction on the $S^5$. Consistent truncations can presumably be considered where the number of single trace fields is finite.} Since such operators are distributional and unbounded it is convenient to work with the Weyl operators instead. That is, associated to
the single trace fields define:
\begin{equation}
 w(h)= \exp( i \int h^\alpha \phi_\alpha) 
 \end{equation}
  where  the $h \in \H$ are real smooth test functions on $S_{d-1} \times \mathbb{R}$ of compact support,
 valued in 
  the dual of the functional vector space spanned by the single trace fields. We denote the pairing $\left< h, \phi \right>
  = \int h^\alpha \phi_\alpha$ where the integral is over $S_{d-1} \times \mathbb{R}$ and we sum over different single trace fields $\alpha$.
The resulting operators will become a unitary subset $\mathcal{W}_u$ of the Weyl algebra that we wish to construct. 
We denote the correspondence $w: \H \rightarrow \mathcal{W}_u$. 
  
Furthermore it is possible to pick these functions $h$ to have compact support inside a causally complete region $O \subset K$ and this gives the vector space $h \in \H(O)$ a net structure
over the space of causally complete regions $K$. It satisfies isotony: $\H(O_1) \subset \H(O_2)$ for $O_1 \subset O_2$
as well as additivity in the sense that $h \in \H(O)$ can always be written as $h = \sum h_i$ for $h_i \in \H(O_i)$ where $O = \cup_i O_i$ is some cover.
To establish this, one uses a partition of unity associated to the covering $O_i$ \cite{roberts2004more, fredenhagen1990generalizations}.
  
 For each $w(h)$ there is a corresponding operator at finite $N$:
\begin{equation}
w_N(h) =  \exp( i \int h^\alpha \phi^{(N)}_\alpha)  \in \mathcal{M}_N
\end{equation}
where an example of such an operator, in a matrix like model, is $\phi^{(N)}_\alpha ={\rm Tr} M_1 M_2 \ldots M_k$ with $M_i$ the $N \times N$ matrices
and the trace here is the matrix model trace \cite{tHooft:1973alw,Coleman:1985rnk}.  The \emph{form} of the single trace expression is fixed, so operators associated to $\phi_\alpha$ exist in each of the fixed $N$ theories. This is the form of the single trace fields in non abelian gauge theories with adjoint matter. We are not restricted
to studying such theories, although it is useful to have this example in mind.

Given a sequence of states $\sigma_N \in \mathcal{B}(\mathscr{K})_\star$ we will eventually identify $ \phi_\alpha$ as the limit of $ \phi_\alpha^{(N)} - \big< \phi_\alpha^{(N)} \big>_{\sigma_N}$. 
We have normalized these operators so that the two point function has a smooth large-$N$ limit after the subtraction of the one point function. 
The limit of the two point function will determine the resulting algebraic structure on the Weyl operators. Two examples to keep in mind are the sequence of vacuum states $\sigma_N(\cdot) = \left< \Omega_N \right| \cdot \left| \Omega_N \right>$, such that the two point function of single trace (primary) operators are fixed by conformal invariance.
And also the thermal states:
\begin{equation}
\sigma_N^\beta(\cdot) = Z_N^{-1} {\rm Tr}_{\mathscr{K}_N} e^{ - \beta H_N }  (\cdot)
\end{equation}
where $H_N $ is the Hamiltonian of the CFT on $S^{d-1} \times \mathbb{R}$ and $Z_N$ is the partition function. 

Note that with this normalization three point functions of single trace fields are suppressed:
\begin{equation}
\left<  \phi^{(N)}_1  \phi^{(N)}_2  \phi^{(N)}_3 \right>^c_{\sigma_N} = C_{123} = \mathcal{O}(N^{-1})
\end{equation}
where $c$ means the connected correlator. Thus non-linear effects and the coupling $C_{123}$ kick in for Weyl operators scaling like $w_N( \hat{h} N )$. If we act
with an operator with such a scaling we would maintain the classical limit but we would also induce a new
classical background. Correspondingly the operators $w(\hat{h} N)$ do not have a smooth large-$N$ limit and so will not be described by the codes that we construct here.

The sources $h \in \H$ define a real linear vector space that we will equip with a symplectic form.
This is determined by the commutator Green's function for $\phi_\alpha$: 
  \begin{equation}
  \label{eq:beta}
  \beta(h_1, h_2) = \frac{1}{2} \int h_1^\alpha  \int h_2^\beta   \,\, i \left< [\phi_\alpha^1, \phi_\beta^2 ] \right>_\sigma
  \end{equation}
  where $\sigma$ is the limiting state of the $N=\infty$ theory. As we will discuss in Section~\ref{sec:deg}, $\beta$ has degeneracies associated
 to the stress tensor and conserved currents. Aside from these cases, that we will ignore for now, we expect that $\beta$
  is non-degenerate, meaning that there is no $h \in \H$ such that $\beta(h,h_2) =0$ for all $h_2 \in \H$. 
   This is to be contrasted to a free theory where the commutator Green's function 
  of the fundamental fields would lead to a large degeneracy. This is because $\phi_\alpha$ would then satisfy some equation of motion that would make $\beta$ degenerate on the range
  of the equations of motion. Hence for free QFT the phase space is constructed as a quotient of this space.  
  Here we are working with generalized free fields so there is no such equation of motion. 
  
  The (state dependent) symplectic form $\beta$ will determine the algebra and a positive symmetric bilinear form $\alpha$ determines the state on the limiting theory
  and this is determined by the anti-commutator Green's function:
  \begin{equation}
  \label{acgreen}
\alpha(h_1,h_2) =  \frac{1}{2} \int h_1^\alpha    \int h_2^\beta   \,\,  \left< \{ \phi_\alpha^1, \phi_\beta^2  \} \right>_\sigma
  \end{equation}
  which one can prove always satisfies:
  \footnote{ Which follows from:
\begin{equation}
\left< ( \phi_1 + \lambda \phi_2)^\dagger ( \phi_1 + \lambda \phi_2) \right> \geq 0
\end{equation}
with $\lambda = 1/2 \left< [\phi_1, \phi_2]\right>/\left< \phi_2 \phi_2 \right>$ giving:
\begin{equation}
4 \left< \phi_1 \phi_1 \right> \left< \phi_2 \phi_2 \right> -    | \left< [\phi_1, \phi_2]\right>|^2  \geq 0
\end{equation}}
  \begin{equation}
  \label{bdab}
  \beta(h_1, h_2) ^2 \leq \alpha(h_1, h_1) \alpha(h_2,h_2)
  \end{equation}
We say that $\alpha$ is compatible with $\beta$. Using $\alpha$ we can turn $H$ into a real Hilbert space. We should then complete $H \rightarrow \overline{H}$, by including limits, with respect to the associated norm. We continue to consider the case where $\beta$ is degenerate even under extension to $\overline{H}$.

The bound \eqref{bdab} allows us to introduce \cite{Kay:1988mu,petz1990algebra,Hollands:2017dov,Longo:2021rag} a unique (up to unitary equivilence) one particle structure $(\kappa, \mathscr{H}_1)$ 
where $\kappa : \H \rightarrow \mathscr{H}_1$ and $\mathscr{H}_1$ is the 
one particle Hilbert space with the inner product satisfying:
\begin{equation}
\label{stdsub}
\left( \kappa(h_1), \kappa(h_2) \right)_{\mathscr{H}_1} = \alpha(h_1,h_2) + i \beta(h_1, h_2)
\end{equation}
and $\kappa(\H) + i \kappa( \H)$ is dense in $\mathscr{H}_1$. We give a discussion of this in Appendix~\ref{app:oneparticle}.

  While we would like to work exclusively from the boundary theory point of view, we make some comments about how this works after assuming the standard AdS/CFT duality \cite{Maldacena:1997re,Witten:1998qj,Banks:1998dd,hamilton2006holographic}.
  By our general expectations of AdS/CFT, there will be a dual spacetime that is asymptotically AdS,
  with a bulk field $\Phi_\alpha$ associated to each single trace field $\phi_\alpha$, and using the covariant phase space formalism \cite{Crnkovic:1986ex,Lee:1990nz,Iyer:1994ys} the symplectic vector space $\tilde{H}$ is defined by the space of linearized
  solutions to the equations of motion for $\Phi_\alpha$ about the background geometry (these form the tangent space of the phase space for the fully non-linear theory), with boundary conditions at the AdS boundary and modulo linearized diffeomorphisms and gauge transformations that vanish at the boundary. 
  The boundary source $\int h^\alpha \phi_\alpha$ defines a solution:
  \begin{equation}
  h \rightarrow \Phi_\alpha(Y) = \int_{S_{d-1} \times \mathbb{R} } d\mu(x) h^\beta(x) G^{R-A}_{\alpha\beta}(x,Y) 
  \end{equation}
  to the bulk equations of motion where $R-A$ refers to the retarded minus advanced bulk to boundary Green's function. 
  The subtraction guarantees that the boundary conditions are such that the growing (source) term in $\Phi_\alpha$, near the AdS boundary, vanishes. Hence $\Phi_\alpha$
  is a solution to the bulk equations of motion with no boundary sources. 
  The symplectic form:
  \begin{equation}
  \label{stdsf}
 \beta(h_1,h_2) =  \int_\Sigma j( \Phi_\alpha^1, \Phi_\alpha^2 )
  \end{equation}
 is the standard symplectic flux associated to the quadratic fluctuations of the fields around the background spacetime and associated
 to some local gravitational effective theory. 
 In this expression $\Sigma$ is a bulk (AdS) Cauchy slice that we may anchor at any boundary Cauchy slice on the conformal boundary of AdS . If we pick the boundary anchor to lie at $t = t_c$ after the compactly supported
  sources are vanishing, \eqref{stdsf} only sees the retarded solution and we may replace $G^{R-A} \rightarrow G^R $. We then deform $\Sigma$ to the past so that it approaches the AdS boundary with  $t< t_c$.
  In this case the symplectic flux (which is now rotated so that the radial direction is ``time'') pairs near boundary sources $h$ with responses, and since we only have the retarded function there is the standard retarded response:
 \begin{equation}
  \beta( h_1, h_2) = \int_{S_{d-1} \times \mathbb{R} } d\mu(y) \theta(t_c - t_y)    h_1^\alpha(y) \int_{S_{d-1} \times \mathbb{R} } d\mu(x)  h_2^\beta(x)  i \left< \phi_\alpha \phi_\beta  \right>_{R,\sigma} - (1\leftrightarrow 2)
  \end{equation}
But since the sources are localized before $t_c$ we may drop the step function, then the two terms together give \eqref{eq:beta}.    
While the space of solutions $\Phi^\alpha \in \tilde{H}$ for given $h$ does not give all solutions, we expect it gives a dense subspace $H \subset \tilde{H}$ which will be sufficient for our purposes. 
In particular when we pick a state with some symmetric form $\alpha(h_1,h_2) $ and work out the completion $\overline{H}$, we will indeed encounter more general solutions. It must in fact include the symplectic space of linearized solutions for the bulk theory $\tilde{H}$, at least outside any bulk horizons.
That is, the single particle Hilbert space $\mathscr{H}_1$ discussed around \eqref{stdsub} is cognizant of the bulk wavefunctions.  This will be true also for sub-regions $O$ where the resulting one particle
structure will be cognizant of bulk wavefunctions in the causal wedge for $O$. In this case the modes are entangled with modes behind the causal horizon as discussed in Appendix~\ref{app:oneparticle}.

 \begin{figure}[h!]
\centering
\includegraphics[scale=.3]{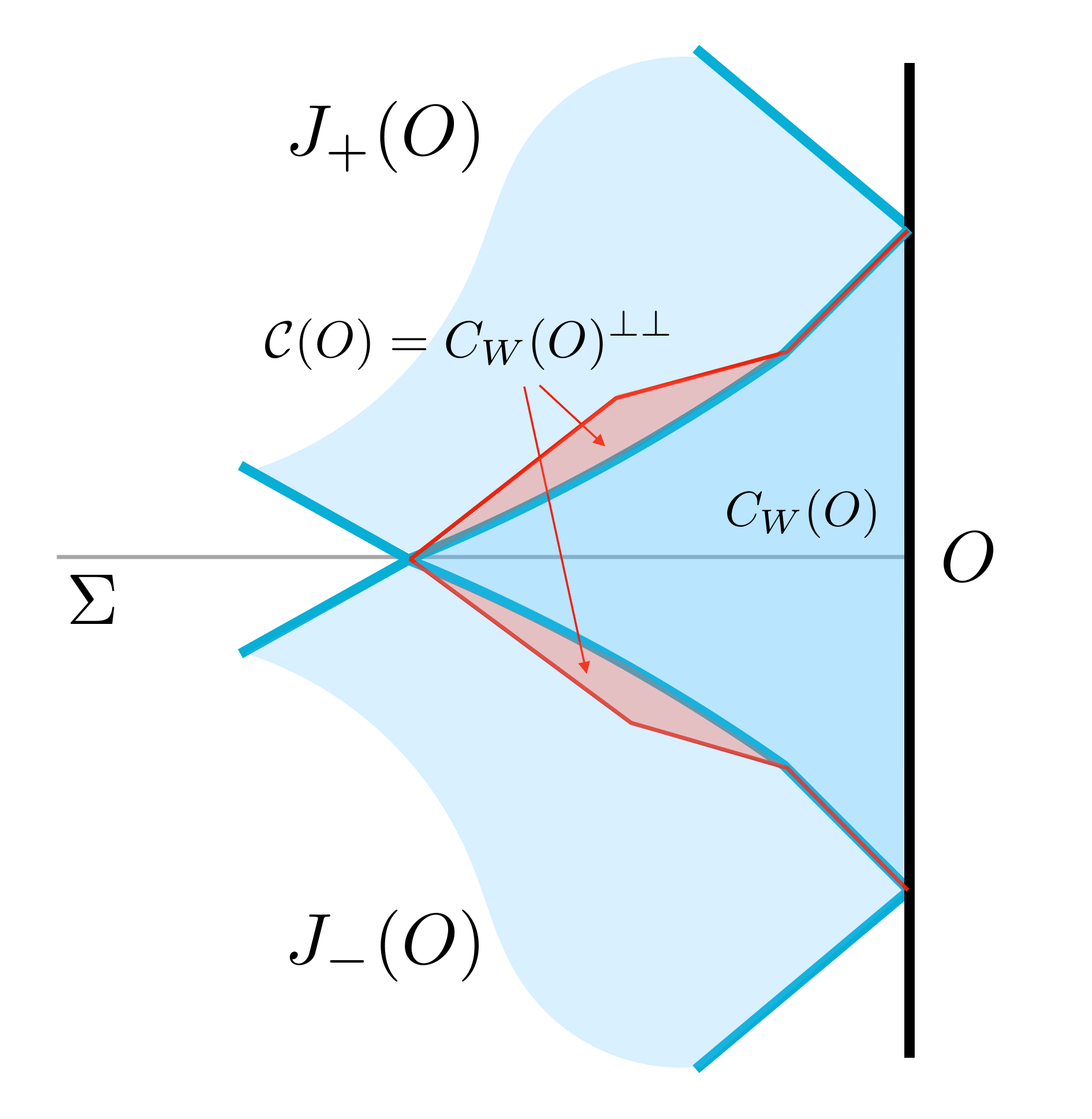} \hspace{2cm} \includegraphics[scale=.3]{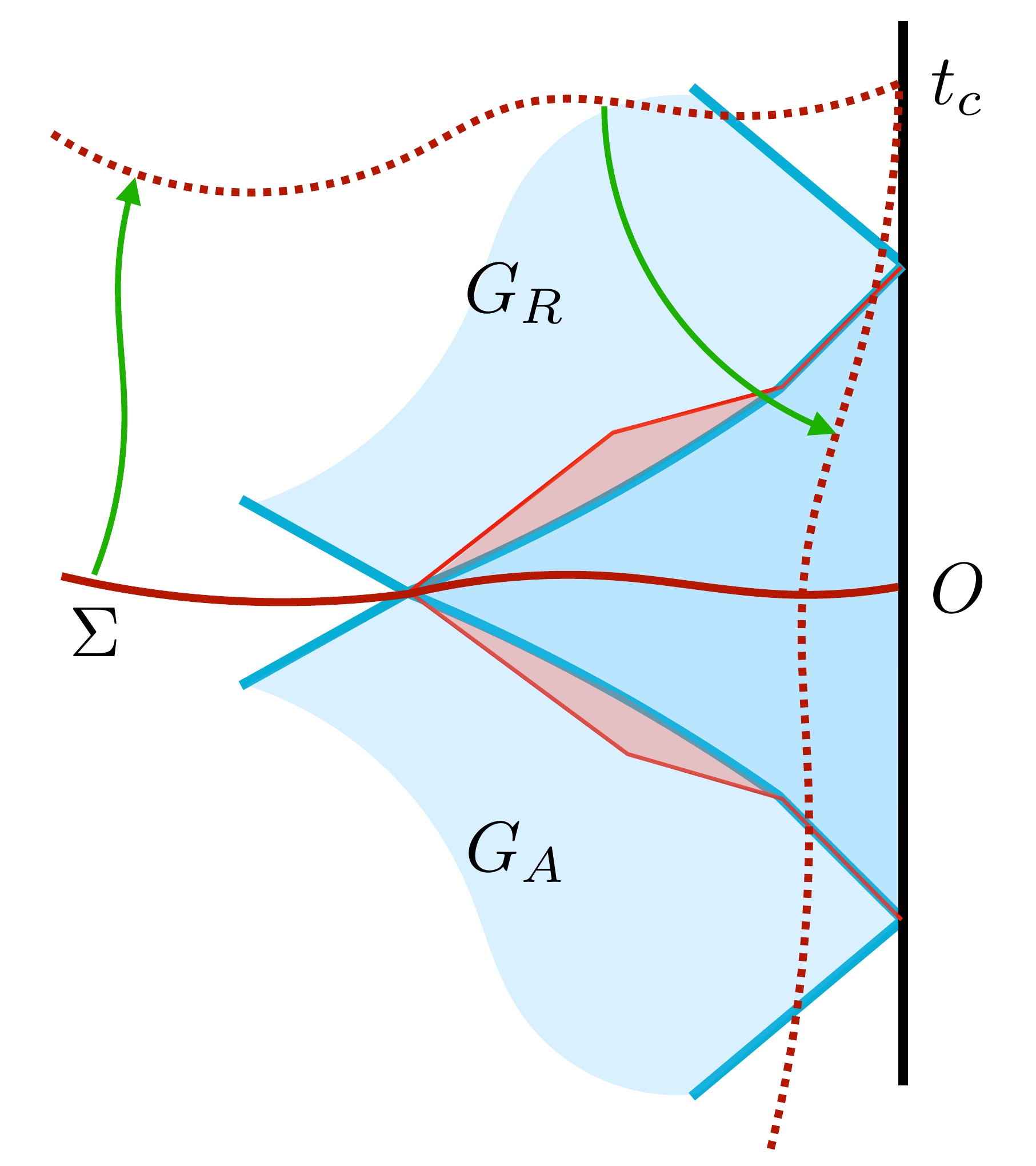}
\caption{ The case of AdS/CFT. (\emph{left}) The causal wedge is the spacetime region $J^+(O) \cap J^-(O)$. The algebra of fields will be associated to the causal completion of this regions, dubbed
the causal domain. This is a slightly larger region as discussed in Appendix~\ref{app:causal}. The picture here is a cartoon showing a cross section of a $D>2$ spacetime through the causal wedge. (\emph{right}) A cartoon of the symplectic flux deformation argument is given in the main text.
The bulk symplectic form is evaluated on a time slice $\Sigma$, but can be deformed to the boundary where it becomes a product of sources and response.}
\end{figure}

There is a natural commutative product defined on the Weyl operators which is simply:
\begin{equation}
\label{timesproduct}
w(h) \times w(g) = w(g + h)
\end{equation}
We also have $w(h)^\star = w(-h)$. After allowing
linear sums of the linearly independent Weyl unitaries, we get an abelian algebra denoted $\mathcal{W}_\times$. 
Such an abelian algebra is state independent and will be the same for different large-$N$ sectors of the theory. This should allow us to compare different sectors, a subject we do not go into great detail on in this paper. 

Given some non-degenerate symplectic form $\beta(\cdot,\cdot)$ we can alternatively turn $\mathcal{W}_u$ into a net
of (non-abelian) unital C$^\star$ algebras \cite{Slawny:1972iq,petz1990algebra}.
We give an associative algebraic structure to the Weyl operators $\mathcal{W}_u$
\begin{equation}
\label{weylalg}
 w^1 w^2 =  e^{i \beta(h_1,h_2)} w^1 \times w^2\, \qquad  w_i = w(h_i)
\end{equation}
and extend this linearly to the full set $\mathcal{W}_\times$. We set $w(0) = 1$ the identity.
 Since the product is different we will denote the resulting $^\star$-algebra as $\mathcal{W}$. 
 There is a unique $C^\star$ norm defined on such an algebra, as we review in Appendix~\ref{app:weyl}. 
The norm completion is then
the  unital $C^\star$ algebra $\overline{\mathcal{W}}^{\| \cdot \|}$. We will continue to denote this algebra as $\mathcal{W}$ for simplicity. 
We will argue later that $\mathcal{W}$ inherits the net structure  over $K$ of the vector spaces $\H(O)$ - this includes the isotony and additivity properties. Thus
$\mathcal{W}$ is an additive net of $C^\star$ algebras, see Definition~\ref{def:C}

Given such a $C^\star$ algebra we can also define a state as a positive linear functional $\sigma : \mathcal{W} \rightarrow \mathbb{C}$ which is positive
on positive elements of the $C^\star$ algebra and such that $\sigma(1) = 1$. 
Given a positive symmetric form $\alpha$, as in \eqref{acgreen}, there is a unique gaussian state (dubbed quasi-free in the literature):
\begin{equation}
\label{assumpo}
\sigma(w(h)) = \exp( - \frac{1}{2}\alpha(h,h))
\end{equation}
And given a state $\sigma$ and $C^\star$ algebra $\mathcal{W}$ we can form the GNS representation $\pi$ and define
the net of von Neumann algebras:
\begin{equation}
\mathcal{C} = \pi_\sigma(\mathcal{W})''
\end{equation}

Equivalently we can work on the Fock space $\mathscr{H}_F = e^{\mathscr{H}_1}$ for the single particle Hilbert space discussed above. We represent:
\begin{equation}
\pi_F( w(h)) = \exp( i (  a( \kappa(h) )   + a(\kappa(h) )^\dagger ))
\end{equation}
where the creation and annihilation operator satisfy the usual relations $[ a(\psi'), a(\psi)^\dagger ] =  \left( \psi', \psi \right)_{\mathscr{H}_1}$
such that:
\begin{equation}
\left< 0 \right| \pi_F( w(h)) \left| 0 \right> = \exp(- \frac{1}{2} \alpha(h,h) )
\end{equation}
where $\left| 0 \right>$ is the vacuum of the Fock space. 
The code subspace will be this Fock space. 
It is guaranteed that $\pi_\sigma$ is equivalent to the Fock representation, so we can work with either. 

\subsection{Superselection sectors}

Let us take a step back and consider the possibility of the appearance of superselection sectors. Since it is quite natural to work with $\mathcal{W}$ a $C^\star$ algebra there are potentially different folia of states related to disjoint representations
of the $\mathcal{W}$ algebra. These would be distinct from the Fock representation. For the local $\mathcal{W}(O)$ it is expected that some local short distance condition on the state, such as the Hadamard condition, is enough to fix a unique folia \cite{verch1994local,Hollands:1999fc}. Hence this is a question for the global algebra $\mathcal{W}$. Adhering to the local QFT philosophy of algebraic QFT, one usually constructs such a global algebra from the local ones. In this case it is called the quasi-local algebra and can be constructed mathematically as an inductive limit of the local algebras \cite{haag2012local,Pontello:2020csg}. 
Such a construction is only possible for directed nets $K$, which is not the case for the Lorentzian cylinder. Nevertheless the algebra we have defined, $\mathcal{W}$, can still be considered as a stand-in for the quasi-local algebra with the property that any consistent (under restriction $O_1 \subset O_2$) family of representations on $\mathcal{W}(O)$  can be lifted to a unique representation on $\mathcal{W}$. See \cite{fredenhagen1990generalizations} for more details on this generalization of the quasi-local algebra that is important for developing the theory of superselection sectors.

 The different representations would give different von Neumann algebras via the double commutant operation, potentially of different types. 
The folia of a representation $\pi$ is defined as the collection of states given by density matrices on the Hilbert space of the representation \cite{haag2012local}. 
Such states are also referred to as being $\pi$-normal \cite{bratteli2012operator} since they are normal states for the von Neumann algebra $\pi(\mathcal{W})''$. 
Two representations are disjoint if there is no intertwiner that relates them. Rather a folia is determined by the quasi-equivalence 
relations on representations: two representations $\pi_i$ are quasi-equivalent if their induced von Neumann algebras are naturally isomorphic (the isomorphism must be consistent on $\pi_i(\mathcal{W})$). This equivalence allows for such things as tensoring in a reference or projecting on the commutant (without projecting on the center should one exist.)

The main representation of interest $\pi_\sigma$, associated to $\alpha$, corresponds to the folium called the Fock representation. 
Physically we should imagine this is generated by the single
particle Hilbert space of quantum fields in the bulk spacetime. Other disjoint representations would have an infinite number of excitations from this Fock space. These would be excluded by holographic bounds and furthermore would require completely different large-$N$ codes to describe, with different commutation relations $\beta$. As such in the large-$N$ limit a different kind of sector physics will arise compared to the standard theory of superselection sectors. We will return to the question of such large $N$ sector physics in the Discussion section. 
For now we simply note that we stick to the Fock representation.
 
 The one caveat being that $\beta$ may be degenerate and this actually does lead to a primitive form of superselection sectors associated to the possible existence of a center. 
 We discuss this possibility next.

\subsection{Stress tensor and currents}
\label{sec:deg}

Some important examples of single trace fields
are the stress tensor and conserved currents:
\begin{align}
h^\alpha \rightarrow h^{\mu\nu} \, \qquad \phi_\alpha^{(N)} \rightarrow \tau_{\mu\nu}^{(N)} \equiv N_T^{-1/2}  T_{\mu \nu}^{(N)} 
\\
h^\alpha \rightarrow a^\mu \, \qquad \phi_{\alpha}^{(N)} \rightarrow j_\mu^{(N)} \equiv N_J^{-1/2} J_\mu^{(N)} 
\end{align}
where $N_T, N_J \propto N^\#$ are the coefficients of the vacuum stress tensor correlation function that is fixed by conformal invariance up to this number. 
The leading $N$ dependence is assumed to be a power of $N$, where $\# > 0$.  For simplicity, we will assume that this scaling with $N$ also determines the leading $N$ dependence of the connected thermal correlation function of the stress tensor and currents for fixed temperature. 
This is the case in many models of interest.  Since $T_{\mu\nu}^{(N)}$ is symmetric traceless and conserved we may regard $h^{\mu\nu}$ as the space of smooth symmetric functions
modulo total derivatives $\partial^{(\mu} b^{\nu)}$ and shifts by the cylinder metric $b \eta^{\mu \nu}$
for functions $b^{\mu},b$.  That is $h^{\mu\nu}$ represents such an equivalence class. The net structure is
then defined so that $h^{\mu\nu} \in \H(O)$ if one of its representatives is localized in $O$. A similar discussion applies to the currents and sources $a^\mu \cong a^\mu + \partial^\mu \lambda$.

The stress tensor and current normalizations are independently fixed by the Ward identity, and so we need to divide by $N_{T,J}^{-1/2}$ (after subtracting
the one point function)  in order to have a large-$N$ limit where the two point functions of the single trace fields are finite.
Typically this will suffice for the local stress tensor and current operators and the resulting large-$N$ limit will be dual to the metric
fluctuations and bulk gauge fields in a putative gravitational dual.

 However the generators of the conformal symmetries that are constructed using $T_{\mu\nu},J_\mu$ can have different large-$N$ limits since
 their fluctuations could be quite different to those of the local fluctuations in $\tau_{\mu\nu}, j_\mu$. 
For example the conformal generators:
\begin{equation}
H^{(N)}_\xi = \int_{S^{d-1}} d \Sigma^\nu \xi^\mu T_{\mu \nu}^{(N)}
\end{equation}
where $d\Sigma^\nu$ is the volume form for the $\Sigma = S^{d-1}$ Cauchy slice of the cylinder and $\xi$ is some conformal Killing vector
of the cylinder. The Hamiltonian arises from $\xi = \partial_t$. Since these annihilate the vacuum,  an asymptotic code based on the vacuum 
would be expected to contain the corresponding large-$N$ operator $H_\xi$ despite the predicted large-$N$ counting based on $\tau_{\mu\nu}$. 
In particular we will see that the origin of $H_\xi$ is not that of:
\begin{equation}
U^{(N)}_{\xi} \equiv \int_{S^{d-1}} d\Sigma^\mu \xi^\nu \tau_{\mu\nu}^{(N)} = N_T^{-1/2} H^{(N)}_\xi
\end{equation}
which is frozen in the large-$N$ limit in vacuum due to the pre-factor. Rather $H_\xi$ can be understood as arising from quadratic combinations of the other large-$N$ fields. 
In AdS/CFT these composites are simply the integrated bulk stress tensor, at least in the vacuum setting. These composites are constructed directly from the symmetries
of the bulk operator algebras, and the assumption that these symmetries are unitarily represented on the bulk Hilbert space.\footnote{In fact the conformal symmetries can be reconstructed from the modular operators of the bulk, under a seemingly unrelated assumption that the causal and entanglement wedges are equal for $O$ any double cone boundary region. See the Discussion section. } Hence the bulk theory will typically reorganize how such symmetry generators arise and we must be careful when we treat these operators in the large-$N$ limit. 

A signature of these difficulties can be found in the previously constructed Weyl algebra as a degeneracy of the symplectic form $\beta$ \eqref{eq:beta}.
These arise from conserved charges in the stress tensor and the currents that are (approximately) preserved by the defining state. 
Let us focus on $\xi = \partial_t$ which is sufficient to illustrate the main point. We will return to the general case at the end of this section. 
These degeneracies potentially lead to a center for the Weyl algebra that
we construct the above due to the limit
of commutation relations:
\begin{equation}
 \left[U^{(N)} , \phi^{(N)} \right] =  i N_T^{-1/2} \partial_t \phi^{(N)} 
\end{equation}
Thus so long as the background is time translation invariant (to leading order) $\partial_t \left< \phi^{(N)}  \right>_{\sigma_N} = \mathcal{O}(1)$, then after using this to do one point subtraction we conclude that
$ \left[U^{(N)} , \phi^{(N)} \right]  \rightarrow 0$ as an operator in the large-$N$ limit \cite{Witten:2021unn,Chandrasekaran:2022eqq}.
Or in other words, there is a degeneracy for the pre-symplectic form: 
\begin{equation}
\beta( h^{\mu\nu} =  \xi^{(\nu} \delta_\Sigma^{\mu)} , h_2 ) = 0
\end{equation}
for all $h_2 \in \H$, where $\delta_{\Sigma}^\nu$ is a $d-1$ form delta function on $\Sigma$. If the background were not time translation invariant then large-$N$ counting would give $\partial_t \left< \phi^{(N)}  \right>_{\sigma_N} = \mathcal{O}(N)$ and there would be no degeneracy and no center. We assume this is not the case. 

We can work out the limiting central operator:
\begin{equation}
\label{defU}
 \exp( i s U )
= \exp( i s \int_{S^{d-1}} d\Sigma^\mu (\partial_t)^\nu \tau_{\mu\nu})
\end{equation}
for all $s \in \mathbb{R}$ and these will commute with all $w(h)$. Thus it will appear as a center for the von Neumann algebra $\mathcal{C}$.  
In fact the correct mathematical framework to deal with this center is that of superselection sectors. The set of allowed superselection sectors are determined by
all possible limiting states on $\mathcal{W}$ and in particular the fluctuations of $U$ in these states. That is the possible superselection sectors will be determined dynamically.\footnote{We dismissed superselection sectors previously, however the existence of a center necessitates a reassessment. The sectors will however \emph{only} arise from the center with the previous discussion applying to the non-central part.} 

There are two natural possibilities that we consider:
\begin{itemize}
\item The central variable $U$ would be fixed to some constant $U=g$ in the state of interest:
\begin{equation}
\sigma^g( e^{i sU }) = \int d U  \mu_g(U) e^{ i s U} \qquad \mu_g(U) = \delta( U- g)
\label{deltamu}
\end{equation}
This would happen in the vacuum sector and other low energy sectors. It will also arise in so called fixed area states \cite{Dong:2018seb,Akers:2018fow} that we discuss below. 
In these case $e^{i s  U}$ is naturally removed in the GNS representation where it takes a fixed value. 
\item The central variable $U$ is determined by some function that then describes the energy fluctuations on scales of order $\mathcal{O}(N_T^{1/2})$, which then
must be Gaussian:
\begin{equation}
\label{gaussianmu}
\sigma( e^{i sU }) = \int d U \mu(U) e^{ i s U}\,,
\qquad \mu(U) = (c/\pi)^{1/2} \exp( - c U^2)
\end{equation}
This happens for the canonical ensemble for temperatures $T > T_{c}$ above the CFT deconfinement transition, dual to Hawking-Page transition in holographic theories, where we assume $T_{c} = \mathcal{O}(1)$. Thus the center will be excited and $\mathcal{Z}(\mathcal{C})
= \{ e^{ i \mathbf{s} U} \}_{\mathbf{s} \in \mathbb{R}}$ will be non-trivial.
  \end{itemize}
  
 In fact, as we discuss in the Appendix~\ref{app:center}, there is a natural class of states  $\sigma_\mu$ determined by any finite measure $\mu$ on $\mathbb{R}$. 
The abelian von Neumann algebra associate to the GNS representation $\pi_{\sigma_\mu}$ is given by $L^\infty(\mathbb{R}, \mu)$. States on this folium correspond to other measures $\nu$ on $\mathbb{R}$ that are absolutely continuous with respect to $\mu$: $\nu \ll \mu$.

We can work out which central state to use by taking the relevant limit.
For a thermal state:
\begin{align}
\sigma_N^\beta( e^{i s (X^{(N)}- N_T^{-1/2} \bar{E}_N )}) &= 
 \int dE \rho(E) \exp( - \beta E + i s N_T^{-1/2} (E-\bar{E}_N))
 \end{align}
 where $\bar{E}^\beta_N$ is the average energy of the canonical ensemble at temperature $1/\beta$
 and for $T > T_c$ we can approximate the density of states via $\rho(E) \propto e^{S_N(E)}$
 with $S_N(E) \sim  N_T^{\frac{1}{d}} E^{\frac{d-1}{d}}$.
Taking the limit we find \eqref{gaussianmu} with $c = \mathcal{O}(T^{-(d+1)})$ determined by the limiting specific heat. 
The resulting center represents the fluctuations of the energy $\delta E \sim N_T^{1/2}$ in the thermal state. 
 These fluctuations are large enough that a bulk energy operator
that can resolve sub-leading $\mathcal{O}(1)$ energy differences is not present here. 
Such a large center was predicted to arise in holographic codes already back in \cite{Harlow:2016vwg,Dong:2018seb,Akers:2018fow} where $U$ is thought of as the (re-scaled) area operator in the bulk. Although we point out that such an area operator will only be relevant for boundary regions that realize type-I algebras at fixed $N$ in the microscopic theory. Otherwise we do not expect an exact center to arise even in the limit of large $N$.

In the thermal setting we can also discuss a different family of representations, often called fixed area states where again we find a $\delta$ function measure. This will give an example where a different folium of states
arises for the same Weyl algebra. Consider a generalization of the canonical thermal state:
\begin{equation}
\label{fixedg}
\sigma_N^g(\cdot) = (Z_N^f)^{-1} {\rm Tr}_{\mathscr{K}_N} e^{ - \beta H^{(N)} }  f( N_T^{-1/4}(H^{(N)} - E_g) ) (\cdot)
\end{equation}
where $f(x)$ is a fixed Gaussian peaked around $0$ with standard deviation of $\mathcal{O}(1)$. We pick $E_g = \bar{E}^{(N)}_\beta + g N_T^{1/2}$ where
$ \bar{E}^{(N)}_\beta$ is the average energy of the canonical thermal state at fixed temperature and $g$ is an offset that labels some superselection sector. 
The function $f$ modulates the thermal ensemble on the scale of $\delta E \sim N_T^{1/4}$ with average offset by $g N_T^{1/2}$. Any scale $\delta E \sim N_T^{\alpha}$
with $0 < \alpha < 1/2$ would have been sufficient.\footnote{A different choice was used in \cite{Chandrasekaran:2022eqq} since they wanted to study $\mathcal{O}(1)$ energy fluctuations that could then be probed by an energy operator. The inclusion of this energy operator in the bulk algebra gave rise to a crossed product algebra. We do not include such an $\mathcal{O}(1)$ energy operator in our analysis so the crossed product does not arise. This is a perfectly valid choice for a code that is blind to such energy fluctuations. See \cite{toappear} where we show how to include such an energy operator in the context of our asymptotically isometric codes.} Since $U$ resolves energies on the scale of $N_T^{1/2}$ the gaussian limits to a delta function
so the limiting state has measure $ \mu_g(U) = \delta( U- g)$. Since the $\delta$ function measure is not absolutely continuous with respect to the Gaussian measure, these
states are in a different folium for $\mathcal{W}$ as compared to the limit of the canonical ensemble.

These states have the interpretation as a fixed area state with fluctuations that are sub leading to the canonical ensemble, but the
time uncertainties $\delta t \sim N_T^{-1/4}$ are still small enough that we can trust there could be some emergent notion of a classical bulk dual. The large $N$
limit of these states give rise to the von Neumann algebra: $\mathcal{C}_g = \pi_{\sigma^g}(\mathcal{W})''$ 
acting on $\mathscr{H}_g$ with no center $\mathcal{Z}(\mathcal{C}_g) = \mathbb{C} 1$ (ignoring any other charges). Since shifts in $g$
correspond to sub-leading shifts in the average microcanonical energy - these algebras are all unitarily equivalent for different values of $g$. 
Hence it is more natural to label the algebra $\mathcal{C}_0 = \mathcal{C}_g$ and Hilbert space $\mathscr{H}_0 = \mathscr{H}_g$.  Based on general discussions found in \cite{Leutheusser:2021frk,Leutheusser:2021qhd}, we expect that
$\mathcal{C}_0$ is a type-III$_1$ factor. 
 Since the $\delta E \sim N_T^{1/4} \rightarrow \infty$ the energy fluctuations
are still large, however we do not see these fluctuations on the code since there is simply no energy operator. 
The von Neumann algebra for the canonical ensemble $\mathcal{C}$ and of the fixed area state $\mathcal{C}_0$ 
are related via $\mathcal{C} = \mathcal{C}_0\, \overline{\otimes}\, L^\infty(\mathbb{R})$
with $\mathscr{H} = \mathscr{H}_0\, \overline{\otimes} \, L^2(\mathbb{R})$.

 \subsubsection{General considerations}
 
\label{sec:deg2}
 
Now let us consider the more general case, with more than one symmetry generator.  We first remind the reader of a subtlety we have ignored for now: it may be that $\beta$ is non-degenerate but under norm completion of the vector space $H$
the complete $\beta$ may be degenerate. For this reason we will work with the completed quantities for the rest of this section. 
To complete the vector space we replace $H$ by it's closure induced by the inner product given by $\alpha$.
It is possible that $\alpha$ is already degenerate on the degenerate subspace  of $\beta$, in which case we pick some other $\alpha' \geq \alpha$ where $\alpha'$ is non-degenerate.\footnote{
Positive symmetric forms have an order relation given by $\alpha'(h,h) \geq \alpha(h,h)$ for all $h \in H$. An example where $\alpha$ is degenerate is the vacuum state.
In which case we could pick $\alpha' = \alpha + \epsilon \alpha_\beta$ where $\alpha_\beta$ derives from the anti-commutator Green's function \eqref{acgreen} for a thermal state above the
deconfinement temperature and $\epsilon > 0$ is any number. The specific form of $\alpha'$ will not matter for our construction. }
We use this to complete $H$ as well as $\beta,\alpha$. For the rest of the section we use the same symbol $H$ and $\beta,\alpha$ to denote this completion.

The degenerate directions of $\beta$ are determined by the unbroken symmetry generators about some background. This includes angular momentum charges and any other global charges. The symplectic form $\beta$ 
will have a vector subspace $\H_0 \subset \H$ which is degenerate
with $\beta(h_0,h) = 0$ for $h_0 \in \H_0$ and for all $h \in \H$. 
We expect this will be a finite dimensional vector space $\H_0 \cong \mathbb{R}^{n_Q}$ for the number
of symmetry generators.\footnote{$n_Q$ will be different for vacuum and thermal codes. 
For thermal codes we expect $n_Q = 1 + d(d-1)/2 + {\rm dim} \mathfrak{g}$ for the unbroken energy, $SO(d)$ angular momentum plus the dimension of the Lie algebra $\mathfrak{g}$
for the global symmetry group. And for the vacuum code
we expect $n_Q = (d+2)(d+1)/2 +  {\rm dim}  \mathfrak{g} $ for the unbroken conformal generators plus global symmetry generators.  }

All of these symmetry generators will have large fluctuations in sufficiently high temperature thermal states
and so the same story as the energy fluctuations will apply. 
We can also consider a state $\sigma^g$ that fixes all the potentially central operators with values determined by $g = \{ g_i \}_i$.
In the vacuum state all the charges will automatically be fixed.

We consider the quotient space $\H/\H_0$ where $[h] = [h+ h_0]$ for $h_0 \in H_0$. The symplectic form is well defined on $\H/\H_0$ being non-degenerate there:
$\beta([h_1],[h_2]) \equiv \beta(h_1,h_2)$. The quotient space has a natural net structure with isotony and causality given by $h \in \H/\H_0(O)$ if there
exists some $h \in [h]$ with $h \in \H(O)$. 
We have the short exact sequence:
\begin{equation}
0  \rightarrow \H_0  \mathop{\rightarrow}^{\iota} \H \mathop{\rightarrow}^{[\cdot]} \H/ \H_0 \rightarrow 0
\end{equation}
$\H_0$ is a closed linear subspace of $\H$. Thus we can decompose $\H = \H_0 \oplus \H_0^\perp$ for the closed perpendicular subspace with respect to 
the inner product $( \cdot, \cdot)_{\H}' = \alpha'(\cdot,\cdot)$. Then  define $p_0 : H \rightarrow H_0$ by orthogonal projection onto $H_0$, using  $( \cdot, \cdot)_{\H}'$. This satisfies $p_0 \circ \iota ={\rm Id}_{H_0}$.
Then the splitting lemma implies that $ \H_0 \oplus \H_0^\perp = \H \cong {\rm im}( \iota)  \oplus {\rm ker}( \iota)
= \H_0 \oplus \H/\H_0 $  where the isomorphism $\varphi : \H \rightarrow \H_0 \oplus \H/\H_0$ satisfies 
$\varphi \circ \iota (h_0) = h_0 \oplus 0$ and $[\cdot] \circ \varphi^{-1} (h_0 \oplus h_\perp)= h_\perp$. While this construction depends on our choice $\alpha'$ 
it is important to note that we are not modifying the state of interest here, we are just ensuring the existence of $\varphi$. Since there is
no canonical choice, any other choice will do just as well. Note that $\varphi$ does not preserve the net structure in the
sense that $\varphi^{-1}( 0 \oplus [h]) \notin H(O)$ with $[h] \in \H/\H_0(O)$ and this could lead to difficulties. 

Since the test function space along with the symplectic form $\beta$ has this direct sum structure, the Weyl algebra $\mathcal{W}_0$ will be a direct product \cite{Hollands:2017dov} of the Weyl algebra
based on $\H/\H_0$ and an abelian Weyl algebra based on $H_0$. This is similarly true for
the von Neumann algebra:
\begin{equation}
\label{natis}
\mathcal{C} \cong \mathcal{C}_0 \overline{\otimes} L^\infty(\mathbb{R}^{n_Q})
\end{equation}
represented on $\mathscr{H} \cong \mathscr{H}_0 \overline{\otimes} L^2(\mathbb{R}^{n_Q})$. Again we note that both $\mathcal{C}_0$  and $\mathcal{C}$ have
the net structure over $K$. But this net structure is not compatible, implying for example that the isomorphism $\Phi$ in \eqref{natis} induced by $\varphi$ does
not map $\Phi^{-1}( c \otimes 1) \notin \mathcal{C}(O)$
for $c \in \mathcal{C}_0(O)$.  Or in other words $\Phi(\mathcal{C}) \subset \mathcal{C}_0 \overline{\otimes} L^\infty(\mathbb{R}^{n_Q})$ defines a different net over $K$
as compared to $\mathcal{C}_0 \otimes 1$. 
To see this difference in another way, note that $\bigvee_O \mathcal{C}_0(O) \otimes 1$ for all local $O$
does not generate the entire algebra. Whereas $\bigvee_O \mathcal{C}(O)$ does.

Note that even for the vacuum representation we have to confront this degeneracy issue. 
 In particular since there are no fluctuations of the symmetry generators, they will induce only the fixed charge
sector $\pi_g$ for some $g_i =0$. 
One might hope then that one can exclusively work with $H/H_0$ and correspondingly $\mathcal{W}_0$ and $\mathcal{C}_0$. However 
this is not always possible since we need to map the Weyl operators to
 to the finite $N$ theory causing a problem since the quotient effectively arises only dynamically as $N \rightarrow \infty$. 
In particular this issue will need to be addressed when constructing our encoding maps $V_N$ in the presence of these charges.
Here we need a way to map the quotient space $\H/\H_0$  back to $\H$, since the test functions at finite $N$ live naturally in $\H$ without the quotient. 
To do this we need to specify a unique representative $p : \H/\H_0 \rightarrow \H$ such that
the kernel of this map is trivial. We use $p(\cdot) = \varphi^{-1} ( 0 \oplus \cdot)$. 

Consider the representation $\pi_{\sigma^g}$ based on the GNS
construction for states with fixed charges:
\begin{equation}
\sigma^g( w( h_i) ) = e^{ i h_i g^i }  \, \qquad h_i \in \H_0
\end{equation}
then:
\begin{equation}
\pi_{\sigma^g}( w(h) )  = \pi_{\sigma^g}( w(h + h_0))
\end{equation}
for any  $h_0 \in H_0$. This implies that $\pi_{\sigma^g}$ is not a faithful representation. 
In particular $\pi_{\sigma^g}( w(h) )  = \pi_{\sigma^g}( w \circ p[h]  )$.

\subsection{Formal definitions and assumptions}

Let us attempt to formalize the discussion so far with some definitions. Firstly, associated to the space of test functions $\H$ over the single trace fields, there is commutative $^\star$-algebra generated by unitary Weyl operators $\mathcal{W}_u \subset \mathcal{W}_\times$, with commutative product \eqref{timesproduct},
associated to $ h \in \H$. That is the span of a finite number of elements from $\mathcal{W}_u$
\begin{equation}
\mathcal{W}_\times = {\rm Span} (w(h), h \in \H)
\end{equation}
which are linearly independent for distinct $h$. The Weyl unitaries maintain the net structure with $\mathcal{W}_\times(O)$ generated by unitaries $\mathcal{W}_u(O)$ associated to test
functions $\H(O) \subset \H$. Based on this we give the following definition:

\begin{definition}[Single trace algebra]
\label{def:sta}
Consider a sequence of CFTs described by the net $\mathcal{M}_N$. 
We say that this sequence has a single trace algebra if the pair $(\H,\mu_N)$ exists,
where $\H : K \supset O \rightarrow \H(O) $ is an isotonic net of real vector spaces valued in the single trace test functions on the Lorentzian cylinder $K$,  \emph{compactly} supported inside the the causally complete sub-regions $O$ (possibly equal to $K$), and $\mu_N$
is a sequence of unitary preserving $^\star$-maps on the associated set of Weyl unitaries $\mu_N: \mathcal{W}_u(O)
\rightarrow \mathcal{M}_N(O)$ such that $ \mu_N(w(h)^\star)  = \mu_N(w(h))^\star$ for all $h \in \H$.
\end{definition}

Given the above structure we set:
\begin{equation}
w_N(h) \equiv \mu_N( w(h))  \qquad \, \qquad \forall \, h \in \H
 \end{equation}
 where $w_N$ was used in the previous subsection.\footnote{While it would be consistent with the definition that these maps are trivial (with $\mu_N(w) = 1$) this case will be excluded once we demand the existence of a single trace sector associated to a sequence of states.} 
 The map $\mu_N$ obviously does not preserve the commutative algebra of $\mathcal{W}_\times$ but since it maps unitaries to unitaries we expect it to be somewhat well behaved. 
We might be tempted to turn these maps into linear maps via linear extension. However we will need to subtract the one point functions before we do this, so this
will have to wait. Instead:

\begin{definition}[Large-$N$ sector] 
\label{def:largeN}
Consider a sequence of CFTs with a single trace algebra as in Definition~\ref{def:sta}. We say that this sequence has a large-$N$ sector for a sequence of normal states $\sigma_N
\in (\mathcal{M}_N)_\star$, if there is a limiting (possibly degenerate) symplectic form $\beta(\cdot, \cdot) : \H \times  \H  \rightarrow \mathbb{R}$ and a limiting positive symmetric form
$\alpha(\cdot, \cdot) : \H \times  \H  \rightarrow \mathbb{R}$
such that:
\begin{itemize}
\item[(a)]  The following limit applies: 
\begin{equation}
\lim_{N \rightarrow \infty} \frac{\sigma_N (w^1_N w^2_N  \ldots w^k_N )}{\sigma_N((w^1\times w^2 \times \ldots w^k)_N)}
= \prod_{i < j} \exp( i \beta(h^i,h^j) ) \qquad w^i = w(h^i)
\end{equation}
for all finite subsets of Weyl unitaries $\forall \{ w^1,w^2 \ldots w^k \} \in \mathcal{W}_u$ with $k \geq 2$.
Here we have set $w_N^i = \mu_N(w^i)$ 
and $(w^1\times w^2 \ldots w^k)_N = \mu_N( w^1 \times w^2 \ldots w^k)$.

\item[(b)] The following limit applies: 
\begin{equation}
\lim_{N \rightarrow \infty} \frac{\sigma_N( w^1_N) \sigma_N(w^2_N)}{ \sigma_N( (w^1 \times w^2)_N)} 
= \exp( \alpha(h^1,h^2)) \qquad  \forall w_{1,2} \in \mathcal{W}_u
\end{equation}
\item[(c)]  Upon restricting $\alpha,\beta$ to $H(O) \subset H$ for a local region $O$ then $(H(O),\alpha|_O,\beta|_O)$ forms a factorial standard subspace. See \cite{Longo:2021rag} and a summary in Appendix~\ref{app:oneparticle}. 
\end{itemize}
\end{definition}

It is important we only demand (a) for a fixed number $k$ of the Weyl operators that is independent of $N$, which is a standard consideration when taking the large-$N$ limit.
While we allow for degenerate $\beta$ for future considerations, for now we will continue to assume that $\beta$ (and its completion) is non-degenerate. 

As a corollary to these definitions, if the limits apply, we can endow the linear span $\mathcal{W}$ with an algebraic structure as in \eqref{weylalg}.
This $C^\star$-algebraic structure inherits the causality requirement of the net structure of $\mathcal{M}_N(O)$ and the Weyl unitaries $\mathcal{W}_u(O)$ since
when $w_1 \in \mathcal{W}_u(O_1), w_2 \in \mathcal{W}_u(O_2)$ for $O_1 \subset O_2'$, then:
\begin{equation}
\beta(h_1,h_2)  = \beta(h_2,h_1)\,\, \implies \,\, \beta(h_1,h_2) = 0
\end{equation}
implying that $[w_1,w_2] =0$. Isotony is also inherited from isotony of $\mathcal{W}_u$. Finally additivity arises since we can always decompose a smooth
function supported in the compact $O$ into a sum of smooth functions supported in $O_i$ with $i \in I$. Hence $w(h_O) = \times_{i \in I} w(h_{O_i})$ which can be turned
into the $C^\star$ algebra product over local Weyl elements $w(h_{O_i})$.  Hence we take $\mathcal{W}$ to be an additive net of $C^\star$-algebras
Definition~\ref{def:C}.

Similarly, with such a single trace sector, we can prove the compatibility bound \eqref{bdab} holds without further assumption. Note that:
\begin{equation}
\left[ \frac{ \sigma_N( (w^i_N)^\star  w^j_N)}{ \sigma_N(w_N^i)^\star  \sigma_N(w_N^j) }\right]_{ij} > 0
\end{equation}
for any finite set of $h^i$ due to positivity of the state $\sigma_N$. This
implies that (after taking the limit):
\begin{equation}
\left[  \exp(  i \beta(h^i, h^j) + \alpha(h^i, h^j)  ) \right]_{ij} \geq 0
\end{equation}
If we take a further scaling limit on $h^i \rightarrow h^i \tau$, with $\tau \rightarrow 0$, we can expand the exponential. The leading matrix then degenerates with all entries $1$
- such a matrix has one non-zero eigenvalue. 
Positivity constrains the leading correction in $\tau$  on the remaining degenerate subspace:
\begin{equation}
\tau^2 \sum_{ij} (p_i)^\star p_j ( i \beta(h^i, h^j) + \alpha(h^i, h^j) ) \geq 0
\end{equation}
if $\sum p_i = 0$. Since $h^i$ were arbitrary, we may drop this later constraint on $p_i$.
This implies that $ i \beta(h^i,h^j) +\alpha(h^i,h^j)$ is a positive self adjoint matrix. 
Cauchy-Schwarz  $| M_{ij} |^2 \leq M_{ii} M_{jj}$ for any two fixed $i,j$. Thus:
\begin{equation}
 \beta(h^1,h^2)^2 \leq | i \beta(h^1,h^2) + \alpha(h^1,h^2) |^2\leq
\alpha(h^1,h^1) \alpha(h^2,h^2) 
\end{equation}
Thus the positivity statement on \eqref{bdab} is automatically satisfied
and $\sigma$ can be extended linearly from $\mathcal{W}_u$ to the entire $\mathcal{W}$ where it becomes a quasi-free state.

Given a large-$N$ sector we can now define:
 \begin{equation}
 \gamma_N( w_i) = \frac{\mu_N(w_i) |  \sigma_N((w_i)_N)| }{ \sigma_N((w_i)_N)}\, \qquad w_i \in \mathcal{W}_u\, \qquad \gamma_N(1) = 1
 \end{equation}
 and extend this linearly to the subalgebra of $\mathcal{W}$ consisting of finite sums of Weyl unitaries. 
 Note that $\gamma_N(w) = \gamma_N(w)^\star$. This definition makes sense because the Weyl unitaries are linearly independent. 
 There is a Banach $\star$-algebra 
  $\mathcal{W}_B \subset \mathcal{W}$ given by infinite sums $w = \sum_i c_i w_i$
for unique $w_i \in \mathcal{W}_u$ and with $ \| w \| \leq \| w\|_B \equiv \sum_i | c_i | < \infty$. This is a Banach algebra
with respect to the norm $ \| w\|_B$ - in particular it is complete with respect to this norm. However this norm will not satisfy the $C^\star$ identity
and so is different from $\| \cdot \|$. In fact $\mathcal{W}_B$ is
$\| \cdot \|$-norm dense in $\mathcal{W}$, that is $\overline{\mathcal{W}_B}^{\| \cdot \|} = \mathcal{W}$, and the unitarity
of $\mu_N(w_i)$ for $w_i \in \mathcal{W}_u$ implies that we can extend $\gamma_N$ to maps on $\mathcal{W}_B$ with the estimate:
\begin{equation}
\| \gamma_N(w) \| \leq \| w \|_B
\end{equation}
which is notably independent of $N$. See Appendix~\ref{app:weyl}. We say that $\gamma_N$ is uniformly in $N$ pointwise
bounded on  $\mathcal{W}_B$. 

Now $\mathcal{W}_B$ has a net structure that derives from $\mathcal{W}_u$ with $w \in \mathcal{W}(O)$ if each Weyl unitary in the decomposition of $w$ satisfies $w_i \in \mathcal{W}_u(O_i)$ with $O_i \subset O$. This is also preserved under completion $\overline{\mathcal{W}_B(O)}^{\| \cdot \|} = \mathcal{W}(O)$. 
It is clear that these linear $\star$ maps preserve the net structure in maping $\mathcal{W}_B(O) \rightarrow \mathcal{M}_N(O)$. 
Furthermore  we can use Definition~\ref{def:largeN} to show that:
 \begin{equation}
 \sigma_N \circ \gamma_N(w) \rightarrow \sigma(w) \qquad \forall w \in \mathcal{W}_B
 \end{equation}
 for the state $\sigma(w(h))$ defined in (\ref{assumpo}), where the linear extension of both of these is consistent. 
 Similarly the limit on $\gamma_N$ extends to finite products:
  \begin{equation}
  \label{ash}
 \sigma_N ( \gamma_N(w_1) \gamma_N(w_2)  \ldots \gamma_N(w_k) ) \rightarrow \sigma(w_1 w_2 \ldots w_k) \qquad \forall w_1, w_2, \ldots w_k \in \mathcal{W}_B
 \end{equation}

Let us discuss the kernel of $\gamma_N$. We expect this is non-empty. That is the linear subspace:
\begin{equation}
\mathcal{Z}_N = \{ a \in \mathcal{W}_B : \gamma_N(a) = 0 \}
\end{equation}
will be non-trivial.
Note that this is not an ideal - if $\gamma_N$ was a quantum channel, it would be an ideal.  We expect this kernel is an important aspect
of the non-isometric nature of this code.  We will see that it determines the null states of the bulk - arising from 
different Weyl operators that map to the same operator at fixed $N$.

We now discuss how to use the above structure to define the code subspace. 
Define the GNS Hilbert space $\mathscr{H}$ and representation $\pi_\sigma$ for the net with representative vector
denoted $\left| [a] \right>$ for $a \in \mathcal{W}$ such that $\pi_\sigma(b) \left| [a] \right> = \left| [b a] \right>$.
Denote $\left| [1] \right> \equiv \left| \eta \right>$.
The resulting net of von Neumann algebras is:
\begin{equation}
\label{cw}
\mathcal{C}(O) = (\pi_\sigma(\mathcal{W}(O)))''
\end{equation}
Note that $\pi_\sigma$ restricts consistently to local algebras $\mathcal{W}(O)$ and also that we can replace $\mathcal{W}(O) \rightarrow \mathcal{W}_B(O)$ in \eqref{cw}. 

Since $\mathcal{W}$ is an additive net of C$^\star$ algebras it is possible to show that $\mathcal{C}$ has the properties of isotony (1), causality (2) and additivity (4) for von Neumann algebra nets given above Definition~\ref{def:vN}. The standard
property (3) follows from Definition~\ref{def:largeN} (c) as discussed in Appendix~\ref{app:oneparticle}. 
Thus we henceforth take $\mathcal{C}$ to be an additive net as defined in Definition~\ref{def:vN}. 
It is possible that the center is non-trivial for the entire net $\mathcal{C} = \mathcal{C}(K)$ and that $\sigma$ is not faithful for 
for this algebra. These properties have only been shown for the local algebras. In many situations we expect the hyperfinite (7) property for $\mathcal{C}$ and in Section~\ref{sec:crhd} we will explicitly assume this.
 For now it will not be necessary.

Out next goal is to construct the encoding maps $V_N$. We start with the simple case where the state $\sigma$ is faithful for $\mathcal{W}$,
in the subsection after we will see the more general case. 

\subsubsection{$V_N$: Faithful case} 
\label{sec:faith}

As a warm up let us start by assuming that $\pi_\sigma$ is faithful for $\mathcal{W}(K) = \mathcal{W}$. That is $\pi_\sigma(w) = 0 \implies w =0$. Then the cyclic vector $\eta$ arising from the GNS construction is also separating. We will return to the non-faithful case after the simpler discussion that follows. 

Without loss of generality we can take $\sigma_N$ to derive from a vector $= \omega_{ \psi_N} $ with $ \psi_N \in \mathscr{K}_N$. We can do this
by replacing $\mathscr{K}_N$ with a larger purifying Hilbert space (the GNS representation of $\sigma_N$ would do) and working with
the representation of the net $\mathcal{M}_N$ on this larger Hilbert space. In a slight abuse of notation we continue to refer to this representation and the new larger Hilbert space as $\mathcal{M}_N$ and $\mathscr{K}_N$ respectively. 
Once we have made this replacement we note that $\mathcal{M}_N \subset \mathcal{B}(\mathscr{K}_N)$ is no longer a complete net since  $\mathcal{M}_N(K) \neq \mathcal{B}(\mathscr{K}_N)$
and Haag duality is violated on the boundary by the extra purifier. We will thus relax these conditions on $\mathcal{M}_N$ moving forward.

Consider the densely defined linear operator $V_N : \mathscr{H} \rightarrow \mathscr{K}_N$
\begin{equation}
\label{defV}
V_N  \left| [a] \right> = V_N \pi_\sigma(a) \left| \eta \right>  =  \gamma_N(a) \left| \psi_N \right>\, \qquad a \in \mathcal{W}_B
\end{equation}
This is a densely defined operator by our assumption
on the separating property of $\eta$ for $\mathcal{C}(K)$. In particular the definition is consistent since no operator $\pi_\sigma(a) \in  \pi_\sigma(\mathcal{W}(K))$ can annihilate $\eta$. Note that the kernel of $V_N$ contains $ \pi_\sigma(\mathcal{Z}_N) \left| \eta \right>$. 

The (possibly unbounded) operator $V_N$ limits to an isometry on a dense subspace with $a,b \in \mathcal{W}_B$:
\begin{equation}
\label{dense-isom}
\lim_{N \rightarrow \infty} \left( V_N \left| [a] \right> ,V_N \left| [b] \right> \right) = \left< [a] \right. \left| [b] \right>
\end{equation}
We can also derive:
\begin{equation}
\label{dense-reconstruct}
\lim_{N \rightarrow \infty} \left( \gamma_N(a) V_N - V_N \pi_\sigma(a) \right) \left| [b] \right> =0
\end{equation}
which follows from the asymptotic homomorphism property \eqref{ash}. While this seems to be the desired form of an asymptotic code, it has
some drawbacks. $V_N$ being only densely defined makes it technically difficult to work with. In particular if it is unbounded then we have to be very careful with the applicability of these codes beyond the dense subspace. Rather we will assume that the above results for $V_N$ extend to the full Hilbert space as follows:

\begin{definition}
\label{def:fc}
A large-$N$ sector, as given in Definition~\ref{def:largeN}, for which the limiting state $\sigma$ is faithful for $\mathcal{W}$, is called \emph{fully convergent} if there exists
a uniformly bounded sequence of operators $\tilde{V}_N \in \mathcal{B}(\mathscr{H}, \mathscr{K}_N)$ such that:
\begin{equation}
\label{tvn}
\lim_N (V_N - \tilde{V}_N) \left| [w] \right> =0  \qquad \forall w \in \mathcal{W}_B 
\end{equation}
\end{definition}

This definition gives rise to the following.
Consider a single trace algebra $(\H, \mu_N)$, Definition~\ref{def:sta} and a fully convergent large-$N$ sector, Definition~\ref{def:fc}, determined by the sequence $\sigma_N = \omega_{\psi_N}$ 
and that limits to a faithful state $\sigma$ for $\mathcal{W}$ then:
\begin{align}
&wo-\lim_N \tilde{V}_N^\dagger \tilde{V}_N  = 1 \\
&so-\lim_N \tilde{V}_N \pi_\sigma(w) - \gamma_N(w) \tilde{V}_N =0\, \qquad \forall w \in \mathcal{W}_B
\end{align}
Which are defining features of an
asymptotically isometric code $(\mathcal{C}, V_N, \gamma_N)$ that we may as well define and study in its own right. 
We will do this in Section~\ref{sec:aic} after we return to the non-faithful case.

Up until now we have relied on what might be thought of as kinematical properties of the $\gamma_N$ maps and correlation functions in the state $\psi_N$ which arise from the large-$N$ limits of matrix like theories.  In order to establish a given map $\gamma_N$ satisfies the fully convergent property Definition~\ref{def:fc}, we need to know more details on \emph{how} 
these correlation functions approach the large-$N$ answers. There could be be several different physical mechanisms that lead to \eqref{tvn}, so we have instead
kept this definition broad. 
In Appendix~\ref{app:uniform} we give one scenario involving an improved convergence
requirement on the large-$N$ correlation functions, $\sigma_N(\gamma_N(w_1) \ldots \gamma_N(w_k))$, that we call ``uniform operator system convergence''.
This requires the existence of a sequence of operator subsystems  of $\mathcal{W}_B$ (these are linear spaces of operators that need not form
an algebra) whose union is norm dense in $\mathcal{W}$ and for which the convergence properties of the correlation function is uniform on a fixed system. 
See Definition~\ref{uosc} for more details. Physically we imagine the operator system allows us to work with linear spans of the unitary operators $w_N(h)$ - that have bounded amplitude and allow only bounded frequency and wavelength in the smearing function $h$, and so for sufficiently high $N$ such operators should have all converged. This later assumption  sounds at least plausible and so happily it can be used to prove the  \emph{fully convergent} condition given in Definition~\ref{def:fc}. See Lemma~\ref{lem:uosc} for details.

\subsubsection{$V_N$ : Non-faithful case}

There are naturally two possible ways for $\sigma$ to not be faithful. 

The existence of degenerate subspace $H_0 \subset H$ for $\beta$ (after norm completion)  gives rise
to a center for $\mathcal{W}$. 
Then we can imagine states $\sigma$ where this
center is not faithfully represented on the GNS Hilbert space (such as the fixed energy/charge states $\sigma = \sigma_g$ envisioned in Section~\ref{sec:deg}).
 We deal
with this in the definition of $V_N$ by explicitly projecting out the center/degeneracies using $p([h])$ defined in Section~\ref{sec:deg2}:
\begin{equation}
\label{defV2}
V_N  \left| [w(h)] \right> = V_N \pi_\sigma(w(h)) \left| \eta \right>  =  \gamma_N(w(p[h])) \left| \psi_N \right>
\end{equation}
for all $w(h) \in \mathcal{W}_u$ and $h \in \H$. In particular \eqref{defV2} should now replace \eqref{defV}.
The linear extension to $\mathcal{W}_B$ is automatic. This definition is consistent now (assuming that the non-faithfulness is only due
to the central modes) since $\pi_\sigma(w(h )) = \pi_\sigma(w(p[h]))$ and if some linear combination vanishes such as $\pi_\sigma(w(h)) - \pi_{\sigma}(w(h + h_0) )= 0$,
where $h_0 \in \H_0$, then $ \gamma_N(w(p[h])) -  \gamma_N(w(p[h + h_0])) = 0$ also vanishes trivially. Also since $\omega_{\psi_N}$ limits to $\sigma$,
the fluctuations of the ($N_T^{1/2}$ rescaled) charges will vanish for $\psi_N$ in the limit and thus, inside correlation functions of $\psi_N$, $\gamma_N(w(h))$ and $\gamma_N(w(h +h_0))$ will have the same
limiting $N \rightarrow \infty$ behavior. Hence \eqref{dense-isom} and \eqref{dense-reconstruct} will be preserved.  
In particular, for the later equation we do not modify $\gamma_N(w(h)) 
\rightarrow \gamma_N(w(p[h]))$ since this would not preserve the net structure and
lead to unacceptable violations of boundary causality. 

The second way that $\sigma$ may not be faithful is the case where $\mathcal{C}$ is a complete net, with $\mathcal{C}' =  \mathbb{C} 1$ 
then $\eta$ clearly cannot be separating.\footnote{The case $\mathcal{C}' = \mathbb{C} 1$ is the extreme example. In general we will need a combination of both cases studied in Section~\ref{sec:faith} and here. For example in the single particle Hilbert space discussed in the Appendix~\ref{app:oneparticle}, there is a subspace that gives rise to an  irreducible part of the Fock space representation ($D^2 =-1$)  and the orthogonal complement that gives the faithful part of the Fock space. These two orthogonal subspace become tensor products on the Fock space and so can be treated independently.}
This would be the case for a code based on a sequence of vacuum states (after dealing with the center). We can proceed if we assume the existence of another state that is generated from $\sigma$ and that is faithful. 

Consider a sequence of operators $a_i \in \mathcal{W}_B$ such that
\begin{equation}
\label{faitht}
\tilde{\sigma}(\cdot) = \sum_i \sigma( a_i^\dagger \cdot a_i )
\end{equation}
is a faithful state on $\mathcal{W}$. We assume this arises from the limit of states on $\mathcal{M}_N(O)$ with:
\begin{equation}
\label{limitfaitht}
\tilde{\sigma}_N = \sum_i \sigma_N( \gamma_N(a_i)^\dagger \cdot \gamma_N(a_i) )
\end{equation}
where we assume the sum converges for each $N$.

Using this we can construct a new representation $\tilde{\pi} = \pi_{\sigma} \otimes 1_R$ on $\mathscr{H} \overline{\otimes} \mathscr{H}_R$
for some separable $\mathscr{H}_R$. Then
$\tilde{\eta} = \sum_k \left| \left[ a_k \right] \right> \otimes \left| k \right>_R$ is  separating
for $\pi_\sigma(\mathcal{W})'' \otimes 1_R = \mathcal{B}(\mathscr{H}) \otimes 1_R$
 where $\left| k \right>_R$ is an orthonormal basis on the reference
and $a_k \in \mathcal{W}_B$ are the infinite sequence of operators appearing in \eqref{faitht}. 
$\tilde{\eta}$ is also cyclic for the same algebra since it is separating for the commutant. 
The sequence of states  $\tilde{\psi}_N = \sum_k \gamma_N(a_k) \left| \psi_N \right> \otimes \left| k
\right>_R$ are normalizable. 
We then define:
\begin{equation}
W_N  \pi_\sigma(a) \otimes 1_R \left| \tilde{\eta} \right> = \gamma_N(a) \otimes 1_R \big| \tilde{\psi}_N \big>
\qquad a \in \mathcal{W}_B
\end{equation}
This is again a densely defined operator $W_N : \mathscr{H}  \overline{\otimes} \mathscr{H}_R \rightarrow \mathscr{K}_N  \overline{\otimes} \mathscr{H}_R$
by the separating property for $\tilde{\eta}$.  We find the asymptotic behavior:
\begin{equation}
\label{dense-isom2}
\lim_{N \rightarrow \infty} \left( W_N \big| \widetilde{[a]} \big> , 1\otimes E_i W_N \big| \widetilde{[b]} \big> \right) = \big< \widetilde{[a]} \big| E_i \big| \widetilde{[b]} \big>
\end{equation}
for all minimal projections $E_i = \left| i \right> \left< i \right| \in \mathcal{B}(\mathscr{H}_R)$ and all $a,b \in \mathcal{W}_B$
where $\big| \widetilde{[b]} \big>
= \tilde{\pi}(b) \left|\tilde{\eta}\right>$. And similarly:
\begin{equation}
\label{dis3}
\lim_{N \rightarrow \infty} E_i \left( \gamma_N(a) \otimes 1_R \,\, W_N\big| \widetilde{[b]} \big> - W_N \pi_\sigma(a) \otimes 1_R  \big| \widetilde{[ b]} \big> \right) = 0
\end{equation}
for all $i$ and $a,b \in \mathcal{W}_B$.

Then:
\begin{definition}
\label{def:fcnf}
A large-$N$ sector, as given in Definition~\ref{def:largeN}, is called \emph{fully convergent} if there exists a faithful state as given in \eqref{faitht},
arising from a limit of the states given in \eqref{limitfaitht}  and such that there is 
a uniformly bounded sequence of operators $\tilde{W}_N : \mathcal{B}(\mathscr{H}  \overline{\otimes} \mathscr{H}_R, \mathscr{K}_N  \overline{\otimes} \mathscr{H}_R)$ with:
\begin{equation}
\lim_N (W_N - \tilde{W}_N) \big| \widetilde{[w]} \big>  \qquad \forall w \in \mathcal{W}_B 
\end{equation}
where $W_N$ was defined above. 
\end{definition}

Assuming such a fully convergent large-$N$ sector we can demonstrate that:
\begin{align}
&wo-\lim_{N \rightarrow \infty}  \tilde{W}_N^\dagger E_i \tilde{W}_N=E_i   \\
&so-\lim_{N \rightarrow \infty}E_i \left( \gamma_N(a) \otimes 1_R \tilde{W}_N -\tilde{W}_N \pi_\sigma(a) \otimes 1_R \right) = 0
\end{align}
After which we can take matrix elements $\left< i \right| \tilde{W}_N \left| i \right> : \mathscr{H} \rightarrow \mathscr{K}_N$
and use this as our encoding isometry. The result is again an asymptotically isometric code. 
We expect a similar derivation, as with the faithful case, 
of a fully convergent large-$N$ sector (Definition~\ref{def:fcnf}) from the uniform operator system convergence condition (Definition~\ref{uosc}) discussed in Appendix~\ref{app:uniform} and the assumption of the existence of a state \eqref{faitht},
arising from a limit of the states \eqref{limitfaitht}.

\subsection{Comments on AdS/CFT}

\label{sec:comments}

In AdS/CFT we have the extrapolate dictionary which is $\gamma_N$ restricted to $\mathcal{W}_B(O)$ since the Weyl operators are by definition smeared near the
boundary of AdS. We have the HKLL\cite{Hamilton:2006az} reconstruction
maps which involve the equality between $\mathcal{C}(O)$ and the weak closure $\pi_\sigma(\mathcal{W}_B(O))''$. This is HKLL because it relates
the bulk operators $\mathcal{C}(O)$ to the boundary operators under limits.\footnote{Actually this is really HKLL for the causal domain. See Appendix~\ref{app:causal} and the discussion in Section~\ref{sec:aic}.} In particular since $\mathcal{C}(O)$ is computed using the double commutant and the single trace sector is describing
bulk quantum fields - the  commutant operation respects bulk causality. Thus $\mathcal{C}(O)$ can only involve fields in the causal wedge/domain.

HKLL then (implicitly) applies the extrapolate dictionary \cite{Banks:1998dd}, which in our picture are the maps $\gamma_N$, to complete
the operator reconstruction as a boundary operator. In principle it might be possible to give a simple extension of $\gamma_N$ from $\mathcal{W}_B(O)$  to the weak closure $\mathcal{C}(O)$.
However we emphasize that this is not automatic due to an order of limits issue: we need to take large-$N$ last and the weak closure first rather than the other way around. 
In earlier
drafts of this paper we had essentially assumed this extension of $\gamma_N$ to $\mathcal{C}(O)$ is possible, however realized that this is not necessarily true or even needed. In particular in AdS/CFT it has never been established that HKLL works at this level - given a formal expression for some bulk operator as a limit of sums of near boundary operators $\mathcal{W}_B(O)$  one simply assumes the extrapolate dictionary $\gamma_N$ factors through this expression. 
Since arguments in favour of HKKL are usually formal anyway such an expression is sufficient. However in reality the convergence property through this map is not obvious. And actually we will demonstrate it is not important. For one thing, entanglement wedge reconstruction will be available to us, as we establish in Section~\ref{sec:crhd} and such reconstructions do not have this convergence issue. Furthermore the reconstruction at the level of $\mathcal{W}_B(O)$ is already sufficiently constraining from a causality point of view.

 We then have the global bulk to boundary map $V_N$ which were constructed using the extrapolate dictionary and a linear span of the global boundary insertions of operators.
 This is of course the usual procedure \cite{Almheiri:2014lwa,Cotler:2017erl,Jafferis:2017tiu}. 

\section{Asymptotically isometric codes}

\label{sec:aic}

We are now ready to define the main axiomatic framework of interest in this paper.
Consider a sequence of CFTs labelled by $N \in \mathbb{N}$ each with an additive net $K \supset O \rightarrow \mathcal{M}_N(O) \subset \mathcal{B}(\mathscr{K}_N)$ in the vacuum sector.

\begin{definition}
\label{def:code}
Given such a sequence of theories we define an \emph{asymptotically isometric code} as the triple $(\mathcal{C}, V_N, \gamma_N)$ where $V_N \in \mathcal{B}(\mathscr{H},\mathscr{K}_N)$
is a bounded encoding map, $\mathcal{C}$ is an additive net of von Neumann algebras over $K$, 
and $\gamma_N$ is a sequence of uniformly in $N$, pointwise bounded unital $^\star$-maps, $\gamma_{N} : \mathcal{A} \rightarrow \mathcal{M}_N $ defined on a weakly dense domain $\mathcal{A}(O)$ of $\mathcal{C}(O)$ for all $O$.\footnote{That is $\mathcal{A}(O)'' = \mathcal{C}(O)$ for all $O$.} The algebras $\mathcal{A}$ are a sub-net $\mathcal{A}(O) \subset \mathcal{C}(O)$ of $^\star$-algebras and the $\gamma_N$ maps are consistent
with the common net structure shared by the domain and range. Altogether  these must satisfy:
\begin{align}
& wo-\lim_{N \rightarrow \infty} V_N^\dagger V_N = 1 \\
\label{ssv}
& so-\lim_{N \rightarrow \infty} (\gamma_N(a) V_N - V_N a) =0
\end{align}
for all $a \in \mathcal{A}$.

We say that the sequence of states $\psi_N \in \mathscr{K}_N$ is \emph{represented} on the code if there is some $\psi \in \mathscr{H}$ such that
$\| \psi_N - V_N \psi \| \rightarrow 0$. 
A \emph{standard asymptotically isometric code} has a sequence of vectors $\varPhi_N \in \mathscr{K}_N$ represented by $\eta \in \mathscr{H}$ on the code where $\eta$ (resp. $\varPhi_N$) is cyclic and separating for $\mathcal{C}(O)$ (resp. $\mathcal{M}_N(O)$) for all local regions $O \subset K$. 

\end{definition}

All codes that we consider will be standard codes, and we may forget to explicitly state this.

Note that we have not stipulated that $\mathcal{M}_N$ is a complete net. This allows for vector states $\psi_N \in \mathcal{K}_N$ that are not pure
states for $\mathcal{M}_N$ which implies that $\mathcal{M}_N$ is not irreducibly represented on $\mathscr{K}_N$.  This could occur if, relative to the original CFT net,
we include a reference $\mathscr{K}_N = \mathscr{K}_N^0 \overline{\otimes} \mathscr{R}_N$ where $\mathscr{K}_N^0$ now refers to the CFT Hilbert space with no reference. 
Similarly, while we often consider $K =  K_0 \equiv S_{d-1} \times \mathbb{R}$ we can also consider the QFTs to live on more general regions. 
For example it is sometimes necessary to include a reference as an explicit boundary region, as is done for the Page curve computations in \cite{Penington:2019npb,Almheiri:2019psf,Penington:2019kki,Almheiri:2019qdq}.
In which case there might be a code based on the larger causal net $K =  K_0 \times \mathbb{Z}_2$
setting $(O \times 0 )' = O' \times 1$ and
\begin{equation}
\mathcal{M}_N(K_0 \times 1) =  \mathcal{M}_N(K_0)  \bar{\otimes}  \mathcal{B}(\mathscr{R}_N)\,
\qquad \mathcal{M}_N(K_0 \times 0) =  \mathcal{M}_N(K_0)  \bar{\otimes} 1
\end{equation}
And this could allow us to maintain a complete (Haag dual) net for the microscopic theories.  
We could also define codes for the causal net $K = K_0 \times K_0$ a doubling of the original CFT $\mathcal{M}_N= \mathcal{M}_N(K_0) \overline{\otimes} \mathcal{M}_N (K_0)$
acting on the doubled Hilbert space $\mathscr{K}_N = \mathscr{K}_N^0 \overline{\otimes} \mathscr{K}_N^0$. This  would allow us to study the sequence of thermo-field double states while maintaining locality of the purifying subsystem. 
While these operations of tensoring in a new Hilbert space factor, and picking entangled states on this larger Hilbert space, by definition keeps us in the vacuum sector of the $\mathcal{M}_N(K_0)$ theories, in the large-$N$ limit these operations can easily give rise to different codes. Hence superselection sectors arise from the limit. See the discussion Section~\ref{sec:d}. 

\begin{remark}
\label{rem:one}
We have established in Section~\ref{sec:sta} that a single trace algebra (Definition~\ref{def:sta}) and a large-$N$ sector (Definition~\ref{def:largeN}) 
with the fully convergent condition (Definition~\ref{def:fcnf}) gives rise to an asymptotically isometric code as defined above. Where in particular the weakly dense domain of $\gamma_N$ is $\mathcal{A}(O) = \pi_\sigma(\mathcal{W}_B(O))$.  The code is standard
if the state $\psi_N$ that the large-$N$ sector is based on is cyclic and separating for $\mathcal{M}_N(O)$ which will typically be the case. 
\end{remark}

In general the strategy we have followed is to define a code that is as minimally constrained as possible. 
The goal was to do this while maintaining Remark~\ref{rem:one}, and at the same time
allowing for a strong enough condition that will allow us to prove that these codes have nice quantum error correcting properties.
  In particular the maps $\gamma_N$ have a very
weak condition imposed on them. We do not require they are quantum channels, ultraweakly continuous, or even bounded maps. 
They just have to reconstruct on a dense $^\star$-subalgebra $\mathcal{A}$, that may not even be a $C^\star$ algebra. 
And as we will see this is sufficient in Theorem~\ref{thm:qec}.

We show a simple consequence of these definitions:
\begin{lemma}
\label{lem:C}
Given an asymptotically isometric code $(\mathcal{C},V_N,\gamma_N)$ then the bounded linear operator $V_N$ are in-fact uniformly bounded:
\begin{equation}
\| V_N^\dagger V_N \|  \leq C \qquad \forall N
\end{equation}
for some fixed $C$.
\end{lemma}
\begin{proof}
Follows since a sequence of operators $V_N^\dagger V_N$ is weakly convergent to some fixed operator ($1$ in our case) iff the sequence uniformly bounded by the uniform boundedness principle. See comment under Corollary~\ref{cor-ubp} in Appendix~\ref{app:convergence}. 
\end{proof}

We will interpret $\mathcal{C}(O)$ as the causal wedge for $O$. We give two pieces of evidence for this. Firstly the net $\mathcal{C}$ is additive by assumption, and also
by construction from the previous section. Additivity means
that we can generate $\mathcal{C}(O)$ by using a boundary spacetime cover of $O$ by small diamonds $O_i$, that is $O = \cup_i O_i$
implies that $\mathcal{C}(O) = \vee_i \mathcal{C}(O_i)$. We claim this additivity property is the hallmark of the causal wedge, since the bulk dual to these
local diamond regions $O_i$ cover a region near the boundary that generates the bulk causal wedge region via causal completion. See Appendix~\ref{app:causal} for a detail discussion of this matter, including
a slightly different definition of the causal wedge, dubbed the bulk causal domain, that is more naturally suited to the causal structure of bulk quantum fields. The bulk causal domain is typically slightly larger than the causal wedge but satisfies many of its properties. We will however continue to refer to $\mathcal{C}(O)$ as the causal wedge algebra or simply the causal wedge. See Figure~\ref{fig:bulk-add}.

 \begin{figure}[h!]
\centering
\includegraphics[scale=.4]{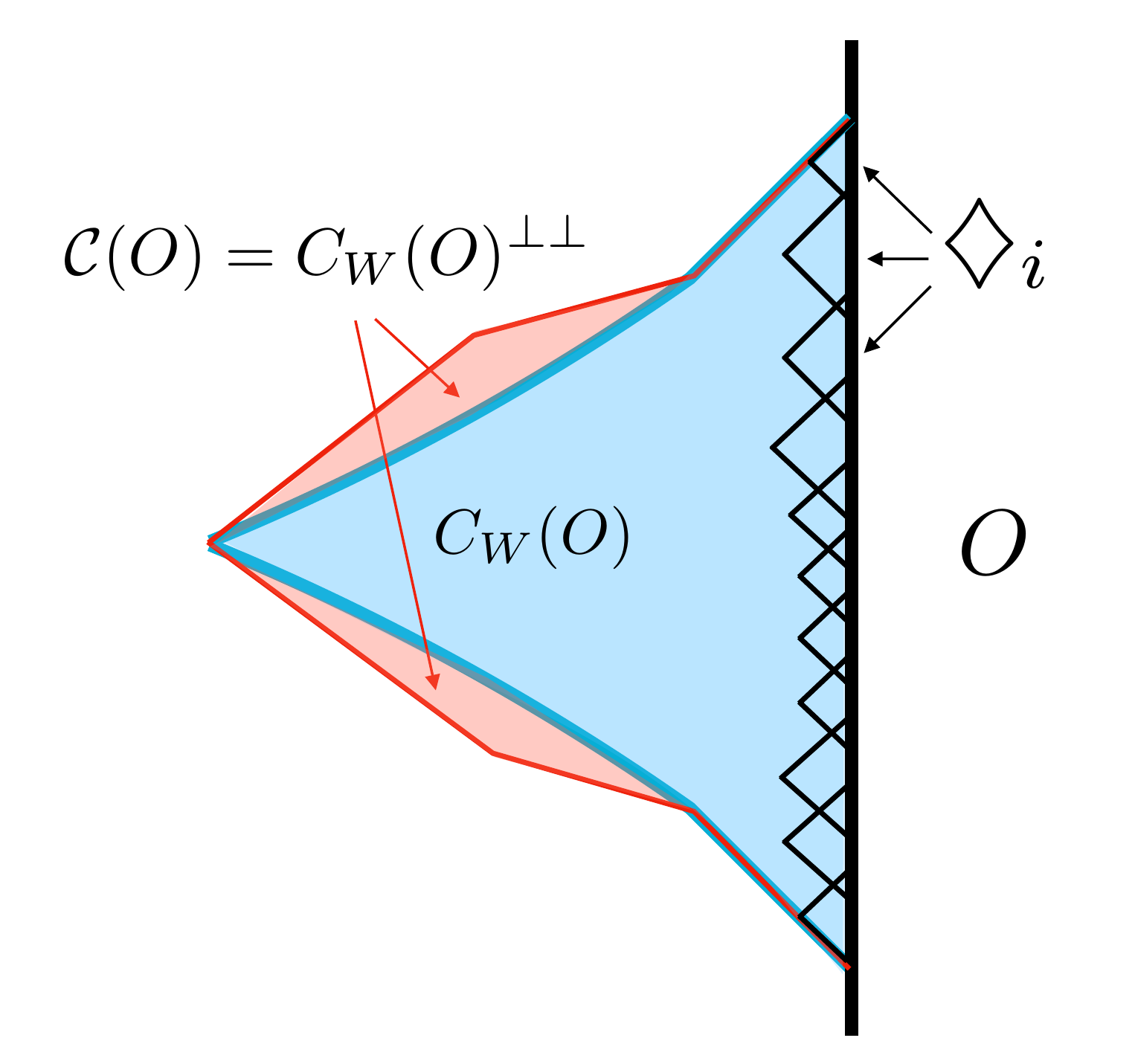}
\caption{In AdS/CFT the additivity property of the causal domain comes from a conjecture for the dynamics of general bulk QFT which is supported by Borchers timelike tube theorem. See Appendix~\ref{app:causal}. \label{fig:bulk-add}}
\end{figure}

Secondly, by definition, the reconstruction maps $\gamma_N$ are simple.  This is because in specific models we can explicitly write  down these maps and they do not
involve anything complicated like the boundary density matrix, or modular operators. The causal wedge is of course expected to be simple to reconstruct.
While the maps need only work on a weakly dense sub-algebra $\mathcal{A}(O)$ this should be sufficient for our purposes.  It seems plausible the maps we constructed in Section~\ref{sec:sta} could be extended further to $\mathcal{C}(O)$, however this need not be the case.
In Section~\ref{sec:comments} we gave an interpretation of 1. $\gamma_N$ as the extrapolate dictionary and  2. the statement $\mathcal{A}(O)'' = \mathcal{C}(O)$ as the Rindler-like HKKL reconstructions, yet we left open the question of convergence when composing these procedures to get a final operator in the microscopic theory.
Baring the convergence question, both of these procedures are simple (allowing for fixed small errors.)
We will later prove Theorem~\ref{thm:qec} that gives a quantum channel reconstructing of all of $\mathcal{C}(O)$, however we no longer have reason to expect
this is simple. 

Note that even if the defining net $\mathcal{M}_N$ is complete then
$\mathcal{C}$ need not be complete. In particular the commutant $\mathcal{C}' = \mathcal{C}(K)'$ may be non-trivial. 
Such behavior is related to the emergence of black holes in the bulk, with $\mathcal{C}'$  describing operators behind the causal horizon
of the full boundary theory. An example of where this case arises is a pure state black hole formed from collapse.
 A natural question then is whether we can
extend the causal wedge algebra for $K$ to a larger subalgebra of $\mathcal{B}(\mathscr{H})$. Where by extension we mean, with the ability
to reconstruct this larger sub-algebra from the net $\mathcal{M}_N = \mathcal{M}_N(K)$.
In-spite of the fact that all the information is in principle available to the entire boundary $K$, there is no reason for the map $V_N$ to be invertible at fixed $N$, so this turns out to be a subtle question. If we simply require a version of the asymptotic reconstruction statement \eqref{ssv}, then a simple map that in principle does the trick is $\beta_N : n \rightarrow V_N n V_N^\dagger$ for all $n \in \mathcal{B}(\mathscr{H})$.\footnote{There is very likely a connection to the physical construction developed by Papadodimas-Raju \cite{Papadodimas:2012aq,Papadodimas:2013jku} which considered
a similar setup. See also comments in \cite{Chandrasekaran:2022eqq}. Note that we do not need to use
 mirror operators to reconstruct $\mathcal{C}'$ since these are already defined by the map $V_N$. 
}
 It is easy to show that this works if we demand that the code has strong operator convergence, instead of weak: $so-\lim_N V_N^\dagger V_N = 1$. 
Indeed we derived strong operator convergence for $V_N^\dagger V_N$ in Appendix~\ref{app:uniform} based on the strengthened ``uniform operator system'' convergence assumption. We could have just as easily assumed strong operator convergence in our definition of an asymptotic code. We chose not to simply because it plays a minimal role in this paper. We expect there might be other reasons to impose strong convergence, such as while studying different large $N$ sectors. Future work might need to impose this condition. 
Note that the $\beta_N$ maps given above are not unital and so their physical status are not so clear. We will do slightly better in Lemma~\ref{lem:global}. 
 
We now consider a generalization of this non-completeness for local boundary regions $O$ rather than all $K$. This is more robust compared to the global reconstuctions due to the presence of large amounts of entanglement and the physical interpretation in terms of quantum error correction.  We can prove the main theorems only assuming weak operator convergence of $V_N^\dagger V_N$. 
Consider the following inclusion of von Neumann algebras:
\begin{equation}
\mathcal{C}(O') \subset \mathcal{C}(O)'
\end{equation}
Suppose the inclusion is proper then we say that Haag duality is violated for the bulk theory. This generalizes the previous non-completeness since
$\mathcal{C}(K') = \mathcal{C}(\emptyset) = \mathbb{C} 1$. 
Violations of Haag duality are not limited to the existence
of horizons in the bulk, and are expected to be the generic situation. See Figure~\ref{fig:haag}. In Section~\ref{sec:crhd} we discuss extensions of $\mathcal{C}(O)$, the constraints on such extensions,
and when maximal extensions exists.  A maximal extension defines the notion of an entanglement wedge. For recent discussions of the relation between Haag duality, additivity in the algebraic approach to AdS/CFT see \cite{Casini:2019kex,Benedetti:2022aiw}.
In \cite{toappear} we plan to discuss a few settings that provably do not have Haag duality violations so that the causal wedge is already maximal. This then implies equality of the causal and entanglement wedges. These settings are well known from AdS/CFT and always involve codes that realize boundary geometric symmetries that hold fixed the causal wedge.

 \begin{figure}[h!]
\centering
\includegraphics[scale=.4]{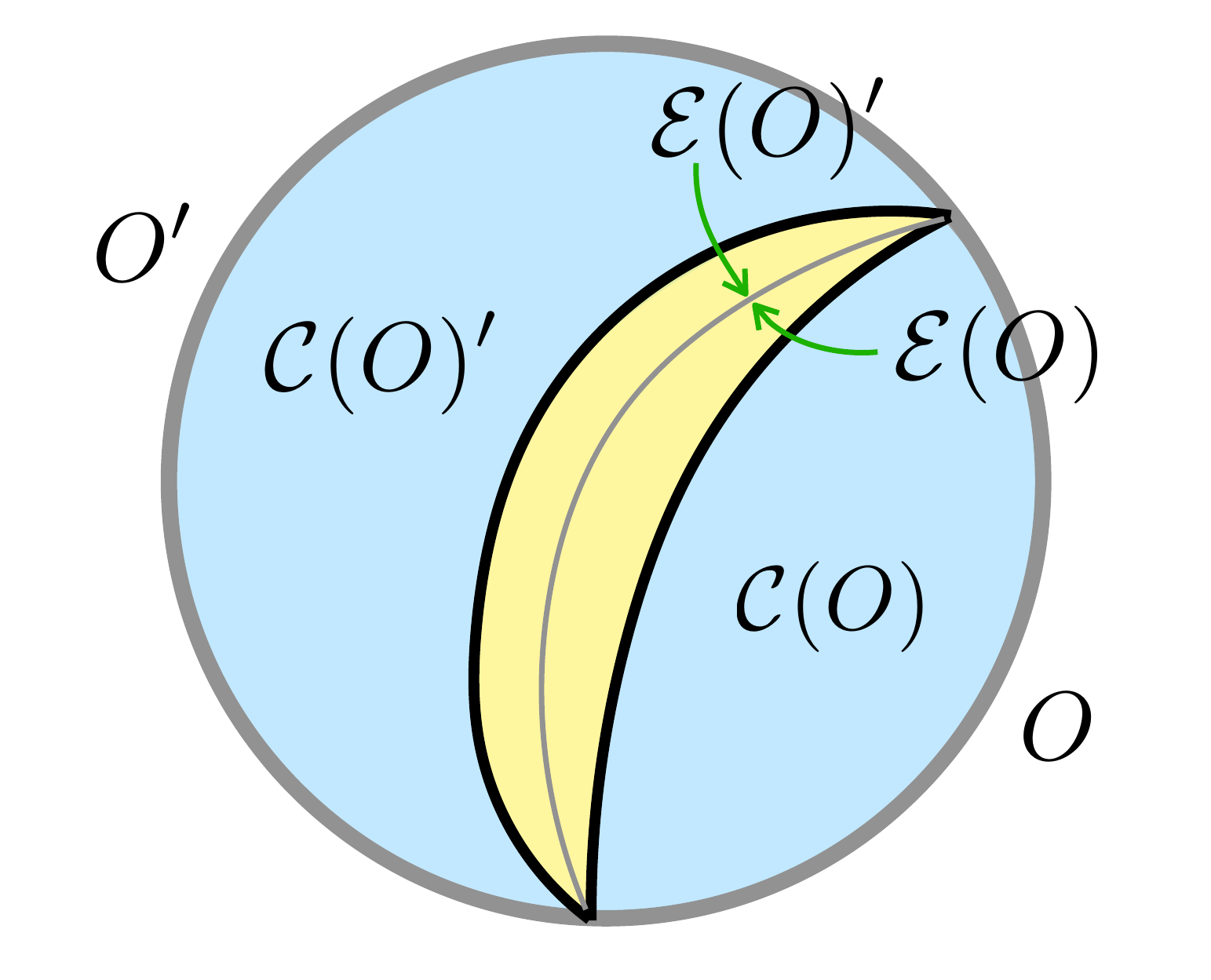}
\caption{The extension of the causal wedge to the entanglement wedge can give back Haag duality $\mathcal{E}(O') = \mathcal{E}(O)'$. 
\label{fig:haag}}
\end{figure}

\subsection{Hyperfinite condition}

An important assumption that we make in Section~\ref{sec:crhd} is that the bulk algebras $\mathcal{C}(O)$ are all hyperfinite (see Definition~\ref{def:vN}).
A von Neumann algebra is hyperfinite if it can be generated by an increasing sequence of finite-dimensional von Neumann subalgebras.  
One way to establish this is via the split property: if the property (8) given above Definition~\ref{def:vN} holds for the net $\mathcal{C}(O)$, then the hyperfinite condition applies.

We expect the split property applies to CFTs at fixed $N$. 
There is an important question of whether the split property is maintained by the net $\mathcal{C}$. It is often stated that generalized free fields
do not satisfy the split property. We disagree with this, or at least we have in mind a different assignment of operators to regions $O$ than is usually assumed.
We are choosing to assign the causal wedge $\mathcal{C}(O)$ to regions to make the following statements. 

If the generalized free fields are just that of QFTs in higher dimensions then we expect that these will have algebras that satisfy the split property. 
In particular these QFTs have perfectly good high temperature thermodynamic properties.
Of course the thermodynamics cannot be that of a $d$ dimension QFT, since the fields live in higher dimensions. This is where the holographic bounds come into play, but this is independent of the question of the existence of a split for $\mathcal{C}(O)$. In AdS/CFT for the bulk Hilbert space ordered by energy, thermal AdS typically involves excitations that live in $AdS_{d+1} \times M_{D-1-d}$ for some compact $M_{D-1-d}$.  The gravity modes
propagate in this space and correspond to an infinite tower of lower dimensional Kaluza Klein fields - each of which we might include in the test function space $H$. Treating the compact space $M_{D-1-d}$ as part of the causal wedge we expect it is the thermal properties of these higher dimensional fields that will determine the existence or not of a split. Since there is no Hagedorn growth at this point we expect the split will apply \cite{Buchholz:1989bj,Buchholz:2006hp}.
Of course the bulk thermal states for such modes will be thermodynamically unstable for $T \gtrsim 1/R_{AdS}$  in the canonical ensemble, however this is not actually the relevant ensemble here. We are simply working with a bulk Hilbert space description and we do not care if the bulk thermal state maps to the actual boundary thermal state. 
Of course eventually at high enough energies the gas
will give way to stringy excitations which will give way to small localized black holes and then big black holes where the high energy thermal behavior is that expected
of a $d$ dimensional QFT \cite{Aharony:1999ti}. These later effects occurs at an $N$ dependent energy so we can ignore them in our vacuum codes. We are imagining a scenario where we pick the string scale $\ell_s$ to scale to zero as $\ell_P \rightarrow 0$
in the large-$N$ limit. 

However if the code subspace is actually a string theory we expect that
the split property will break due to Hagadorn like growth of the density of states. 
So it is not clear the hyperfinite assumption applies to stringy codes. However since we expect a countable number of stringy modes then it is possible that the split property applies to each mode, so that the countable infinite tensor product of these will still be hyperfinite. Thus we view the hyperfinite assumption as plausibly applying to a large class of holographic theories, stringy or not.

It is natural to take the algebras $\mathcal{C}(O)$ for compact $O$ to be the unique hyperfinite type-III$_1$ factor. This can often be established from the properties
we have already discussed.  The von Neumann algebras for non compact geometric regions, $\mathcal{C}(K)$ or $\mathcal{C}(K_0 \times 0)$ etc. are allowed to have a non-trivial center.

\subsection{Some examples}

\label{sec:examples}

A short list of codes we might consider include\footnote{For each code we include some expected properties of these codes, that we envision could be derived from an explicit construction of these
codes using Weyl algebras as in Section~\ref{sec:sta}, perhaps under some additional technical assumptions. Unless otherwise stated the relevant causal regions
will be the net $K = K_0 \equiv S_{d-1} \times \mathbb{R}$.}:
\begin{enumerate}
\item The \emph{vacuum} code. Here the CFT vacuum $\Omega_N$ are represented on the code by some vector $\eta \in \mathscr{H}$. Additionally the net $\mathcal{C}$ is in the vacuum sector with $(u(g), \eta)$. 
The net $\mathcal{C}$ is irreducibly represented on $\mathscr{H}$ with $\mathcal{C}' = \mathbb{C} 1$.

\item A \emph{thermal} code.  We work with a non-complete net $\mathcal{M}_N \subset \mathcal{B}(\mathscr{K}_N)$ allowing for a purifying reference. 
The canonical ensembles $\sigma_N^\beta$ are represented with vectors $\left| \beta_N \right> \in \mathscr{K}_N$, 
such that $\omega_{\beta_N}|_{\mathcal{M}_N} = \sigma_N^\beta$. 
The $\beta_N$ satisfy the KMS condition for $\mathcal{M}_N$
and time translations $U_N(t) \in \mathcal{M}_N$.  
 The $\beta_N$ are in turn represented on the code by $\beta \in \mathscr{H}$.

For small temperatures  $T< T_c$ we expect this code
to be in the same sector  as the vacuum code (the net $\mathcal{C}$ is isomorphic to the vacuum code), with conformal symmetries including time translations unitarily represented on $\mathscr{H}$. In particular $\mathcal{C}$ will be of type-I$_\infty$. 
For temperatures above the deconfinement transition $T>T_c$ the net $\mathcal{C}$ will develop a center for  $\mathcal{Z}(\mathcal{C}(K)) = L^\infty(\mathbb{R}^{n_Q})$ and will be of type-III$_1$
implying $\mathcal{C}$ cannot be the vacuum code. In particular $\mathcal{C}'$ represents operators behind the horizon of a black hole. 
Time translations will still be unitarily represented, although they will have an outer action on $\mathcal{C}$. The Hamiltonian is now not affiliated to $\mathcal{C}$. 

\item A \emph{microcanonical code}. This is a high temperature thermal code where the energies and charges are fixed asymptotically, such as in \eqref{fixedg}, so that
there is a purification $\big| \beta_N^g \big> \in \mathscr{K}_N$ of $\sigma_N^g$ that is represented by $\left| \beta^g \right>$ on the code. These will have very similar properties to the high temperature thermal codes with $\mathcal{C}$ a type-III$_1$ factor (with the center absent.)  Again time translations will be unitarily represented but only as an outer action.  
\item \emph{Time shifted} codes. These are based on the above microcanonical codes, but where now $U_N(t) \big| \beta_N^g \big>$ is represented on the code. Above the deconfinement temperature these codes 
will be disjoint/perpendicular to the original microcanonical codes for any fixed $t \neq 0$. 

\item A \emph{thermofield double} code. By enlarging the net to $K_d = K_0 \times K_0$ we can track the locality of the purifying system.
We expect these codes $\mathcal{C}$ over $K_d$ will be complete with $\mathcal{C}(K \times \emptyset)$ behaving like a thermal code.
To avoid the appearance of a center we should work with the appropriate microcanonical TFD state.
In the low temperature phase we will have $\mathcal{C}(K_d) \cong \mathcal{C}(K  \times \emptyset) \overline{\otimes}
\mathcal{C}(\emptyset \times K)$. At high temperatures this factorization will be absent. Although we do expect $\mathcal{C}(K_d)' = \mathbb{C} 1$.

\item A \emph{pure state black hole} formed from collapse. Consider applying a heavy ``dust operator''  smeared near a fixed time slice with energy $\sim N_T$ sufficient
to form a large Black Hole - for a precise microscopic realization \cite{Anous:2016kss}. 
Then consider time evolving this state for some time so that the Black Hole settles to equilibrium. Correlation functions
will agree with the thermal correlation functions at late time, thus we wait times of order $N^\#$ for some small $\#$ to achieve this in the limit.  The large-$N$ correlators would look thermal in the limit $N \rightarrow \infty$, so we expect that $\mathcal{C}$ would look like
the thermal code, possibly with large energy fluctuations. Note that the main difference to the high temperature thermal code is that here $\mathcal{M}_N' = \mathbb{C} 1$
while $\mathcal{C}'$ will be non-trivial type-III$_1$ algebras in both cases. Hence it is possible to imagine attempting to reconstruct operators behind $\mathcal{C}'$ from $\mathcal{M}_N$ in this case, by extending the entanglement wedge of $K$ to the full bulk algebra. 

\item An \emph{energy eigenstate}. Consider the sequence of states $\left| E_N \right>$ where we order the energy eigenstates and pick the
$(e^{S_N})'th$ one holding fixed the entropy of a large black hole $S_N = S_N(T)$ for $T>T_c$. By assuming the eigenstate thermalization hypothesis we expect one can construct codes for such states
since the correlation functions will look thermal. The large-$N$ limit will give thermal correlators in the single trace sector, so the putative bulk will have an equilibrium horizon with $\mathcal{C}'$ non-trivial and $\mathcal{M}_N' = \mathbb{C} 1$, similar to the previous code. The central energy fluctuations  will be absent in $\mathcal{C}$. 
\end{enumerate}

\section{QEC, complementary recovery and Haag duality}

\label{sec:crhd}

Having defined our codes, we now wish to study their quantum information properties.
Our first task is to prove that extensions of the causal wedges maintain boundary causality. 
It turns out that this constraint comes from a well known phenomenon in quantum information theory called information-disturbance tradeoff \cite{kretschmann2008information}. 
We will then use this to extract causal properties such as entanglement wedge nesting.
Our second task will be to discuss various conditions under which we can prove the existence of a maximal extension of the causal wedge, culminating
in an asymptotic version of the JLMS \cite{Jafferis:2015del} condition and an asymptotic version of Takesaki's theorem \cite{takesaki1972conditional} on the stability of modular flow. 

Compared to the semi-informal results in Section~\ref{sec:sta} this Section will be based on completely rigorous results. Proofs of the Theorems stated here will be given in Section~\ref{sec:p}.
These proofs rely in some way on the assumption that the bulk algebras are hyperfinite. 
While the theorems we establish might also be true in the non-hyperfinite case, our proof techniques cannot be applied then.
It is a logical possibility that these results are simply not be true in the non-hyperfinite case. Indeed there are many well known
conditions that are equivilent to the hyperfinite condition such as amenable and injective von Neumann algebras \cite{takesaki2003theory}. The applicability of Theorem~\ref{thm:qec} might simply be another such condition.  

In the theorem statements and proofs we will use dual maps and Petz maps. These are defined as follows:

\begin{definition}
\label{def:dual}
Given a faithful quantum channel $\beta : \mathcal{N}\subset \mathcal{B}(\mathscr{H})  \rightarrow \mathcal{L} \subset \mathcal{B}(\mathscr{K})$ 
and a cyclic and separating vector $H \in \mathscr{K}$ with
$\mathcal{N}$ in standard form (i.e. there exist some cyclic and separating vector $\eta \in \mathscr{H}$ for $\mathcal{N}$), the $H$-\emph{dual channel} $\delta'_{\beta,H} :  \mathcal{L}' \rightarrow \mathcal{N}'$ is defined by the matrix elements:
\begin{equation}
\left< h  \right| n \delta_{\beta,H}'(m') \left| h \right> = \left< H \right| \beta(n) m' \left| H \right> \, \qquad \forall n \in \mathcal{N}\,, m' \in \mathcal{L}_N' % \, \qquad H_N \in \mathcal{P}_{V_N \eta; \mathcal{L}_N}^\natural
\end{equation}
where $h \in \mathcal{P}_{\eta; \mathcal{N}}^\natural \subset \mathscr{H}$ is uniquely defined by the above equation and the choice that it lives in the natural self-dual cone (\ref{natural-cone}) \cite{araki1974some} associated to $\eta$ and $\mathcal{N}$.
This channel has nice properties, including the fact that it is faithful normal unital completely positive, see \cite{accardi1982conditional}
Given the same setup the \emph{Petz map} $\delta_{\beta,H}(\cdot) : \mathcal{L} \rightarrow \mathcal{N}$, is:
\begin{equation}
\delta_{\beta,H}(\cdot) = J_{\eta;\mathcal{N}} \delta_{\beta,H}'( J_{H; \mathcal{L}}\cdot J_{H; \mathcal{L}}) J_{\eta;\mathcal{N}}
\end{equation}
where $J$ is the modular conjugation operator, defined below in Section~\ref{sec:maxext}, \eqref{TTresults}.
\end{definition}

We use some of the properties in the proof. 

\subsection{Asymptotic Information Disturbance Tradeoff}

Given a code let us ask about what information in the bulk, accessible to some algebra $\mathcal{N} \subset \mathcal{B}(\mathscr{H})$, can be reconstructed from $\mathcal{L}_N \equiv \mathcal{M}_N(O)$.
In this case $\mathcal{L}'_N$ can be thought of as the errors in the sense of quantum error correction and any information contained in $\mathcal{N}$ that leaks to $\mathcal{L}'_N $ cannot be extracted without disturbance and so cannot be accessible to $\mathcal{L}_N$ ( $\lnot(4) \implies \lnot(1)$ in Theorem~\ref{thm:qec} ). Equivilently, all errors in $\mathcal{L}'_N$ cannot disturb $\mathcal{N}$ if all the information is available to $\mathcal{L}_N$ ($(1) \implies (4)$ in Theorem~\ref{thm:qec}).
Conversely if no errors in $\mathcal{L}'_N$ disturb $\mathcal{N}$, then the information in $\mathcal{N}$ can be extracted from $\mathcal{L}_N$
($(4) \implies (1)$ in Theorem~\ref{thm:qec}).
 These statements are well known and are very important for the general theory of quantum error correction.
 The new result that we establish here is a version of information-tradeoff that
holds for asymptotically isometric codes.

While we eventually have in mind the triple $(\mathcal{C},V_N, \gamma_N)$ 
with some $\mathcal{N} \subset \mathcal{B}(\mathscr{H})$ that can be reconstructed from $\mathcal{L}_N = \mathcal{M}_N(O)$, for the next Theorem~\ref{thm:main} and associated Definition~\ref{def:main} we actually only need
a sequence of bounded linear maps $V_N : \mathscr{H} \rightarrow \mathscr{K}_N$ such that $wo-\lim_N V_N^\dagger V_N = 1$ and:

\begin{theorem}
\label{thm:main}
\label{thm:qec}
Given a hyperfinite von Neumann algebra $\mathcal{N} \subset \mathcal{B}(\mathscr{H})$ with a cyclic vector $\eta \in \mathscr{H}$ then the following conditions pertaining to some sequence of von Neumann algebras $\mathcal{L}_N \subset \mathcal{B}( \mathscr{K}_N)$ are equivalent:
\begin{itemize}
\item[(1)] There exists a sequence of quantum channels 
$\beta_N : \mathcal{N} \rightarrow \mathcal{L}_N$ satisfying:
\begin{equation}
so-\lim_{N \rightarrow \infty} \left( \beta_N(n) V_N - V_N n  \right) = 0
\end{equation}
for all $n \in \mathcal{N}$.
\item[(2)] 
There exists a sequence of uniformly in $N$ pointwise bounded unital $^\star$-maps $\gamma_N : \mathcal{D} \rightarrow \mathcal{L}_N$
from a weakly dense unital $\star$-algebra $\mathcal{D} \subset \mathcal{N}$ with $\mathcal{D}'' = \mathcal{N}$, satisfying:
\begin{equation}
\label{so-d}
so-\lim_{N \rightarrow \infty} \left( \gamma_N(d) V_N - V_N d  \right) = 0
\end{equation}
for all $d \in \mathcal{D}$.
\item[(3)] For all unitaries $u \in \mathcal{N}$ and all $\rho \in \mathcal{B}(\mathscr{H})_\star^+$
\begin{equation}
\lim_{N \rightarrow \infty} \| \left( \rho \circ {\rm Ad}_{V_N} ( \cdot) - \rho_u \circ {\rm Ad}_{V_N} ( \cdot) \right)|_{\mathcal{L}_N'} \| = 0
\end{equation}
where $\rho_u(\cdot ) = \rho( u^\dagger \cdot u )$. 
\item[(4)] For all bounded 
sequences $m_{N}' \in \mathcal{L}_{N}'$ the following limit applies:
\begin{equation}
wo-\lim_{N \rightarrow \infty} [V_N^\dagger m_{N}' V_N, n ] =0
\end{equation}
for each $n \in \mathcal{N}$.
\end{itemize}

If any of the conditions (1-4) are satisfied and $\mathcal{N}$ is in standard form (there exists a cyclic and separating
vector $\eta$ for $\mathcal{N}$) then
$\beta_N$ in (1) can be chosen to be faithful and to fix $\eta$ uniformly: 
\begin{equation}
\label{uniform1}
\lim_{N \rightarrow \infty} \| \omega_{V_N \eta} \circ \beta_N -  \omega_\eta |_{\mathcal{N}} \| = 0
\end{equation}
Furthermore if the sequence $\varPhi_N$ is cyclic and separating for $\mathcal{L}_N$ and represents $\eta$
on the code, then 
the $\varPhi_N$-dual channel for $\beta_N$, $\alpha_N' \equiv \delta'_{ \beta_N, \varPhi_N } : \mathcal{L}_N' \rightarrow \mathcal{N}'$ (Definition~\ref{def:dual}) satisfies:
\begin{equation}
\label{alphaV}
wo-\lim_{N \rightarrow \infty} ( \alpha'_N(m_N') - V_N^\dagger m_N' V_N) =0
\end{equation}
for all bounded sequences $m_N' \in \mathcal{L}_N'$. 
\end{theorem}
\begin{proof}
See Section~\ref{sec:small}. 
\end{proof}

\begin{definition}[Reconstructible]
\label{def:main}
Consider an algebra $\mathcal{N} \subset \mathcal{B}(\mathscr{H})$ assumed hyperfinite
with a cyclic vector. We say that $\mathcal{N}$ is \emph{reconstructible} from $\mathcal{L}_N$ if it satisfies any one of the equivalent conditions in (1-4) of Theorem~\ref{thm:main}.
In short $\mathcal{N} \rightsquigarrow \mathcal{L}_N$. If $V_N$ is associated to an asymptotically isometric code, Definition~\ref{def:code},
with $\mathcal{L}_N = \mathcal{M}_N(O)$ for some $O$ then we will simply write $\mathcal{N} \rightsquigarrow O$. 

If additionally there is a vector $\eta \in \mathscr{H}$ represented on the code by the sequence $\varPhi_N$ and such that both $\eta$ and $\varPhi_N$ are cyclic and separating for $\mathcal{N}$ and $\mathcal{L}_N$ respectively, then we say the reconstruction is standard and write $\mathcal{N} \rightarrowtail \mathcal{L}_N$
or $\mathcal{N} \rightarrowtail  O$ if $\mathcal{L}_N = \mathcal{M}_N(O)$. 
\end{definition}

\begin{remark}
The fact that (2) implies the other conditions is remarkable. In particular one can start with a map $\gamma_N$ from an asymptotically isometric code, with minimal constraints, and end up
with an approximation to it - $\beta_N$ from condition (1) that is a normal unital completely positive map. It is somewhat like a double dual map: applying the dual map defined in Definition~\ref{def:dual} twice.  Such a procedure allows one to pass from a (completely) positive non-normal map to a (completely) positive normal map \cite{accardi1982conditional}. However our result here is seemingly considerably stronger since $\gamma_N$ need not be (completely) positive and only needs to be defined on a dense subspace, yet the result is a normal unital completely positive map. This comes at the expense of 
only approximating $\gamma_N$ when  restricted to $\mathcal{D}$ with:
\begin{equation}
\label{apsame}
so-\lim_{N \rightarrow \infty} (\gamma_N(d) - \beta_N(d) )V_N = 0\,, \qquad \forall d \in \mathcal{D}
\end{equation}
 It seems reasonable also that the ability to extend from $\gamma_N$ to $\beta_N$ is dependent on the hyperfinite nature of $\mathcal{N}$. 
\end{remark}

\begin{remark}
\eqref{apsame} is an important result for these codes, since it implies that we can have multiple reconstruction maps and they need not be the same. In particular cyclicity  of $\mathcal{M}_N(O)'$ for $V_N \eta$ (assuming a standard code) was used previously \cite{Kelly:2016edc,Faulkner:2020hzi}, for exact codes, to show the maps must be exactly the same. However we cannot apply such cyclicity argument through the strong convergence statement above since operator norms $m'$ diverge when approximating arbitrary states in the Hilbert space. The best we can do, is that for any $\epsilon$ there is some $N_\star$ such that:
\begin{equation}
\left\| (\gamma_N(c) - \beta_N(c) ) \left| \psi \right> \right\| < \epsilon \qquad \psi \in \mathscr{R}_N^{(1)}(O)  \subset \mathscr{K}_N\, \qquad \forall N > N_\star
\end{equation}
where $\mathscr{R}_N^{(1)}(O) = \{ m' V_N \left| \eta \right>, m' \in \mathcal{M}_N(O)' : \| m' \| \leq 1 \}$ is a \emph{subset} of the Hilbert space. 
\end{remark}

\begin{remark}
The non cyclic case can be treated after projecting the code further to some smaller Hilbert space. For example, this is treated
for the small codes discussed in Section~\ref{sec:small}. We do not develop this case here - in this paper it is used as a proof technique.
Standard asymptotically isometric codes will always give rise to standard reconstructions when $O$ is local. We expect this
will always be the case for holographic theories. If $O= K$, the entire Lorentzian
cylinder or equivalent, then the reconstruction need not be standard. We still have a nice theorem that applies in this case, although
in Section~\ref{sec:maxext} we will only work with standard reconstructions. 
\end{remark}

We now explore some immediate consequences of Definition~\ref{def:main}. These results are analogous to some of the results proven in \cite{Faulkner:2020hzi} for exact
error correcting codes. 
We have some simple results assuming henceforth the existence of an underlying asymptotically isometric code Definition~\ref{def:code}:

\begin{lemma}
Given an asymptotically isometric code then:  \label{older}
\begin{itemize}
\item[(i)] For any non-empty $O$ the causal wedge $\mathcal{C}(O)$ is reconstructible from $\mathcal{M}_N(O)$. That is $\mathcal{C}(O) \rightsquigarrow O$.
\item[(ii)] If $ \mathcal{N}  \rightsquigarrow O_1$ then $\mathcal{N} \rightsquigarrow  O_2$ for all $O_1 \subset O_2$. 
\item[(iii)] Suppose that $\mathcal{N}_1 \subset \mathcal{N}$ has a cyclic vector, then $\mathcal{N} \rightsquigarrow O \implies \mathcal{N}_1 \rightsquigarrow O$. 
\end{itemize}
\end{lemma}

%\begin{lemma}
%\label{older2}
%\begin{itemize}
%\item[(i)] For any non-empty $O$ the causal wedge $\mathcal{C}(O)$ is reconstructible from $\mathcal{M}_N(O)$. That is $\mathcal{C}(O) \rightsquigarrow O$.
%\item[(ii)] If $ \mathcal{N}  \rightsquigarrow O_1$ then $\mathcal{N} \rightsquigarrow  O_2$ for all $O_1 \subset O_2$. 
%\item[(iii)] Suppose that $\mathcal{N}_1 \subset \mathcal{N}$ has a cyclic vector, then $\mathcal{N} \rightsquigarrow O \implies \mathcal{N}_1 \rightsquigarrow O$. 
%\end{itemize}
%cursed 2.
%\label{older2}
%\end{lemma}

\begin{proof}
(i) By definition with condition (2) in Theorem~\ref{thm:main}. (ii) Using condition (1) of Theorem~\ref{thm:main} and noting that $\mathcal{M}_N(O_1) \subset \mathcal{M}_N(O_2)$.  (iii) We restrict $\beta_N$ to $\mathcal{N}_1$ in condition (1). 
\end{proof}

We can also use different reconstruction quantum channels at the same time (this was proven for finite dimensional approximate codes in \cite{Cotler:2017erl}):
\begin{lemma}
Suppose that $\mathcal{N}_ i  \rightsquigarrow O_i$, for $i = 1, \ldots k$, then for any $\beta_N^i : \mathcal{N}^i \rightarrow \mathcal{M}_N(O_i)$ given in (1)
of Theorem~\ref{thm:main} for each of these regions $O_i$ we have:
\begin{equation}
wo-\lim_{N \rightarrow \infty} V_N^\dagger \beta^1_N(n_1) \beta^2_N(n_2) \ldots \beta^k_N(n_k) V_N =  n_1 n_2 \ldots n_k
\end{equation}
for any $n_i \in \mathcal{N}_i$. 
\end{lemma}
\begin{proof}
Follows from Lemma~\ref{lem:seq}. 
\end{proof}

A basic input to the causality results we discuss below is:

\begin{lemma}[Causality]
\label{lem:caus}
If $\mathcal{N}_1$ is reconstructible from $\mathcal{M}(O_1)$ and
if $\mathcal{N}_2'$ is reconstructible from $\mathcal{M}(O_2')$ and $O_1 \subset O_2$
then $\mathcal{N}_1 \subset \mathcal{N}_2$. 
\end{lemma}
\begin{proof}
Fix some $n_1 \in \mathcal{N}_1$. Use condition (1) in Theorem~\ref{thm:main} and the operators $\beta_N^1(n_1) \in \mathcal{M}(O_1) \subset \mathcal{M}(O_2)$
and apply it to condition (4) for $\mathcal{N}_2'$. That is:
\begin{equation}
0= wo-\lim_{N \rightarrow \infty} [ V_N^\dagger \beta_N^1(n_1) V_N, n_2']  = [n_1,n_2']
\end{equation}
for all $n_2' \in \mathcal{N}_2'$. Thus $\mathcal{N}_1 \subset \mathcal{N}_2$. 
\end{proof}

\begin{lemma}
\label{lem:cons}
Consider a local region $O$. If $\mathcal{N}_1 \rightsquigarrow O$ and $\mathcal{N}_2 \rightsquigarrow O$ then $\mathcal{N}_1 \vee \mathcal{N}_2  \rightsquigarrow O$.
\end{lemma}
\begin{proof} 
As part of the assumptions there exists $\eta$ cyclic for $\mathcal{N}_1$ which is then cyclic for  $\mathcal{N}_1 \vee \mathcal{N}_2$. 

Consider $\mathcal{D}$ as the set of finite products of $\mathcal{N}_1$ and $\mathcal{N}_2$. This is a weakly dense unital $^\star$-subalgebra of $\mathcal{N}_1 \vee \mathcal{N}_2$. From (4) of Theorem~\ref{thm:main} for each $\mathcal{N}_1$ and $\mathcal{N}_2$ we find that for all uniformly bounded sequences $m_N' \in \mathcal{M}_N(O)'$ then:
\begin{equation}
wo-\lim_N [ V_N^\dagger m_N' V_N , d ] =0 \qquad \forall d \in \mathcal{D}
\end{equation}
by repeated applications of the formula $[X,n_1 n_2] = [X,n_1] n_2 + n_1 [X,n_2]$. Then apply the condition (4') of Theorem~\ref{thm:qec} given in Section~\ref{sec:p} to conclude that
$\mathcal{N}_1 \vee \mathcal{N}_2 \rightsquigarrow O$.
\end{proof}

As a corollary to Lemma~\ref{lem:caus} we learn that all reconstructible algebras $\mathcal{N} \rightsquigarrow O $, %with $O'$ non-empty,
satisfy $\mathcal{N} \subset \mathcal{C}(O')'$. 
We now generalize our definition of reconstructible:
\vspace{.5cm}

\noindent {\bf Definition~\ref{def:main}'} (Reconstructible')
Suppose that we no longer demand that $\mathcal{N}$ has a cyclic vector. 
Then we can define the notion of reconstructible
with an underlying assymptotic code, Definition~\ref{def:code}, as $\mathcal{N} \rightsquigarrow O$ iff $\mathcal{N} \vee \mathcal{C}(O) \rightsquigarrow O$. 
Finally we we always write $\mathbb{C} 1 \rightsquigarrow \emptyset$ as a trivial reconstruction (recall that $K' = \emptyset$.)
\vspace{.5cm}

This
definition is consistent with the previous definition for algebras with a cyclic vector because of Lemma~\ref{lem:cons}
and Lemma~\ref{older} (i,iii). In this new case we can consider very \emph{small} algebras. It seems these need to be ``dressed'' using the causal wedge in order to apply
the reconstruction theorems. 

Given the above it is natural to define the \emph{reconstruction wedges} \cite{Dong:2016eik,Hayden:2018khn}:
\begin{equation}
\mathcal{E}(O) =\bigvee_{ \mathcal{N} \ensuremath{\rightsquigarrow} O } \mathcal{N} \, ,
\qquad  \mathcal{E}(O') = \bigvee_{ \mathcal{N}' \rightsquigarrow O' } \mathcal{N}'
\end{equation}
with $\mathcal{E}(O) \rightsquigarrow O$  and $\mathcal{E}(O') \rightsquigarrow O'$. In particular for a standard asymptotically isometric code and local $O$ then
these reconstruction wedges become standard (since by assumption there is a vector $\eta$ cyclic and separating for $\mathcal{C}(O)$ and $\mathcal{C}(O')$):
$\mathcal{E}(O) \rightarrowtail O$  and $\mathcal{E}(O') \rightarrowtail O'$.

From this we learn:
\begin{lemma}
We have the following results for the reconstruction wedge:
\begin{itemize}
\item[(i)] \emph{For local $O$ the causal wedge is contained in the reconstruction wedge} and thus the reconstruction wedge is non-trivial:
\begin{equation}
\mathcal{C}(O) \subset \mathcal{E}(O) \subset \mathcal{C}(O')'
\end{equation}
\item[(ii)] \emph{Isotony}:
\begin{equation}
\mathcal{E}(O_1) \subset \mathcal{E}(O_2) \, \qquad O_1 \subset O_2
\end{equation}
\item[(iii)] \emph{Einstein causality}:
\begin{equation}
\mathcal{E}(O) \subset \mathcal{E}(O')'
\end{equation}
\end{itemize}
\end{lemma}
\begin{proof}
(i) Immediate. (ii) From Lemma~\ref{older} (ii). (iii)  From Lemma~\ref{lem:caus} with $O_1 = O_2 = O$
and $\mathcal{N}_1 = \mathcal{E}(O)$ and  $\mathcal{N}_2' = \mathcal{E}(O')$. 
\end{proof}

This is as far as we can go for a general asymptotically isometric code, without any further assumptions.
There are known counter examples to always having a maximal Haag dual extension (at least in finite dimensions \cite{Hayden:2018khn,Akers:2020pmf}) and so we don't expect
general codes to have such extensions. There is however something special about holographic theories with a semi-classical gravity dual, that would imply such Haag dual extensions
always exist for local $O$. We do not have a full understanding on exactly what is special about holographic theories
- although we make some speculative comments in the Discussion~\ref{sec:bulklocality}. We can however write down equivalent conditions characterizing such maximal extensions. Such conditions are well known to be true of holographic theories
with a weakly coupled gravity limit \cite{Jafferis:2015del,Faulkner:2017vdd}. 
We explore these equivalent conditions next. 

\subsection{Maximal extensions and complementary recovery}
\label{sec:maxext}

If we assume the extension $\mathcal{E}(O)$ satisfies Haag duality \cite{haag2012local}, or equivalently complementary recovery \cite{Harlow:2016vwg}, we define:
\begin{definition}
A reconstruction wedge for $O$ is called an \emph{entanglement wedge} if it satisfies Haag duality:
\begin{equation}
\mathcal{E}(O) = \mathcal{E}(O')'
\end{equation} 
More generally $\mathcal{N}$ is an entanglement wedge for $\mathcal{L}_N$
if $\mathcal{N} \rightsquigarrow \mathcal{L}_N$ and  $\mathcal{N}' \rightsquigarrow \mathcal{L}_N'$.
A \emph{standard entanglement wedge} is an entanglement wedge where the reconstructions are standard  $\mathcal{N} \rightarrowtail \mathcal{L}_N$ and  $\mathcal{N}' \rightarrowtail \mathcal{L}_N'$.

An asymptotically isometric code with \emph{complementary recovery} is a code where for all local boundary regions $O$ the reconstruction wedge is an entanglement wedge.
\end{definition}
\begin{remark}
In this case isotony simply becomes the statement of \emph{entanglement wedge nesting} \cite{Wall:2012uf}. 
Entanglement wedges, for standard asymptotically isometric codes with local $O$, are standard. 
\end{remark}

For the global algebra we have:
\begin{lemma}
\label{lem:global}
Suppose that $\mathcal{M}_N' \equiv \mathcal{M}_N(K)' = \mathbb{C} 1$ then $\mathcal{E}(K) = \mathcal{B}(\mathscr{H}) = \mathcal{E}(\emptyset)'$. 
That is the extension is maximal and $\mathcal{E}(K)$ is a (non-standard) entanglement wedge. 
\end{lemma}
\begin{proof}
Note that condition (4) of Theorem~\ref{thm:main} is automatically satisfied for $K$ in this case for all operators $\mathcal{N} =  \mathcal{B}(\mathscr{H})$. 
\end{proof}

We now prove some results about complementary recovery for one wedge. 
Note we have:
\begin{lemma}
\label{lem:ab}
If $\mathcal{N}$ is a standard entanglement wedge for $\mathcal{L}_N$ then there exists
faithful quantum channels 
$\beta_N :\mathcal{N} \rightarrow \mathcal{L}_N,\,\, \alpha_N :  \mathcal{L}_N \rightarrow \mathcal{N} $ and
$\beta_N' :\mathcal{N}' \rightarrow \mathcal{L}_N', \,\, \alpha_N' :  \mathcal{L}_N' \rightarrow \mathcal{N}' $
as defined in Theorem~\ref{thm:main} for $\mathcal{N}$ and $\mathcal{N}'$, and
 such that: 
\begin{align}
wo-\lim_{N \rightarrow \infty} \alpha_N \circ \beta_N (n) = n \, \qquad \forall n \in \mathcal{N}  \\
wo-\lim_{N \rightarrow \infty} \alpha_N' \circ \beta_N' (n') = n' \, \qquad \forall n' \in \mathcal{N}'
\end{align}
\end{lemma}
\begin{proof}
From (1) and (4) of Theorem~\ref{thm:main} as well as the last part of that theorem applicable to standard entanglement wedges. 
\end{proof}

The next theorem we prove is the JLMS condition \cite{Jafferis:2015del}.
To discuss this we need to introduce Tomita-Takeskai modular operators and relative entropy. For more details see \cite{Witten:2018lha,Ceyhan:2018zfg,Faulkner:2020iou}, and we follow the notation in the later work. For simplicity we state our definitions for algebras on the code subspace $\mathscr{H}$ - there are also equivalent definitions for algebras on $\mathcal{K}_N$.
The modular data for $\eta \in \mathscr{H}$,
and an algebra $\mathcal{N} \subset \mathcal{B}(\mathscr{H})$ is denoted $J_{\eta;\mathcal{N}}, \Delta_{\eta;\mathcal{N}}$. For the most part we will assume that $\eta$ is cyclic and separating for $\mathcal{N}$, although similar definitions work when this is not the case. In this case $\Delta_{\eta;\mathcal{N}}$ is a densely defined closed operator
and $J_{\eta;\mathcal{N}}$ is antiunitary. 
We have $J_{\eta;\mathcal{N}} \Delta_{\eta;\mathcal{N}}^{1/2} n \left|\eta \right> = n^\dagger \left| \eta \right>$ which arises from the defining equation. 
Then the important results of Tomita-Takesaki theory are:
\begin{align}
\label{TTresults}
J_{\eta;\mathcal{N}} n J_{\eta;\mathcal{N}} \in \mathcal{N}' \qquad \forall n \in \mathcal{N}  \\
\Delta_{\eta;\mathcal{N}}^{-is} n \Delta_{\eta;\mathcal{N}}^{is}
 \equiv \sigma_{\varrho}^s(n) \in \mathcal{N} \qquad \forall n \in \mathcal{N}
\end{align}
where $\varrho = \omega_{\eta}|_{\mathcal{N}}$
and it is possible to show that the later modular flow only depends on the state $\varrho$ and not any particular purification of it in $\mathscr{H}$ (working in the standard representation.) In the proofs we need some more details on modular theory, such as positive natural cones and relative modular operators. See Section~\ref{proof:jlms} for some brief discussion.

We also need to introduce relative entropy.
The standard definition of relative entropy that works for normal states $\rho,\sigma \in \mathcal{N}_\star$ of a general von Neumann algebra $\mathcal{N}$ was given by Araki using the relative modular operator \cite{araki1976relative,araki1977relative}. 
We need a definition of relative entropy when the states are not normalized. In this case:
\begin{equation}
\label{eq:correctnorm}
S_{\rm rel}(\rho | \sigma) = \rho(1)\left( S_{\rm rel}( \rho/\rho(1) | \sigma/\sigma(1)) + \ln \rho(1)/ \sigma(1) + \sigma(1)/\rho(1) - 1 \right)\, \in \mathbb{R}_{\geq 0} \cup \infty
\end{equation}
where we use the standard Araki definition for the normalized states. We quote some important properties of $S_{\rm rel}(\rho | \sigma)$ that we need. There is monotonicity of relative entropy under the action of a quantum channel $\alpha$:
\begin{equation}
S_{\rm rel}(\rho | \sigma) \geq S_{\rm rel}(\rho \circ \alpha | \sigma \circ \alpha )
\end{equation}
which still applies even for non-normalized states.
There is lower semi-continuity of relative entropy:
\begin{equation}
\label{lowersemi}
\liminf_{n \rightarrow \infty} S(\rho_n |\sigma_n ) \geq S(\rho | \sigma)
\end{equation}
where $\rho_n \rightarrow \rho$ and $\sigma_n \rightarrow \sigma$ in the pointwise sense. This is a very weak sense of convergence (\ref{ultraweak}). It is convergence with respect to the weak topology on $\mathcal{N}_\star$
when thought of as a Banach space. Such a weak form of lower semi-continuity is crucial to our results and is consequence of Kosaki's variational formula for relative entropy, see \cite{ohya2004quantum,kosaki1986relative}.

Finally there is the Pinsker's inequality which becomes:
\begin{equation}
\label{eq:correctp}
S_{\rm rel}(\rho | \sigma)  \geq \frac{\rho(1)}{2} \left( \| \rho/\rho(1) - \sigma/\sigma(1) \|^2 + (1 - \exp(-| \ln \sigma(1)/\rho(1) |))^2  \right)
\end{equation}
so that the vanishing of the relative entropy entails $\rho= \sigma$. We allow such non-normalized states in \eqref{se} and monotonicity still applies
for unital completely positive maps since the normalization is preserved by the channel.

Relative entropy is however a little tricky to work with because it need not be finite and in particular it may
not be continuous under limits. This makes it particularly hard to make statements about relative entropy under the $\lim_{N\rightarrow \infty}$. One way to fix this is to define the smoothed relative entropy for two states $\rho,\sigma$
\begin{equation}
\label{se}
S_{\epsilon}( \rho | \sigma ) = \inf_{ \tilde{\rho},\tilde{\sigma} : \| \rho - \tilde{\rho} \| , \| \sigma - \tilde{\sigma} \| \leq \epsilon}
S_{\rm rel}( \tilde{\rho} | \tilde{\sigma} )
\end{equation}
This quantity is monotonic under the action of a quantum channel $\alpha$:
\begin{equation}
S_{\epsilon}( \rho | \sigma)  \geq \inf_{ \tilde{\rho},\tilde{\sigma} : \| \rho - \tilde{\rho} \| , \| \sigma - \tilde{\sigma} \| \leq \epsilon}
S_{\rm rel}( \tilde{\rho} \circ \alpha | \tilde{\sigma}  \circ \alpha )  \geq S_{\epsilon}( \rho \circ \alpha | \sigma \circ \alpha) 
\end{equation}
since $: \| \rho - \tilde{\rho} \| \geq  \| \rho \circ \alpha - \tilde{\rho} \circ \alpha \|$. 
We will use this to characterize the existence of an entanglement wedge, as in JLMS.

It is sometimes convenient to work with states satisfying $\rho \leq \lambda \sigma$ where $\lambda \geq 1$ 
for normalized $\rho,\sigma \in \mathcal{N}^+_\star$. By definition this condition is the same as the requirement that the max relative entropy is finite $S_{\rm max}(\rho|\sigma) \leq \ln \lambda$ which implies that the relative entropy $S_{\rm rel}$ is finite. For cyclic and separating $\eta \in \mathscr{H}$ inducing $\sigma$ when restricted to $\mathcal{N}$ this condition implies that $\rho$ can be represented by $\psi \in \mathscr{H}$ with $\psi = n' \eta$ and $\| n' \|^2 \leq \lambda$. 

When considering states on a larger algebra (say on $\mathcal{B}(\mathscr{H})$) we will use the following notation for relative entropy:  $S_{{\rm rel}/\epsilon}(\rho|\sigma;\mathcal{A}) \equiv S_{{\rm rel}/\epsilon}(\rho|_{\mathcal{A}}|\sigma_{\mathcal{A}})$  for $\rho,\sigma \in \mathcal{B}(\mathscr{H})_\star$
and $\mathcal{A}\subset \mathcal{B}(\mathscr{K})$ a von Neumann subalgebra.

\begin{theorem}[JLMS]
\label{thm:JLMS}
\label{thm:JLMSplus}
Consider a hyperfinite von Neumann algebra $\mathcal{N} \subset \mathcal{B}(\mathscr{H})$
and von Neumann algebra $\mathcal{L}_N \subset \mathcal{B}(\mathscr{K}_N)$ with a vector $\eta \in \mathscr{H}$ represented by $\varPhi_N$
on the code and such that $\eta$ (resp. $\varPhi_N$) is cyclic and separating for $\mathcal{N}$ (resp.
$\mathcal{L}_N$) then the following conditions are equivalent:
\begin{itemize}
\item[(i)] $\mathcal{N}$ is a standard entanglement wedge for $\mathcal{L}_N$. 
\item[(ii)] For all pairs of normalized states $\rho, \sigma \in   \mathcal{B}(\mathscr{H})_\star^+$
with $\rho|_{\mathcal{N}} \leq \lambda \sigma|_{\mathcal{N}}$ for some $1 \leq \lambda < \infty$,
there exists a monotonic function $\epsilon_N \rightarrow 0$ such that the following limit applies:
\begin{equation}
\label{rhv}
 \lim_{N \rightarrow \infty} S_{\delta_N}( \rho \circ {\rm Ad}_{V_N} |\sigma \circ {\rm Ad}_{V_N}  ; \mathcal{L}_N)  = S_{\rm rel}(\rho|\sigma; \mathcal{N} )
\end{equation}
for all sequences $\delta_N \geq \epsilon_N$ with $\delta_N \rightarrow 0$. 
\item[(iii)] $\mathcal{N} \rightarrowtail \mathcal{L}_N$ 
and one of the following applies:
\begin{itemize}
\item[(iii)$_a$] The modular operator acts covariantly:
\begin{equation}
so-\lim_{N \rightarrow \infty} 
\Delta^{is}_{\varPhi_N %V_N \eta
; \mathcal{L}_N} V_N -
V_N \Delta^{is}_{\eta;\mathcal{N}} = 0 \,,\qquad s \in \mathbb{R}
\end{equation}
\item[(iii)$_b$]
The modular automorphism group for $\eta$ acts covariantly through the code:
\begin{equation}
\label{bbbff2}
so-\lim_{N \rightarrow \infty}  \sigma_{\omega_{\varPhi_N}
|_{\mathcal{L}_N}}^s( \beta_N(n)) V_N
- V_N \sigma_{\omega_\eta|_\mathcal{N}}^s(n) =0 \qquad \forall n \, \in \mathcal{N} \qquad \forall \, s\in \mathbb{R}
\end{equation}
\item[(iii)$_c$] The modular conjugation acts covariantly: 
\begin{equation}
so-\lim_{N \rightarrow \infty} 
J_{\varPhi_N;
\mathcal{L}_N} V_N -
V_N J_{\eta;\mathcal{N}} = 0
\end{equation}
\end{itemize}
\end{itemize}
If any of these conditions is satisfied then
(ii) applies for \emph{all} pairs of normalized states $\rho,\sigma \in \mathcal{B}(\mathscr{H})_\star^+$ (without  the condition of finite max relative entropy) and for all  $\rho,\sigma \in \mathcal{N}_\star^+$
we have:
\begin{equation}
\lim_{N \rightarrow \infty} S_{\rm rel}( \rho \circ \alpha_N | \sigma \circ \alpha_N) = S_{\rm rel}(\rho|\sigma)
\end{equation}
Also $(iii)$ applies for any vector $\psi \in \mathscr{H}$ which is cyclic and separating for $\mathcal{N}$ and any sequence $\Psi_N$  representing $\psi$ on the code, which are cyclic and separating for $\mathcal{L}_N$. 
\end{theorem}
\begin{proof}
See Section~\ref{proof:jlms}. 
\end{proof}

We have phrased  conditions (ii) and (iii) so they are very weak. This is to make it easy to prove existence of an entanglement wedge by utilizing this theorem. The more general statements that follow from these conditions are in the postamble of the theorem.  It was recently pointed out in \cite{Kudler-Flam:2022jwd}, that given an entanglement wedge there could be states $\rho,\sigma \in \mathcal{B}(\mathscr{H})_\star^+$ on the code with infinite relative entropy $S_{\rm rel}(\rho|\sigma;\mathcal{N}) = \infty $ while $S_{\rm rel}(\rho \circ {\rm  Ad}_{V_N}|\sigma \circ {\rm  Ad}_{V_N})$ is finite. Seemingly violating the JLMS condition.
The theorem above tells us that in this case the large $N$ limit of such relative entropies must diverge. This resolution was discussed in \cite{Kudler-Flam:2022jwd}.

\section{Proofs}

\label{sec:small}
\label{sec:p}

Our first task is to prove the QEC Theorem~\ref{thm:qec}. In addition to the statements in the theorem we add two more equivalent statements
that we will prove at intermediate steps. 
\begin{itemize}
\item[(3')] For all unitaries $u \in \overline{\mathcal{D}}^{\| \cdot \|}$, where $\mathcal{D}$ is weakly dense unital $\star$-subalgebra $\mathcal{D} \subset \mathcal{N}$, and all $\rho \in \mathcal{B}(\mathscr{H})_\star^+$ then:
\begin{equation}
\lim_{N \rightarrow \infty} \| \left( \rho \circ {\rm Ad}_{V_N} ( \cdot) - \rho_u \circ {\rm Ad}_{V_N} ( \cdot) \right)|_{\mathcal{L}_N'} \| = 0
\label{3prime}
\end{equation}
where $\rho_u(\cdot ) = \rho( u^\dagger \cdot u )$. 
\item[(4')] For all bounded 
sequences $m_{N}' \in \mathcal{L}_{N}'$ the following limit applies:
\begin{equation}
\label{4prime}
wo-\lim_{N \rightarrow \infty} [V_N^\dagger m_{N}' V_N, d] =0
\end{equation}
for each $d \in \mathcal{D}$ (with $\mathcal{D}$ as in (3')).
\end{itemize}

\subsection{Finite dimensional preliminaries} 
\label{finited}

Consider a sequence of bounded operators $V_N \in \mathcal{B}(\mathscr{H}, \mathscr{K}_N)$ that converges to an isometry in the sense that $wo-\lim_{N \rightarrow \infty} V_N^\dagger V_N \rightarrow 1$ and hence with $\|V_N^\dagger V_N \| < C$, see Lemma~\ref{lem:C}. Assume that $\mathcal{N} \subset \mathcal{B}(\mathscr{H})$ 
is hyperfinite sometimes referred to as approximately finite dimensional and $\mathscr{H}$ comes equipped with a vector $\eta$ cyclic and separating for $\mathcal{N}$. 
This is almost the context of Theorem~\ref{thm:main}, albeit with the additional setup assumption that
$\eta$ is cyclic and separating. We will take this to be the case for this Section~\ref{finited}. We will return to the non-separating case in a later Section after \eqref{aftnonsep}. 

 We aim to use these results to prove Theorem~\ref{thm:main} in the cyclic and separating case. We first must introduce some machinery.\footnote{We thank S. Hollands for pointing out the ease with which some theorems can be proven in the hyperfinite case. See for example Longo \cite{longo1978simple} where the basic theorems of Tomita-Takesaki are established with this method.} 

The hyperfinite condition implies the existence of
a sequence of type-I$_{2^q}$ von Neumann factors (these can always chosen to be factors, see \cite{elliott1976equivalence}) with $\mathcal{N}_q \subset \mathcal{N}_{q+1} \subset \mathcal{N}$ (with common unit) such that:
\begin{equation}
\label{hypfin}
\mathcal{N} = ( \cup_q \mathcal{N}_q )''
\end{equation}
For fixed $q$ we can work on the new code subspace determined by $\mathcal{N}_q \left| \eta \right> = \pi_q'
\in \mathcal{N}_q'$ . Define the bounded operator: 
\begin{equation}
V_{N,q} \equiv V_N \pi_q' : \mathscr{H}_q ( \equiv \pi_q' \mathscr{H})  \rightarrow \mathscr{K}_N
\end{equation}
Define the reduced algebra $\widetilde{\mathcal{N}}_q = \mathcal{N}_q \pi_q'$ and $\widetilde{\mathcal{N}}_q' = \pi_q' \mathcal{N}_q'  \pi_q'$ which act
on the Hilbert space $\mathcal{H}_q$ with unit $1_q = \pi_q'$. These algebras are factors and additionally $\pi_q' \eta = \eta$ is cyclic and separating for $\widetilde{\mathcal{N}}_q$
as can be easily confirmed.  Hence the von Neumann algebra $\widetilde{\mathcal{N}}_q $ represented on $\mathscr{H}_q$ is in standard form. Thus
$\mathscr{H}_q \cong L_q \otimes R_q$ where $L_q,R_q$ are finite dimensional Hilbert spaces of dimension $2^q$ and where:
 \begin{equation}
 \widetilde{\mathcal{N}}_q \cong  \mathcal{B}(L_q) \otimes 1_{R_q} \qquad  \widetilde{\mathcal{N}}'_q \cong 1_{L_q} \otimes \mathcal{B}(R_q)
 \end{equation}
Define the conditional expectation:
\begin{equation}
\label{LRq}
P_q (\cdot) =  \frac{1}{{| L_q| }}1_{L_q} \otimes {\rm Tr}_{L_q} (\cdot)\qquad \, P_q : \mathcal{B}(\mathscr{H}_q) \rightarrow \widetilde{\mathcal{N}}'_q 
\end{equation}

We need the isomorphism $\Phi_q :\mathcal{N}_q \rightarrow \widetilde{\mathcal{N}}_q $ defined via:
\begin{equation}
\label{theisom}
\Phi_q(n)= n \pi_q'
\end{equation}
It is injective due to the separating property of $ \left| \eta \right>$ ($0 =  \Phi_q(n)\left| \eta \right> = n \left| \eta \right>
\implies n =0$) and this implies that $\Phi_q$ is an isomorphism.
Define the inclusion $\iota_q = \Phi_q^{-1} : \widetilde{\mathcal{N}}_q \rightarrow \mathcal{N} $.

Finally given the cyclic and separating $\eta$ for $\mathcal{N}$ we can construct the Petz map, Definition~\ref{def:dual}, associated to the inclusion $\iota_q :\widetilde{\mathcal{N}}_q \rightarrow \mathcal{N}$ (since $\widetilde{\mathcal{N}}_q$ is in standard form):
\begin{equation}
\label{Eq}
\mathcal{E}_{q,\eta} \equiv \delta_{\iota_q,\eta}: \mathcal{N} \rightarrow \widetilde{\mathcal{N}}_q 
\end{equation}
This is a faithful quantum channel that satisfies $\omega_\eta \circ \iota_q \circ \mathcal{E}_{q,\eta}  = \omega_\eta$ \cite{accardi1982conditional}. 

We also need the following convergence statement:
\begin{lemma}[Hiai, Tsukada \cite{hiai1984strong}]
\label{lem:gencond}
For all $ n \in \mathcal{N}$ then:
\begin{equation}
so-\lim_{q \rightarrow \infty} \iota_q \circ \mathcal{E}_{q,\eta}(n) = n
\end{equation}
\end{lemma}

Apply the polar decomposition:
\begin{equation}
V_{N,q} = \widetilde{V}_{N,q} (V_{N,q}^\dagger  V_{N,q} )^{1/2}
\end{equation}
to define the partial isometry $\widetilde{V}_{N,q}$ we have:

\begin{lemma}
\label{lemma:small}
Assuming weak operator convergence of $V_N^\dagger V_N$ to $1$ then
at fixed $q$ the small codes based on $V_{N,q}$ satisfy:
\begin{equation}
\lim_{N \rightarrow \infty} \| \widetilde{V}_{N,q} - V_{N,q} \| = 0
\end{equation}
and $\widetilde{V}_{N,q}$ is an isometry for sufficiently large $N$.
\end{lemma}
\begin{proof}
We apply the weak convergence condition:
\begin{equation}
wo-\lim_{N \rightarrow \infty} V_N^\dagger V_N \rightarrow 1
\end{equation} 
to the vectors in $\mathscr{H}_q$. Define the operator $(V_{N,q})^\dagger V_{N,q} - 1_q = Z_q$. We have the relation between the Hilbert-Schmidt norm and the operator norm:
\begin{equation}
\label{HSON}
\| Z_q \|^2 \leq {\rm Tr} Z_q^\dagger Z_q  = \sum_{\alpha,\beta} |\left< \alpha \right| Z_q \left| \beta \right>|^2 
\leq 2^{4q} \max_{\alpha, \beta}  |\left< \alpha \right| Z_q \left| \beta \right>|^2
\end{equation}
where $\alpha$ labels an orthonormal basis for $\left| \alpha \right> \in \mathscr{H}_q$.
Now for all $\epsilon$ we can pick an $N(\alpha,\beta,\epsilon)$ such that $ | \left< \alpha \right| Z_q \left| \beta \right> | < \epsilon$ for $N > N(\alpha,\beta,\epsilon)$. Set:
\begin{equation}
\label{maxim}
N_1(\delta) \equiv \max_{\alpha,\beta} N(\alpha,\beta, 2^{-2 q}\delta)
\end{equation}
Now if $V_{N,q}^\dagger V_{N,q}$ has a non-zero kernel with projection $\Pi$ to the kernel then:
\begin{equation}
\| Z_q \|^2 = \| V_{N,q}^\dagger V_{N,q} - 1_{q} \|  \geq \| \Pi (V_{N,q}^\dagger V_{N,q} - 1_{q}) \Pi \|   = \| \Pi \| = 1
\end{equation} 
whenever $\Pi$ is non-trivial. So choosing $\delta < 1$ then the kernel must vanish  for $N > N_1(\delta)$ by a contradiction argument. 
Now consider:
\begin{equation}\label{vtd}
\widetilde{V}_{N,q} = V_{N,q} (V_{N,q}^\dagger V_{N,q})^{-1/2}
\end{equation}
which is an isometry by the vanishing of the kernel (or using a polar decomposition). Compute:
\begin{align}
\| \widetilde{V}_{N,q}  - V_{N,q} \| & \leq \| \widetilde{V}_{N,q} \|  \| 1 -  (V_{N,q}^\dagger V_{N,q})^{1/2} \| \\ &  \leq
 \| (1 +  (V_{N,q}^\dagger V_{N,q})^{1/2})^{-1}\| \| 1 -  V_{N,q}^\dagger V_{N,q}\|  < \delta
\end{align}
for $N > N_1(\delta)$. Taking the limit we find for all $\delta > 0$:
\begin{equation}
\limsup_{N \rightarrow \infty} \| \widetilde{V}_{N,q}  - V_{N,q} \| < \delta
\end{equation}
which implies the result. 
\end{proof}

Define:
\begin{equation}
\alpha'_{N,q} = \left.  {\rm Ad}_{\widetilde{V}_{N,q}} \right|_{\mathcal{L}_N'}
\end{equation}
By Lemma~\ref{lemma:small} this is a quantum channel (normal unital completely positive) for sufficiently large $N$
at fixed $q$.  We now apply condition (4') given in \eqref{4prime} and work out the consequences for this small code:

\begin{lemma}
\label{lem:smallce}
Given condition (4') as for Theorem~\ref{thm:main} as stated in \eqref{4prime}, then for fixed $q$:
\begin{equation}
\lim_{N \rightarrow \infty} \| \alpha_{N,q}' - P_q \circ \alpha_{N,q}' \|_{cb} = 0
\end{equation}
\end{lemma}
where cb refers to the completely bounded norm.
\begin{proof}
By the definition of the regular channel norm, for all $\delta >0$ there is a bounded sequence $m_N' \in \mathcal{L}_N'$ such that:
\begin{equation}\label{mnd}
\left| \| \alpha_{N,q}' - P_q \circ \alpha_{N,q}' \| - \| \alpha_{N,q}'(m_N') - P_q \circ \alpha_{N,q}'(m_N') \|  \right| < \delta
\end{equation}
Apply this sequence to condition (4'). Fix some $u \in \widetilde{\mathcal{N}}_q$ and some $\delta_1 > 0$. By the s$^\star$ot density of $\mathcal{D}$ (recall that the double commutant $\mathcal{D}'' = \mathcal{N}$ is equal to closure of $\mathcal{D}$ in the strong$^\star$ operator topology
for unital $^\star$-algebras) there exists a $d \in \mathcal{D}$ such that $\| (d -  u) \left| \alpha \right> \| =  \| (d - \Phi_q^{-1}(u) ) \left| \alpha \right> \| < 2^{-q} \delta_1$  and $\| (d -  u)^\dagger \left| \alpha \right> \|   =  \| (d - \Phi_q^{-1}(u) )^\dagger \left| \alpha \right> \|  < 2^{-q} \delta_1$ for some finite orthonormal basis $\alpha$ for $\mathscr{H}_q$. Then $ \| (d -  u) \left| \zeta \right> \| \leq \sum_\alpha |c_\alpha| \| (d -  u) \left| \alpha \right> \|
\leq \| \zeta \| \delta_1 $ where $\left| \zeta \right> = \sum_\alpha c_\alpha \left| \alpha \right>$.
Hence the operator norm satisfies  $\| d \pi'_q - u  \| = \| d \pi'_q - u \pi_q' \| < \delta_1$. Similarly $\| \pi'_q d  - u \| < \delta_1$.

Using the same argument as in the proof of Lemma~\ref{lemma:small} we can turn (4') into an operator norm statement on the finite Hilbert space: for all $\epsilon >0$ there exists some $N(\epsilon)$ such that:
\begin{equation}
\|  \pi'_q [ V_{N}^\dagger  m_N' V_{N}, d]  \pi'_q \| < \epsilon %2^{2q}
\end{equation}
for all $N > N(\epsilon)$. 
But:
\begin{equation}
\| [ V_{N,q}^\dagger  m_N' V_{N,q}, u] \|  \leq \|  \pi'_q [ V_{N}^\dagger  m_N' V_{N}, d]  \pi'_q \| + 2 \delta_1 \| m_N' \|  C
\end{equation}
Thus:
\begin{equation}
\limsup_{N \rightarrow \infty} \| u^\dagger \widetilde{V}_{N,q}^\dagger  m_N' \widetilde{V}_{N,q} u -  \widetilde{V}_{N,q}^\dagger  m_N' \widetilde{V}_{N,q}  \| \leq 2 \delta_1 \sup_N \| m_N' \|  C
\end{equation}
using Lemma~\ref{lemma:small}. Since this is true for all $\delta_1 > 0$ then:
\begin{equation}
\label{domc}
\lim_{N \rightarrow \infty} \| u^\dagger \widetilde{V}_{N,q}^\dagger  m_N' \widetilde{V}_{N,q} u -  \widetilde{V}_{N,q}^\dagger  m_N' \widetilde{V}_{N,q}  \| = 0
\end{equation}

Now consider the Haar integral over the space of unitaries in $ \widetilde{\mathcal{N}}_q  \subset \mathcal{B}(\mathscr{H}_q)$.
This gives an alternative way to write the conditional expectation:
\begin{equation}
P_q(\cdot) = \int_{u \in  \widetilde{\mathcal{N}}_q } [ d u] u^\dagger ( \cdot ) u
\end{equation}
since $P_q$ clearly fixes $ \widetilde{\mathcal{N}}_q'$, $[P_q(n), v] =0$ for any unitary $v \in \widetilde{\mathcal{N}}_q$ and so $P_q(n) \in  \widetilde{\mathcal{N}}_q'$
and also ${\rm Tr}_{\mathscr{H}_q} P_q(n) = {\rm Tr}_{\mathscr{H}_q} n  $.  These properties uniquely fix $P_q$.

Using dominated convergence, we can take the Haar integral inside the limit of \eqref{domc}:
 \begin{equation}
 \label{alphadom}
0 = \lim_{N \rightarrow \infty} \int [du] \| u^\dagger \alpha_{N,q}' (m_N') u -  \alpha_{N,q}' (m_N' )  \| 
\geq \lim_N \| P_q \circ \alpha_{N,q}' ( m_N' )  -  \alpha_{N,q}' ( m_N' )  \| 
\end{equation}
and we used Jensen's inequality to take the integral inside of the norm. The dominating function to establish the left equality in \eqref{alphadom} is $ 2 \sup_N \| m_N' \| $  which is integrable over the normalized measure, and bounds the norm
difference in \eqref{domc}.

Thus:
\begin{equation}
\limsup_N \| \alpha_{N,q}' - P_q \circ \alpha_{N,q}' \|_{cb} <  2^{q} \limsup_N \| \alpha_{N,q}' - P_q \circ \alpha_{N,q}' \|
\leq 2^q \delta
 \end{equation}
 where we used the relation between the channel norm and the completely bounded norm valid in finite dimensions and also \eqref{mnd}. Since the argument started with an arbitrary $\delta >0$ we can take the limit $\delta \rightarrow 0$ to give the result. 
 
\end{proof}

We can also derive condition (1) from the properties of these small codes:

\begin{lemma}
\label{lem:beta}
Suppose that there is a sequence of faithful quantum channels $\beta_{N,q} :  \widetilde{\mathcal{N}}_q \rightarrow \mathcal{L}_N$ 
such that for all fixed $q$:
\begin{equation}
\label{defaN}
\lim_{N \rightarrow \infty} \| \alpha_{N,q} \circ \beta_{N,q}   - {\rm Id}_{\widetilde{\mathcal{N}}_q} \| = 0\, \qquad \alpha_{N,q} \equiv  \left. {\rm Ad}_{\widetilde{V}_{N,q}}\right|_{\mathcal{L}_N}
\end{equation} 
then property (1) in Theorem~\ref{thm:main} holds and also 
$\beta_N$ in (1) can be chosen to be faithful and to fix $\eta$ uniformly: 
\begin{equation}
\label{uniform2}
\lim_{N \rightarrow \infty} \| \omega_{V_N \eta} \circ \beta_N -  \omega_\eta |_{\mathcal{N}} \| = 0
\end{equation}
\end{lemma}
\begin{proof} 
By the statement, for all $\delta >0$ there exists some $N_1(q,\delta)$ such that:
\begin{equation}
 \| \alpha_{N,q} \circ \beta_{N,q}   - {\rm Id}_{\widetilde{\mathcal{N}}_q} \|  < \delta
\end{equation}
for all $N > N_1(q,\delta)$. 
We use the result of Swingle, Penington et al. that establishes: 
\begin{equation}
\| \alpha_{N,q} \circ \left( \beta_{N,q}(n_1) \ldots \beta_{N,q}(n_k)\right)  -n_1 \ldots n_k \| <  \frac{1}{2} \delta k(3k-1) \| n_1 \|  \ldots \| n_k \|
\end{equation}
for finite subsets $\{ n_1, \ldots n_k \} \in \widetilde{\mathcal{N}}_q$ and $N > N_1(q,\delta)$. This was derived in a finite dimensional setting and so we prove it again in the more general setting in Lemma~\ref{homo}.  Lemma~\ref{lemma:small} guarantees the existence of some $N_2(q,\delta)$ such that:
\begin{equation}
\| V_{N,q} - \widetilde{V}_{N,q} \| < \delta
\end{equation}
with $\widetilde{V}_{N,q} $ an isometry for all $N > N_2(q,\delta)$ if we additionally assume $\delta < 1$. Using this:
\begin{equation}
\label{finvnq}
\| V_{N,q}^\dagger \beta_{N,q}(n_1) \ldots \beta_{N,q}(n_k) V_{N,q}  -n_1 \ldots n_k \| <  \left( \frac{1}{2} \delta k(3k-1) + \delta(2+\delta) \right) \| n_1 \|  \ldots \| n_k \|
\end{equation}
for all $N > N(q,\delta) \equiv \max\{ N_1(q,\delta), N_2(q,\delta) \}$.
We set:
\begin{equation}
\beta_{N} = \beta_{N,q} \circ \Phi_{q} \circ  \iota_q \circ \mathcal{E}_{q, \eta}  
\end{equation}
for some $q = q(N)$ still to be determined.
 Here $\mathcal{E}_{q, \eta} : \mathcal{N} \rightarrow \widetilde{\mathcal{N}}_q$ are the generalized conditional expectations
discussed above around Lemma~\ref{lem:gencond} and $\Phi_q$ is the isomorphism given in \eqref{theisom} and $ \Phi_{q} \circ  \iota_q  = {\rm Id}_{ \widetilde{\mathcal{N}}_q}$ . 
Note that $\beta_N$ is faithful since it is a composition of faithful channels.
Fix:
\begin{equation}
n_i^q \equiv  \iota_q \circ \mathcal{E}_{q, \eta}  (n_i)\, \qquad i = 1, \ldots k
\end{equation}
From Lemma~\ref{lem:gencond} we know that:
\begin{equation}
\label{str1}
so-\lim_{q \rightarrow \infty} n_i^q = n_i \qquad so-\lim_{q \rightarrow \infty} \pi_q' = 1
\end{equation}
Since the $\| n_i^q \| \leq \| n_i \|$ are uniformly bounded (the Petz map is a quantum channel), we also have:
\begin{equation}
\label{str2}
so-\lim_{q \rightarrow \infty} n_1^q n_2^q \ldots n_k^q = n_1 n_2 \ldots n_k
\end{equation}
Fix two $\xi_{1,2} \in \mathscr{H}$ and write:
\begin{align}
\label{allterms}
& \left< \xi_1 \right| \left(V_N^\dagger \beta_N(n_1) \ldots \beta_N(n_k) V_N  -  n_1 \ldots n_k \right) \left| \xi_2 \right> \\ & =\left< \xi_1 \right|  \left(  V_{N,q}^\dagger \beta_{N,q} \circ \Phi_q(n_1^q) \ldots \beta_{N,q} \circ \Phi_q(n_k^q) V_{N,q} 
-   \Phi_q(n_1^q) \ldots \Phi_q(n_k^q) \right)  \left| \xi_2 \right> \\
\label{other1}
&\qquad + \left< \xi_1 \right| (n_1^q \ldots n_k^q - n_1 \ldots n_k) \left| \xi_ 2 \right> \\
\label{other2}
& \qquad + \left< \xi_1 \right|(1- \pi_q')  V_N^\dagger \beta_N(n_1) \ldots \beta_N(n_k) V_N \left| \xi_2 \right> \\
\label{other3}
&\qquad + \left< \xi_1 \right| \left( V_N^\dagger \beta_N(n_1) \ldots \beta_N(n_k) V_N - n_1^q \ldots n_k^q \right) (1-\pi_q') \left| \xi_2 \right> 
\end{align}
Now define:
\begin{equation}
N_{\max}(q) = \max\{q, \max_{0\leq q' \leq q} N(q', 1/(q'+2)) \} 
\end{equation}
for $ q \in \mathbb{N}$. 
We designed this function so that: (a) $N_{\max}(q) < \infty$ for all $q$ (b) it is non-decreasing with $q$
(c) for all $N > N_{\max}(q)$ the estimate \eqref{finvnq} applies at fixed $q$ with $\delta = 1/(q+2) < 1$ and finally (d) $\lim_{q \rightarrow \infty} N_{\max}(q) = \infty$.
Then define (see Figure~\ref{fig:q-N} for a cartoon of this function)
\begin{equation}
q(N) =\begin{cases}  \max\{ q \in \mathbb{N} : N_{\max}(q) < N \}\,&  N >  N_{\max}(0) \\
0 \, & N \leq N_{\max}(0)  \end{cases}
\end{equation}
which satisfies: (a') it exists and is finite for all $N \in \mathbb{N}$ (b')  it is non-decreasing 
(c') $N > N_{\max}(q(N))$ if $N > N_{\max}(0)$ (d')  $\lim_{N \rightarrow \infty} q(N) = \infty$ by (a) and (d).
 Condition (d') implies that for all $\epsilon > 0$ there exists a sufficiently large $N_\star $ (that we can take $> N_{\max}(0)$) such that $q(N) > 1/\epsilon$ for all $N > N_\star$ and by (c') $N > N_{\max}(q(N))$ applied to (c) we have the estimate \eqref{finvnq}  with $\delta = 1/(q(N)+2) < \epsilon/(1+ 2 \epsilon)$
 for $N > N_\star$. In \eqref{finvnq} we have substituted $n_i \rightarrow \Phi_q(n_i^q)$. This establishes that:
\begin{equation}
\lim_{N \rightarrow\infty} \left. \left< \xi_1 \right|  \left(V_{N,q}^\dagger \beta_{N,q} \circ \Phi_q(n_1^q) \ldots \beta_{N,q} \circ \Phi_q(n_k^q) V_{N,q}
-  \Phi_q(n_1^q) \ldots \Phi_q(n_k^q) \right) \left| \xi_2 \right>\right|_{q = q(N)} = 0
\end{equation}
for this particular choice of $q(N)$. Note that because we applied norm convergence in our construction of the curve $(q=q(N),N)$,
this particular curve does not depend on $\xi_{1,2}$ or the operators $n_i$, and this is crucial for the argument to work.

\begin{figure}[h!]
\centering
\includegraphics[scale=.35]{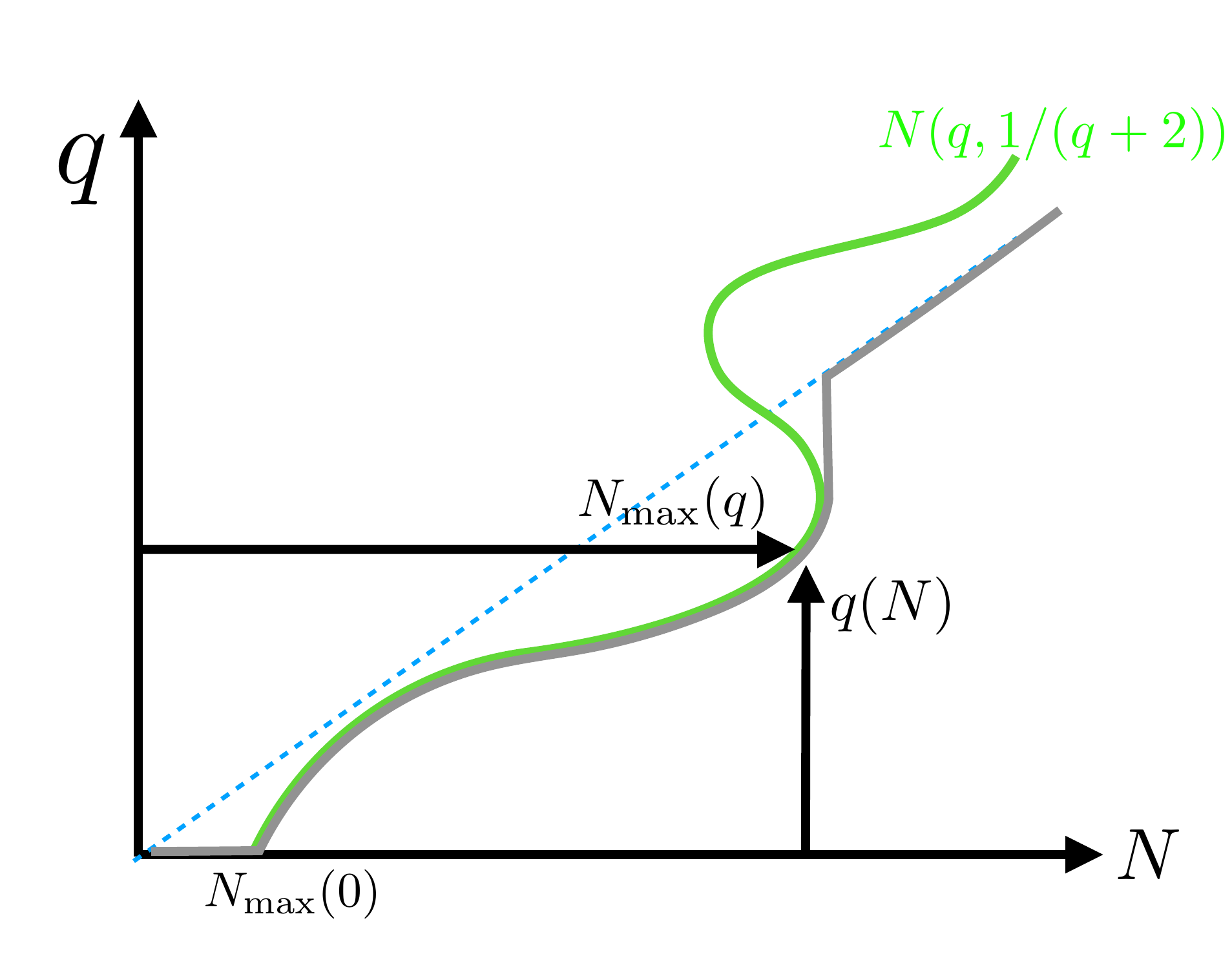}\includegraphics[scale=.35]{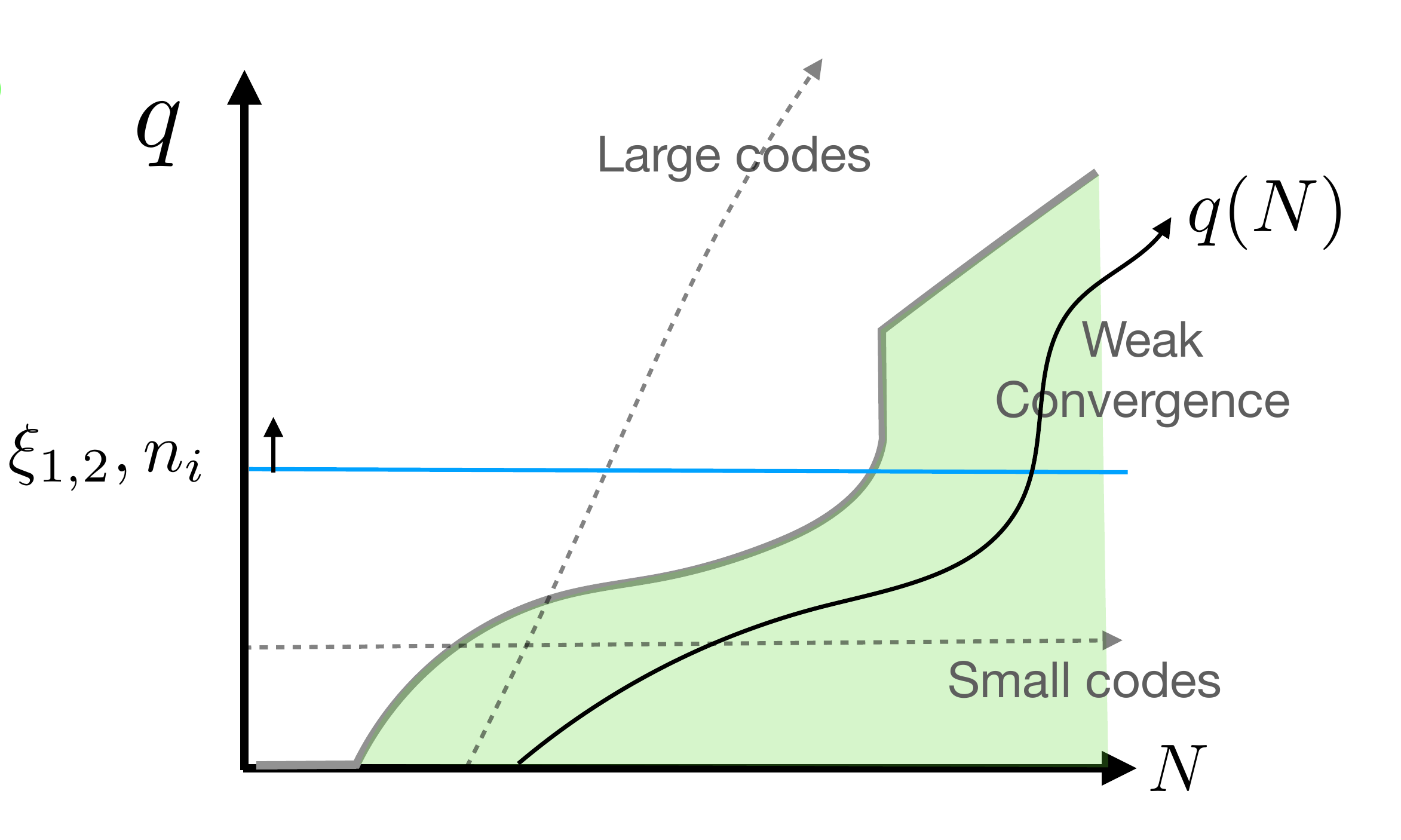}
\caption{ (\emph{left}) Cartoon of requirements on the size of the approximating type-I factor $q(N)$ as a function of $N$. If the function is chosen to lie on the gray line, then we can use the maps  $\beta_{N,q(N)}$ to give asymptotic reconstruction maps. The function must lie at or below the gray curve since larger code subspaces are problematic at fixed $N$ due to holographic bounds.
(\emph{right}) Different possibilities for codes in holographic theories. Most discussions of QEC in AdS/CFT have centered around small codes. Large codes have also been considered more recently \cite{Hayden:2018khn,Akers:2021fut}, and involve signification modifications to the standard error correcting the story. We straddle the two worlds along the weak convergence line. 
\label{fig:q-N}}
\end{figure}

The other terms \eqref{other1}-\eqref{other3} can be estimated to vanish in the limit $q \rightarrow \infty$ using the strong convergence given in \eqref{str1} and \eqref{str2}  and uniform boundedness of $V_N$. Since $q(N) \rightarrow \infty$ as $N \rightarrow \infty$ this establishes that:
\begin{equation}
\label{onlyk2}
\lim_{N \rightarrow \infty} \left< \xi_1 \right| V_N^\dagger \beta_N(n_1) \ldots \beta_N(n_k) V_N \left| \xi_2 \right> - \left< \xi_1 \right| n_1 \ldots n_k \left| \xi_2 \right> =0 
\end{equation}
Using Lemma~\ref{lem:seq} we get the required strong operator convergence statement (1) in Theorem~\ref{thm:main}.
\footnote{Note that the convergence of later \eqref{other1}-\eqref{other3}  will depend on $\xi_{1,2}$ or $n_i$, the key point being they don't simultaneously depend on us taking $N$
large, since otherwise we would have been forced to pick a $q(N)$ which is $\xi_{1,2}$ dependent - and this would have
spoiled the argument. That is because the constructed quantum channel $\beta_N$ should not depend on the operator that we feed it or the matrix elements we compute of it.}
Note that this Lemma only requires the $k=2$ case of \eqref{onlyk2}, and hence we only needed the $k=2$ case in Lemma~\ref{homo}. 

For the final part Lemma~\ref{lem:beta}, note that, since $\mathcal{E}_{q, \eta}$ is a Petz recovery map constructed using $\eta$, it trivially fixes
the state $\omega_\eta$. Also $\eta \in \mathscr{H}_q$
for all $q$. Thus $\omega_\eta|_{\mathcal{N}} = \omega_{\eta}\circ  \mathcal{E}_{q, \eta}$ such that:
\begin{align}
\| \omega_{V_N \eta} \circ \beta_{N}  - \omega_{\eta}|_{\mathcal{N}} \|
& = \left.  \| \omega_{V_{N,q} \eta} \circ \beta_{N,q}  \circ \mathcal{E}_{q, \eta}   - \omega_{\eta} \circ \mathcal{E}_{q, \eta}  \| \right|_{q = q(N)} \\
& \leq  \left. \sup_{n_q \in \widetilde{\mathcal{N}}_q} \frac{| \omega_{\eta} \left( V_{N,q}^\dagger \beta_{N,q}(n_q) V_{N,q}    - n_q  \right) |}{\| n_q\|}
\right|_{q = q(N)}
\leq (3 \delta + \delta^2) 
\end{align}
where  $\delta < 1/(q(N) +2)$. 
 Where in the first inequality we used monotonicity of the ``trace distance'' under the action of a quantum channel  $\mathcal{E}_{q, \eta}$ 
and in the second inequality we used \eqref{finvnq} and the properties (c) and (c') for $q(N)$ discussed above.
Taking the limit $N \rightarrow \infty$ establishes the uniform limit \eqref{uniform1} (establishing \eqref{uniform2}) for $\omega_\eta|_{\mathcal{N}}$ in statement (1) of Theorem~\ref{thm:main}.  
\end{proof}

\subsection{Proof of QEC Theorem~\ref{thm:main}} 

After having setup Lemma~\ref{lem:smallce} and Lemma~\ref{lem:beta} that allow us to pass respectively in and out of the small code
description we are now in a position to prove Theorem~\ref{thm:main}. 

\begin{proof}
{\bf Separating case}: We start with the case with a cyclic and separating $\eta$. 
We will prove the following $(1) \implies (2) \implies (3') \implies (4') \implies (1)$ and $(1) \implies (3) \implies (4) \implies (1)$. 

$(1) \implies (2)$ is trivial. $(2) \implies (3')$.  Take $\rho(1) = 1$ without loss of generality. Fix an $\epsilon >0$ and take $u \in \overline{\mathcal{D}}^{\| \cdot \|}$ unitary from the concrete $C^\star$ algebra.
 There exists a sequence of positive operators $p_N' \leq 1$ with $p_N' \in \mathcal{L}_N'$ that approximates the ``trace distance'':
\begin{equation}
\label{approxlf}
\left| \| \left( \rho \circ {\rm Ad}_{V_N} ( \cdot) -  \rho_u \circ {\rm Ad}_{V_N}  ( \cdot) \right)|_{\mathcal{M}_N(O)'} \| 
 - | \rho \circ {\rm Ad}_{V_N} ( p_N') -\rho_u \circ {\rm Ad}_{V_N}  ( p_N' )|  \right| < \epsilon
\end{equation}
For all $\delta_1 > 0$ there exists some $d\in \mathcal{D}$ such that $\| u - d\| < \delta_1$. 
Set $m_N = \gamma_N(d)$ from statement $(2)$. Computing:
\begin{align} 
\nonumber
& |   \rho \circ {\rm Ad}_{V_N}  ( p_N') -   \rho_u \circ {\rm Ad}_{V_N}  ( p_N')| 
 = |\rho( V_N^\dagger p_N'  ( V_N u^\dagger u - m_N^\dagger m_N V_N) )\non\\
 &\qquad\qquad\qquad\qquad- \rho((u^\dagger V_N^\dagger - V_N^\dagger m_N^\dagger) p_N' V_N u )
+  \rho( V_N^\dagger m_N^\dagger p_N'  (m_N V_N - V_N u) ) |
 \label{eachterm}
\end{align}
If we were to replace $u$ to $d$ to use strong convergence of $so-m_NV_N-V_Nd=0$ on the right-hand side of \eqref{eachterm}, we could estimate each of these terms as $N \rightarrow \infty$ using, triangle, Cauchy-Schwarz, boundedness of $p_N'$ and $V_N$. Explicitly, we have

\begin{align}
\label{a-cs}
&\limsup_{N \rightarrow \infty} | \rho( V_N^\dagger p_N'  ( V_Nu^\dagger u- m_N^\dagger m_N V_N ) )|\non\\
\leq& C^{1/2}\lim_{N \rightarrow \infty}\left( \rho (\lvert V_N d^\dagger d- m_N^\dagger m_N V_N\rvert^2)\right)^{1/2}+\sup_{N }\norm{ V_N^\dagger p_N  V_N }\norm{u^\dagger u-d^\dagger d }\non\\
\leq& C\norm{u^\dagger u-d^\dagger d}\leq C\delta_1(\delta_1+2)
\end{align}
where $\lvert V_N d^\dagger d- m_N^\dagger m_N V_N\rvert^2:= ( d^\dagger d V_N^\dagger  - V_N^\dagger  m_N^\dagger m_N)  ( V_N d^\dagger d- m_N^\dagger m_N V_N) $. Now the argument of $\rho$ in the right hand side of \eqref{a-cs} vanishes in the weak operator topology because of the strong operator limit \eqref{so-d} assumed in (2)
and Lemma~\ref{lem:seq}.
Uniform
boundedness in $N$ of $\gamma_N(d)$ assumed in (2) implies that the argument of $\rho$ is  uniformly bounded.
This is enough to prove that the limit of  \eqref{a-cs}  vanishes; see Lemma~\ref{ultratonot} in Appendix~\ref{app:convergence}.  The next two terms in (\ref{eachterm}) reads
\begin{align}
  &\limsup_{N \rightarrow \infty}  \lvert \rho((u^\dagger V_N^\dagger - V_N^\dagger m_N^\dagger) p_N' V_N u )\non\\
  \leq&\limsup_{N \rightarrow \infty}\lvert \rho((d^\dagger V_N^\dagger-V_N^\dagger m_N^\dagger)p_N'V_Nu)\rvert+  \sup_N \norm{u-d} \norm{  V_N^\dagger p_N' V_N} \norm{u} \rvert\leq C\delta_1\non
\end{align}
similarly,
\begin{align}
    &\limsup_{N \rightarrow \infty}\lvert\rho( V_N^\dagger m_N^\dagger p_N'  (m_N V_N - V_N u) )\rvert\non\\
    \leq&\limsup_{N \rightarrow \infty}\lvert\rho( V_N^\dagger m_N^\dagger p_N'  (m_N V_N - V_N d) )\rvert +\delta_1 C^{1/2} \limsup_N \rho( V_N^\dagger m_N^\dagger m_N V_N)^{1/2}\non
\end{align}
on both of the above inequalities, the first term of the first inequalities will be approximated by Cauchy-Schwarz inequality and vanish based on $so-m_NV_N-V_Nd=0$ as in (\ref{a-cs}).

Similar estimates for each term in \eqref{eachterm} with the replacement $u \rightarrow d$ show that these all vanish in the limit. Hence, also using \eqref{approxlf}
\begin{align}
&\limsup_{N \rightarrow \infty} \| \left( \rho \circ {\rm Ad}_{V_N} ( \cdot) - \rho_u \circ {\rm Ad}_{V_N} ( \cdot) \right)|_{\mathcal{L}_N'} \|   < \epsilon 
\\& \qquad \qquad + \delta_1 C^{1/2} \left( (3+\delta_1) C^{1/2}+   \limsup_{N\rightarrow \infty} \rho( V_N^\dagger \gamma_N(d^\dagger) \gamma_N(d) V_N)^{1/2} \right)
\end{align}
The later $\delta_1$ term comes from approximating $u$ by $d$ in norm along with Cauchy-Schwarz and boundedness of $V_N$ and $p_N'$.  The limit in the second
term can be computed and gives $\rho(d^\dagger d) \leq \| d \|^2  \leq (\| u \| + \| d- u \| )^2 = (1+ \delta_1)^2$. Sending $\epsilon \rightarrow 0$ and $\delta_1 \rightarrow 0$ gives the proof.

\noindent $(3') \implies (4')$ \,\, Fix some $m_N'$ as per (4). Fix two vectors $\xi_{1,2} \in \mathscr{H}$. Consider some $u \in \overline{\mathcal{D}}^{\|\cdot \|}$. The linear functional $\omega_{\xi_1,u^\dagger \xi_2} = \left<\xi_1 \right| \cdot u^\dagger \left| \xi_2 \right>$ can be written as a sum of four positive elements  $\rho_i \in \mathcal{B}(\mathscr{H})_\star^+$. 
Apply (3')  to the four  positive elements of $\mathcal{B}(\mathscr{H})_\star^+$, and we find:
\begin{align}
\lim_{N \rightarrow \infty} &| \left<\xi_1 \right| ( V_N^\dagger m_N' V_N u^\dagger - u^\dagger V_N^\dagger m_N' V_N ) \left| \xi_2 \right>| 
 \\ & \qquad \qquad \leq \lim_{N \rightarrow \infty}  \| m_N'\| \sum_{i=1}^4 \|  \rho_i \circ {\rm Ad}_{V_N} ( \cdot) -  \rho_i \circ {\rm Ad}_{V_N u} ( \cdot) \| = 0
\end{align}
Now, since any bounded operator $d \in \mathcal{D} \subset  \overline{\mathcal{D}}^{\|\cdot \|}$ can be written as a linear combination of four scaled unitaries in 
the $C^\star$ algebra  $\overline{\mathcal{D}}^{\|\cdot \|}$, we are done.

\noindent $(4') \implies (1)$ \,\, 
Consider any fixed $q \in \mathbb{N}$ and use Lemma~\ref{lem:smallce}. Fix a $\delta > 0$
there exists some $N(\delta)$ such that:
\begin{equation}
\| \alpha_{N,q}' - P_q \circ \alpha_{N,q}' \|_{cb} < \delta \qquad \forall N > N(\delta)
\end{equation}
We use the privacy correctability correspondence \cite{crann2016private}, see also~\cite{kretschmann2008continuity}. Since $P_q \circ \alpha_{N,q}' : \mathcal{L}_N' \rightarrow \widetilde{\mathcal{N}}_q'$
we find that $\widetilde{\mathcal{N}}_q$ is $\delta$-private for $ \alpha_{N,q}'$. 
Now $\alpha_{N,q} \equiv {\rm Ad}_{\widetilde{V}_{N,q}}|_{\mathcal{L}_N}$ is a complementary channel by definition. The privacy correctability result, Theorem~4.7 of \cite{crann2016private}, tells
us that  $\widetilde{\mathcal{N}}_q$ is $2 \delta^{1/2}$ correctible for $\alpha_{N,q}$. That is there is a sequence of quantum channel (normal unital completely positive) $\beta_{N,q} :  \widetilde{\mathcal{N}}_q \rightarrow \mathcal{L}_N$ with:
\begin{equation}
\| \alpha_{N,q} \circ \beta_{N,q}  - {\rm Id}_{\widetilde{\mathcal{N}}_q} \|_{cb} < 2 \delta^{1/2}
\end{equation}
for all $N > N(\delta)$. Taking the limit (at fixed $q$) gives:
\begin{equation}
\label{flimit}
\lim_{N \rightarrow \infty} \| \alpha_{N,q} \circ \beta_{N,q}  - {\rm Id}_{\widetilde{\mathcal{N}}_q} \|_{cb}  = 0
\end{equation}
We need $\beta_{N,q}$ to be faithful. We can simply redefine:
\begin{equation}
\beta_{N,q} (\cdot) \rightarrow (1 - \lambda_N) \beta_{N,q}(\cdot) + 1 \lambda_N \omega_{\eta}(\cdot)
\end{equation}
where $0 < \lambda_N < 1$ is a sequence with $\lim_N \lambda_N = 0$. Since $ \omega_{\eta}|_{\widetilde{\mathcal{N}}_q}$
is faithful, the new channel must be faithful for each fixed $N$. It also still satisfies \eqref{flimit}. 
Since this applies for all $q$ we can then use Lemma~\ref{lem:beta} which implies (1)
and the uniformity statement in the last part of Theorem~\ref{thm:main} - see \eqref{uniform1}.

The other set of implications $(1) \implies (3) \implies (4) \implies (1)$ are  simpler versions of the above argument working
with the von Neumann algebra $\mathcal{N}$ rather than $\mathcal{D}$.

We now establish the rest of the last part of Theorem~\ref{thm:main}. 
For each $N \in \mathbb{N}$ there is a $\varPhi_N$-dual quantum channel (normal unital completely positive) $\alpha_N' : \mathcal{L}_N' \rightarrow \mathcal{N}'$ defined as follows:
\begin{equation}
\label{def:pm}
%\left< \eta \right| V_N^\dagger 
\left< \varPhi_N \right| \beta_N(n) m_N' 
%V_N \left| \eta \right> 
\left| \varPhi_N \right> = \left< \eta_N \right| \alpha_N'(m_N') n \left| \eta_N \right>
\end{equation}
where $\eta_N \in \mathscr{H}$ chosen in the natural cone for $\eta$ such that $\omega_{\eta_N}|_{\mathcal{N}} = \omega_{\varPhi_N} \circ \beta_N $. 
Since $\varPhi_N$ is represented by $\eta$ on the code:
\begin{equation}
    \|(\omega_{\varPhi_N} - \omega_{V_N \eta})|_{\mathcal{L}_N}\| \rightarrow 0
    %\leq \frac{1}{2} \| \varPhi_N - V_N \eta \|^2 
\end{equation}
using monotonicty of the linear functional 
norm under the inclusion $\mathcal{L}_N \subset \mathcal{B}(\mathscr{K}_N)$ and one of the  Fuchs–van de Graaf inequalities.
Hence from  \eqref{uniform1} we can take $\beta_N(n)$ such that (uniformly):
\begin{equation}
\lim_{N} \| (\omega_{\eta_N} - \omega_{\eta})|_{\mathcal{N}} \| = 0 
\end{equation}
Since $\eta,\eta_N$ are both in the same natural cone of $\mathcal{N}$ we also have:
\begin{equation}
\lim_{N}\| \eta_N - \eta \| = 0
\end{equation}
Hence from \eqref{def:pm} and $V_N \eta -\varPhi_N \rightarrow 0$
\begin{equation}
\lim_{N \rightarrow \infty} \left< \eta \right| n_1 \left( \alpha_N'( m_N' ) - V_N^\dagger m_N' V_N \right) n_2 \left| \eta \right> =0
\end{equation}
for bounded sequences $m_N' \in \mathcal{L}_N'$. Since we have weak convergence on a dense subspace, uniform
boundedness extends this to weak operator convergence on all $\mathscr{H}$:
\begin{equation}
\label{aftnonsep}
wo-\lim_{N \rightarrow \infty}\left( \alpha_N'( m_N' ) - V_N^\dagger m_N' V_N \right)  =0
\end{equation}

\noindent {\bf Non-separating case:}  We observe that the proof given above of $(1) \implies (2) \implies (3') \implies (4')$ does 
not require the separating property of $\eta$. Nor does $(1) \implies (3) \implies (4)$. We just have to complete these circles of implications. 

We form:
\begin{equation}
\left| \hat{\eta} \right> = \sum_i  \lambda_i \left|\eta_i \right> \otimes \left| i \right>_R
\end{equation}
on
some infinite dimensional reference with orthonormal $ \left| i \right>_R, i = 0,1,\ldots$ basis
and such that $\eta_0 = \eta$ the cyclic vector for $\mathcal{N}$.
Now we pick this state to be cyclic on $\mathcal{B}(\mathscr{H}_R)$ which amounts to picking all $\lambda_i \neq 0$.
Then $\hat{\eta}$ will be automatically separating for $\hat{\mathcal{N}}$. 
Now let us compute the support projection:
\begin{equation}
e' = \overline{\hat{\mathcal{N}} \left| \hat{\eta} \right> } \in \mathcal{N}' \otimes \mathcal{B}(\mathscr{H}_R)
\end{equation}
if $e'=1$ then $\hat{\eta}$ would be cyclic. We do not expect this to be the case. 
However we do know that $e' \in \mathcal{N}' \otimes \mathcal{B}(\mathscr{H}_R)$. 
We also note that:
\begin{equation}
\left._R\left< 0 \right| e' \left| 0 \right>_R \right. \equiv \mu'
\end{equation}
is a bounded positive invertible operator in $\mathcal{N}'$ with $\| \mu' \| \leq 1$. It is invertible by contradiction: suppose there is 
a kernel - then there is some $\chi \in \mathscr{H}$ such that $\mu' \left| \chi \right> =0$. Which implies that:
\begin{equation}
\| e' \left| \chi \right> \left| 0\right>_R \| =0 \implies e' \left| \chi \right> \left| 0\right>_R = 0
\end{equation}
Thus for all $n \in \mathcal{N}$: 
\begin{equation}
\left< \chi \right| \left._R \left< 0 \right| \right. n \otimes 1_R \left| \hat{\eta} \right> =0
\end{equation}
That is $\lambda_0 \left< \chi \right| n \left| \eta \right> =0$ for all $n$
and this implies that $\chi =0$ by the cyclic property of $\eta$ which is a contradiction.

Hence Theorem~\ref{thm:main} applies in the case of the sequence of operators
$\hat{V}_N = (V_N \otimes 1_R) e' 
: e' (\mathscr{H} \otimes \mathscr{H}_R) \rightarrow \mathscr{K}_N \otimes \mathscr{H}_R$ the bulk algebra $\hat{\mathcal{N}} e' $
and the boundary algebra
$\hat{\mathcal{L}}_N = \mathcal{L}_N \otimes 1_R$. We can apply this theorem
since we have proven it in the cyclic and separating case and since $\hat{\eta} \in e' (\mathscr{H} \otimes \mathscr{H}_R)$ is cyclic and separating  for $\hat{\mathcal{N}} e'$. 
Thus if any of the conditions are filled we have the other which we denote in this
hatted form as $(\hat{1}) \implies (\hat{2}) \implies (\hat{3}')
\implies (\hat{4}') \implies (\hat{1})$ and $(\hat{1})  \implies (\hat{3})
\implies (\hat{4}) \implies (\hat{1})$.   Hence our proof is
complete if we can show that $(4') \implies (\hat{4}')$, $(4) \implies (\hat{4})$ 
and $(\hat{1}) \implies (1)$.

\noindent $(4') \implies (\hat{4}')$ and $(4) \implies (\hat{4})$  are essentially the same argument. We pick a dense set of states for $\xi_{1,2} \in \mathscr{H} \otimes \mathscr{H}_R$
which involves finite sums of basis elements $\left| \eta_i \right> \otimes \left| j \right>_R$ given above. Then:
\begin{equation}
\left< \hat{\xi}_1 \right| \left[ \hat{V}_N^\dagger \hat{m}_N'  \hat{V}_N, d \otimes 1_R \right] \left| \hat{\xi}_2 \right>
= \sum_{i_1,j_1, i_2,j_2} (c^1_{i_1j_1})^\star c^2_{i_2j_2}
\left< \eta_{i_1} \right| \left[ V_N^\dagger  \left< j_1 \right| \hat{m}_N' \left| j_2\right>  V_N, d \right] \left|  \eta_{i_2} \right>
\end{equation}
for uniformly bounded $\hat{m}_N' \in \hat{\mathcal{L}}_N'$. We can now take the limit term by term to conclude that:
\begin{equation}
\lim_N \left< \hat{\xi}_1 \right| \left[ \hat{V}_N^\dagger \hat{m}_N'  \hat{V}_N, d \otimes 1_R \right] \left| \hat{\xi}_2 \right> =0
\end{equation}
But since the external states are a dense subspace of $\mathscr{H} \otimes \mathscr{H}_R$ and the operator is uniformly bounded
we finally prove the weak operator convergence on $\mathscr{H} \otimes \mathscr{H}_R$ which implies the weak operator
convergence on the projector subspace of $e'$ and thus $(\hat{4}')$.

\noindent $(\hat{1}) \implies (1)$. We have the channel $B_N : \hat{\mathcal{N}} e' \rightarrow \mathcal{L}_N \otimes 1_R$ which satisfies:
\begin{equation}
\label{soB}
so-\lim_N B_N(\hat{n} ) (V_N \otimes 1_R) e' - (V_N \otimes 1_R) e' \hat{n} =0
\end{equation}
for all $\hat{n} \in \hat{\mathcal{N}} e' $. By the separating property of $\hat{\eta}$ we can define the isomorphism $\Phi : \mathcal{N} \otimes 1_R \rightarrow
\hat{\mathcal{N}} e'$ and then $B_N \circ \Phi(\cdot \otimes 1_R) = \beta_N(\cdot) \otimes 1_R$ where $\beta_N : \mathcal{N} \rightarrow \mathcal{L}_N$ is a quantum channel. 
Then \eqref{soB} becomes:
\begin{equation}
so-\lim_N  ( \beta_N(n) V_N \otimes 1_R) e' - (V_N n \otimes 1_R) e' =0
\end{equation}
taking matrix elements $_R\left< 0 \right| \cdot \left| 0 \right>_R$ which preserves the strong operator convergence statement we find:
\begin{equation}
so-\lim_N  \beta_N(n) V_N \mu'  - V_N n \mu' =0
\end{equation}
Now since we proved that $\mu'$ is a bounded invertible operator we can simply remove this. We arrive at (1). 
\end{proof}

We needed:
\begin{lemma}
\label{homo}
Consider two completely positive unital maps:
\begin{equation}
\| \alpha \circ \beta  - {\rm Id}_{\mathcal{N}} \| < \epsilon 
\end{equation}
where $\alpha : \mathcal{M} \rightarrow \mathcal{N}$ and $\beta: \mathcal{N} \rightarrow \mathcal{M}$ are between unital $C^\star$-algebras,
and we assume nothing about the dimensionality of these algebras. 
Then:
\begin{equation}
\| \alpha \left( \beta(n_1) \ldots \beta(n_k)\right)  -n_1 \ldots n_k \| <  \frac{1}{2} \epsilon k(3k-1) \| n_1 \|  \ldots \| n_k \|
\end{equation}
for finite subsets $\{ n_1, \ldots n_k \} \in \mathcal{N}$. 
\end{lemma}
\begin{proof}
The proof in \cite{Cotler:2017erl} uses Stinespring in finite dimensions. Since there is a general Stinespring representation for completely positive maps \cite{paulsen2002completely}, the
proof remains the same. We give here a short proof for $k=2$ using well known properties of completely positive maps:

For any $n\in\N$, one has
\begin{align}\nonumber
    \norm{\alpha(\beta(n^\dagger n))-\alpha(\beta(n^\dagger ))\alpha(\beta( n))}&\leq \norm{\alpha(\beta(n^\dagger n))-n^\dagger n}+\norm{n^\dagger n-\alpha(\beta(n^\dagger ))\alpha(\beta( n))}\non\\ & \leq3\epsilon\norm{n}^2
    \label{ndn}
\end{align}
where the second inequlity follows from:
\begin{equation} 
\norm{n^\dagger n-\alpha(\beta(n^\dagger ))\alpha(\beta( n))}\leq\norm{\alpha\beta(n^\dagger)}\norm{n-\alpha(\beta(n))}+\\ \norm{n}\norm{n^\dagger-\alpha(\beta(n^\dagger))}\leq 2\epsilon\norm{n}^2
\end{equation}
after using the basic monotonicity result $\norm{\beta(n)}\leq\norm{n}$. Based on properties of completely positive maps:
\begin{align}\label{postive}
    \alpha(\beta(n^\dagger n))-\alpha(\beta(n^\dagger ))\alpha(\beta( n))\geq \alpha(\beta(n^\dagger )\beta( n))-\alpha(\beta(n^\dagger ))\alpha(\beta( n))\geq 0,
\end{align}
the first inequlity being the operator Schwarz inequality on $\beta$, and the positivity of $\alpha$; and the second inequlity is simply the operator Schwarz inequality for $\alpha$. Therefore we can estimate the norm of the middle positive operator using \eqref{postive} and \eqref{ndn}:
\begin{align}\label{choi-Jensen-David}
   \norm{\alpha(\beta(n^\dagger )\beta( n))-\alpha(\beta(n^\dagger ))\alpha(\beta( n))}\leq3\epsilon\norm{n}^2
\end{align}
Consider the positive matrix \cite{bhatia2009positive}
\begin{align}\label{pos-M_3}
    \begin{pmatrix}
    \beta(n_1^\dagger)\beta(n_1) & \beta(n_1^\dagger)\beta(n_2) &\beta(n_1^\dagger)\\
    \beta(n_2^\dagger)\beta(n_1) & \beta(n_2^\dagger)\beta(n_2) & \beta(n_2^\dagger) \\
    \beta(n_1) & \beta(n_2) &1
    \end{pmatrix}=
    \begin{pmatrix}
    \beta(n_1^\dagger) & 0 &0\\
    \beta(n_2^\dagger) & 0 & 0 \\
    1 & 0 &0
    \end{pmatrix}
    \begin{pmatrix}
    \beta(n_1) & \beta(n_2) &1\\
    0 & 0 & 0 \\
    0 & 0 &0
    \end{pmatrix}
\end{align}
A hermitian block matrix $\begin{pmatrix}
    A& X\\
   X^\dagger & B
    \end{pmatrix}$ with block diagonal terms greater than zero is positive iff $A\geq X B^{-1} X^\dagger$, therefore acting with the quantum channel $\alpha\otimes\mathbbm{I}_3$ on (\ref{pos-M_3}) we find this remains  positive iff
    \begin{eqnarray}
        \begin{pmatrix}
    \alpha(\beta(n_1^\dagger)\beta(n_1)) & \alpha(\beta(n_1^\dagger)\beta(n_2)) \\
    \alpha(\beta(n_2^\dagger)\beta(n_1)) & \alpha(\beta(n_2^\dagger)\beta(n_2) )
    \end{pmatrix}\geq\begin{pmatrix}
   \alpha( \beta(n_1^\dagger)) & 0 \\
    \alpha(\beta(n_2^\dagger) )& 0 
    \end{pmatrix}
    \begin{pmatrix}
    1 &0\\
    0 & 0 
    \end{pmatrix}
    \begin{pmatrix}
    \alpha(\beta(n_1)) &\alpha( \beta(n_2) )\\
    0 & 0 
    \end{pmatrix}
    \end{eqnarray}
   or
     \begin{eqnarray}
        \begin{pmatrix}
    \alpha(\beta(n_1^\dagger)\beta(n_1)) - \alpha( \beta(n_1^\dagger)) \alpha( \beta(n_1))& \alpha(\beta(n_1^\dagger)\beta(n_2))-\alpha(\beta(n_1^\dagger))\alpha(\beta(n_2)) \\
    \alpha(\beta(n_2^\dagger)\beta(n_1))-\alpha(\beta(n_2^\dagger))\alpha(\beta(n_1)) & \alpha(\beta(n_2^\dagger)\beta(n_2) )-\alpha(\beta(n_2^\dagger))\alpha(\beta(n_2) )
    \end{pmatrix}\geq0
    \end{eqnarray}
    It is known that semi-positivity of such a Hermitian operator in $\N \otimes\mathbbm{M}_2(\mathbbm{C})$, with diagonal terms greater than or equal to $0$, requires there to exist some $y,\norm{y}\leq 1$, such that $E_{12}=E_{11}^{1/2}yE_{22}^{1/2}$, therefore based on (\ref{choi-Jensen-David}) we have 
\begin{align}
   \norm{{\alpha(\beta(n_1)\beta(n_2)) -\alpha(\beta(n_1))\alpha(\beta(n_2))}}\leq 3\epsilon\norm{n_1}\norm{n_2}
\end{align}
and based on $\norm{\alpha(\beta(n_1))\alpha(\beta(n_2))-n_1n_2}\leq 2\epsilon\norm{n_1}\norm{n_2}$, one has
\begin{align}
   \norm{\alpha(\beta(n_1)\beta(n_2)) -n_1n_2}\leq 5\epsilon\norm{n_1}\norm{n_2}
\end{align}

Technically all we used in the proof  of Lemma~\ref{lem:beta} was the $k=2$ case. So the argument above is sufficient here. The $k > 2$ requires us to
use the Strinespring argument in \cite{Cotler:2017erl} which generalizes to the general $C^\star$ algebra setting. 
\end{proof}

\subsection{Proof of JLMS condition Theorem~\ref{thm:JLMS}}

\label{proof:jlms}

We need some more results in Modular theory that we simply quote here. These can be found in textbooks \cite{stratila2020modular}. We also refer to \cite{Ceyhan:2018zfg} (albeit beware of a notation swtich $\Delta_{\eta|\psi}=\Delta_{\psi,\eta}$).
The natural positive cone \cite{araki1974some} is defined with respect to the non-relative modular data
introduced around \eqref{TTresults}: 
\begin{equation}\label{natural-cone}
\mathcal{P}^\natural_{\eta;\mathcal{N}} = \overline{\{ n J_{\eta;\mathcal{N}} n J_{\eta;\mathcal{N}} \left| \eta \right>; n \in \mathcal{N}\}} = \overline{\{ \Delta_{\eta;\mathcal{N}}^{1/4} n \left| \eta \right>; n \in \mathcal{N}^+\}} \subset \mathscr{H}
\end{equation}
All cyclic and separating vectors in the natural cone have the same modular conjugation operators $J_{\eta;\mathcal{N}}$. 
Any state $\rho \in \mathcal{N}_\star^+$ has a unique representative vector from this cone.
Thus any vector $\psi \in \mathscr{H}$ can be rotated inside $\theta \psi  \in \mathcal{P}^\natural_{\eta;\mathcal{N}}$
using a unique partial isometry $\theta$. 
Relative versions of this modular data exist
for two states $\psi,\eta$. If $\psi$ is in the natural cone of $\eta$ then the relative modular data takes a particularly simple form, so it is convenient to work in the natural cone for this reason.
With $\psi \in \mathcal{P}^\natural_{\eta;\mathcal{N}}$, assumed cyclic and separating for $\mathcal{N}$, then the defining equation of the relative modular data becomes:
$J_{\eta;\mathcal{N}} \Delta_{\eta,\psi;\mathcal{N}}^{1/2} n \left|\psi \right> = n^\dagger \left| \eta \right>$. Where the non-relative version is simply $\Delta_{\psi;\mathcal{N}} = \Delta_{\psi,\psi;\mathcal{N}}$.
We have some useful relations that we use in the proofs, and that we simply quote here
for any two cyclic and separating vectors $\eta,\psi \in \mathscr{H}$:
\begin{align}\label{tomita-opts}
&J_{\eta;\mathcal{N}} \Delta_{\psi,\eta;\mathcal{N}}^{z} = \Delta_{\eta,\psi;\mathcal{N}}^{-\bar{z}} J_{\eta;\mathcal{N}} \,, \qquad \psi \in \mathcal{P}^\natural_{\eta;\mathcal{N}} \\\nonumber %\label{tomita-opts}
&\Delta_{\psi,\eta;\mathcal{N}'}^{z} = 
\Delta_{\eta,\psi;\mathcal{N}}^{-z}\,,  \qquad
\Delta_{u u' \psi, \eta;\mathcal{N}}
= u \Delta_{\psi, \eta;\mathcal{N}} u^\dagger\,, \qquad J_{u u' \psi;\mathcal{N}}
= u u' J_{\psi;\mathcal{N}} (u u')^\dagger
\end{align}
 for unitaries $u \in \mathcal{N}$ and $u' \in \mathcal{N}'$; $z\in\mathbb{C}$. Only the first equation
requires $\psi$ to be in the natural cone. 
See \cite{Ceyhan:2018zfg} for a derivation
and a general definition of the relative modular operators when the vectors are not in the same natural cone (beware the notation change in that paper $\Delta_{\psi|\eta} = \Delta_{\eta,\psi}$.)

We start with the following, which allows us to pass from convergence of the smooth relative entropy with the action of ${\rm Ad}_{V_N}$, to convergence in relative entropy
for the channel $\alpha_N$ that results from reconstruction on the commutant: 
\begin{lemma}
\label{lem:preJLMS}
 Assuming condition (ii) of Theorem~\ref{thm:JLMS} then:
 \begin{enumerate}
 \item $\mathcal{N}' \rightarrowtail \mathcal{L}_N' $
 \item For  $\alpha_N$ given by the complement\footnote{That is applied to $(\mathcal{N}',\mathcal{L}_N')$ rather than
 the original version for $(\mathcal{N},\mathcal{L}_N)$.} version of the reconstruction Theorem~\ref{thm:main}:
\begin{equation}
\lim_{N \rightarrow \infty} S( \rho \circ \alpha_N |\sigma \circ \alpha_N ) = S(\rho|\sigma)
\end{equation}
for all states $\rho, \sigma \in \mathcal{N}_\star^+$ with $\rho < \lambda \sigma$ for some real $1\leq \lambda < \infty$. 
\end{enumerate}
\end{lemma}
\begin{proof}
(1) In \eqref{rhv} consider $\rho = \sigma_{u'}$ for some unitary $u' \in \mathcal{N}'$ which induces the same state on $\mathcal{N}$ i.e. 
$\rho|_{\mathcal{N}} = \sigma|_{\mathcal{N}}$. Hence 
the right hand side of \eqref{rhv} vanishes. 
Also note that the states on the left hand side are asymptotically normalized:
\begin{equation}
 \lim_N ( \sigma \circ {\rm Ad}_{V_N}(1) -1) = \lim_N \sigma( V_N^\dagger V_N -1) =0
\end{equation}
due to ultraweak convergence which agrees with weak operator convergence on uniformly bounded sequences of operators, see Lemma~\ref{ultratonot} in Appendix~\ref{app:convergence}.  
This is also true for $\sigma_{u'} \circ {\rm Ad}_{V_N}(1)$.  Thus we can apply Pinsker's inequality \eqref{eq:correctp} to the smoothed relative entropy \eqref{rhv} to show that:
\begin{equation}
\label{tollm}
\|( \sigma_{u'} \circ {\rm Ad}_{V_N} - \sigma \circ {\rm Ad}_{V_N} ) |_{\mathcal{L}_N}\| 
\rightarrow 0
\end{equation}
and this is condition (3) of Theorem~\ref{thm:main} for the commutant algebra $\mathcal{N}'$.  Hence $\mathcal{N}' \rightarrowtail \mathcal{L}_N' $.

(2) Given two states $\rho, \sigma \in \mathcal{N}_\star^+$ with $\rho < \lambda \sigma$ it is always possible to find states $\hat{\rho},\hat{\sigma} \in \mathcal{B}(\mathscr{H})_\star^+$ that restrict to these states on $\mathcal{N}$. Apply $\hat{\rho},\hat{\sigma}$ to (ii) of Theorem~\ref{thm:JLMS}. 
From (1) above we learn there is a reconstruction map $\beta_N'$. 
 Then combining condition (4) of Theorem~\ref{thm:main} with the last part of this same theorem, then we know that the $\varPhi_N$-dual channel to $\beta_N'$, labelled as $\alpha_N (\equiv \delta'_{\beta'_N,\varPhi_N}$ in Definition~\ref{def:dual}) satisfies:
\begin{equation}
\lim_{N \rightarrow \infty}  \| \hat{\rho} \circ \alpha_N(m_N) - \hat{\rho} \circ {\rm Ad}_{V_N}(m_N) \| = 0
\end{equation}
for any bounded sequence $m_N$. We can turn this into:
\begin{equation}
\label{distance}
\lim_{N \rightarrow \infty}  \| \rho \circ \alpha_N - \hat{\rho} \circ {\rm Ad}_{V_N} |_{\mathcal{L}_N} \| = 0
\end{equation}
by picking an appropriate sequence $m_N$. This is similarly true for $\hat{\rho} \rightarrow \hat{\sigma}$. 
Thus we can always find a monotonic decreasing function $\epsilon(N)$ with $\epsilon(N) \rightarrow 0$ such that:
\begin{equation}
\label{rsc}
\| \rho \circ \alpha_N - \hat{\rho} \circ {\rm Ad}_{V_N} |_{\mathcal{L}_N} \| , \| \sigma \circ \alpha_N - \hat{\sigma} \circ {\rm Ad}_{V_N} |_{\mathcal{L}_N} \|  < \epsilon(N)
\end{equation}
Pick the monotonic sequence $\delta_N =  \max\{ \epsilon_N, \epsilon(N) \}$ in  (ii) of Theorem~\ref{thm:JLMS}
then since for fixed $N$ we have
$\epsilon(N) \leq \delta_N$
this implies that  $\sigma \circ \alpha_N, \rho \circ \alpha_N$ are  included in the smoothing procedure for each $N$:
\begin{align}\label{tochain}
 \liminf_{N \rightarrow \infty} S_{\delta_N}( \hat{\rho} \circ {\rm Ad}_{V_N} |\hat{\sigma} \circ {\rm Ad}_{V_N}  ; \mathcal{L}_N)  &\leq  \liminf_{N \rightarrow \infty} S( \rho \circ \alpha_N |\sigma \circ \alpha_N  ; \mathcal{L}_N)  \\ & \leq \limsup_{N \rightarrow \infty} S( \rho \circ \alpha_N |\sigma \circ \alpha_N  ; \mathcal{L}_N)  \leq  S(\rho|\sigma; \mathcal{N} )
\end{align}
where in the last inequality we used monotonicity of relative entropy.  
After using \eqref{rhv} we get equality through the chain implying the limit of $ S( \rho \circ \alpha_N |\sigma \circ \alpha_N  ; \mathcal{L}_N)$ exists
and equals $S(\rho|\sigma; \mathcal{N} )$.  
\end{proof}

\begin{proof}[\bf Proof of Theorem~\ref{thm:JLMS}] 
Since condition (ii) and (iii) are asymmetric with respect to $\mathcal{N} \leftrightarrow \mathcal{N}'$
and $\mathcal{L}_N \leftrightarrow \mathcal{L}_N'$ the commutant versions are also part of the theorem statement since $(i) \iff (i)'$, using a hopefully clear notation. We will prove the following sequence of implications:   $(i) \implies (ii) \implies (iii)_a  \implies (iii)_b \implies (iii)_c  \implies (i)$.

\vspace{.4cm}
\noindent $(i) \implies (ii)$ Pick any two normalized states $\rho,\sigma \in \mathcal{B}(\mathscr{H})_\star^+$ as in the statement of (ii). 
Now from condition (i) and Theorem~\ref{thm:main} for $\mathcal{N}'$ we have the quantum channel $\alpha_N : \mathcal{L}_N \rightarrow \mathcal{N}$ with:
\begin{equation}
\lim_{N \rightarrow \infty}  \| \rho \circ \alpha_N - \rho \circ {\rm Ad}_{V_N} |_{\mathcal{L}_N} \| = 0
\end{equation}
and similarly for $\sigma$.  Hence there is some $\epsilon(N) > 0$ monotonic and limiting to zero which  satisfies \eqref{rsc}.
Pick $\epsilon_N = \epsilon(N)$ in statement (ii). Since for any $\delta_N \geq \epsilon_N$, also vanishing as $N \rightarrow \infty$, the pair $\rho \circ \alpha_N , \sigma \circ \alpha_N $ is included in the smoothing
procedure, so:
\begin{align}
 \nonumber
 \liminf_{N \rightarrow \infty} S_{\delta_N}( \rho \circ {\rm Ad}_{V_N} |\sigma \circ {\rm Ad}_{V_N}  ; \mathcal{L}_N)  & \leq  \liminf_{N \rightarrow \infty} S( \rho \circ \alpha_N |\sigma \circ \alpha_N  ; \mathcal{L}_N) \\ &\leq \limsup_{N \rightarrow \infty} S( \rho \circ \alpha_N |\sigma \circ \alpha_N  ; \mathcal{L}_N)  \leq S(\rho|\sigma; \mathcal{N} )
 \label{tochainb}
\end{align}
where the last inequality uses monotonicity of relative entropy.

 We next show that:
\begin{equation}
\label{tochaina}
\liminf_{N \rightarrow \infty}  S_{\delta_N}( \rho \circ {\rm Ad}_{V_N} \circ \beta_N |\sigma \circ {\rm Ad}_{V_N}  \circ  \beta_N ; \mathcal{N}) 
\geq S( \rho  |\sigma ; \mathcal{N}) 
\end{equation}
by an application of lower semi-continuity of the relative entropy in the weak topology of the predual $\mathcal{N}_\star$. See \eqref{lowersemi}.

We give some more details for \eqref{tochaina}: for all $\epsilon > 0$ by the definition of $\inf$ in $S_\delta$ \eqref{se} we know that there exists a sequence $\tilde{\rho}_N,\tilde{\sigma}_N  \in \mathcal{N}_\star^+$ with:
\begin{equation}
\label{eq:infe}
S( \tilde{\rho}_N |  \tilde{\sigma}_N )  < S_{\delta_N}( \rho \circ {\rm Ad}_{V_N} \circ \beta_N |\sigma \circ {\rm Ad}_{V_N}  \circ  \beta_N ; \mathcal{N}) + \epsilon
\end{equation}
and
\begin{equation}
\label{dcons}
\| \tilde{\rho}_N - \rho\circ {\rm Ad}_{V_N} \circ \beta_N  \|, \| \tilde{\sigma}_N - \sigma\circ {\rm Ad}_{V_N} \circ \beta_N  \|  < \delta_N
\end{equation}
This last property implies that:
\begin{equation}
\lim_{N \rightarrow \infty} \tilde{\rho}_N (n) = \rho(n), \qquad \lim_{N \rightarrow \infty} \tilde{\sigma}_N (n) = \sigma(n) \qquad \forall n\in \mathcal{N}
\end{equation}
due to the standard reconstructible statement in (i) and Lemma~\ref{lem:ab}.
Since we have pointwise convergence of the states we can apply the lower semi-continuity statement \eqref{lowersemi} to show:
\begin{equation}
\liminf_{N \rightarrow \infty}  S( \tilde{\rho}_N |  \tilde{\sigma}_N ) \geq S( \rho  |\sigma ; \mathcal{N}) 
\end{equation}
Combining with \eqref{eq:infe} we find:
\begin{equation}
S( \rho  |\sigma ; \mathcal{N})  < \liminf_{N \rightarrow \infty}  S_{\delta_N}( \rho \circ {\rm Ad}_{V_N} \circ \beta_N |\sigma \circ {\rm Ad}_{V_N}  \circ  \beta_N ; \mathcal{N})  + \epsilon
\label{seeafter}
\end{equation}
true for all $\epsilon>0$ so sending $\epsilon \rightarrow 0$ gives \eqref{tochaina}.
Equality through the chain of \eqref{tochaina}, \eqref{tochainb} gives condition (ii) for each such pair of normal states $\rho,\sigma$. We point out that at no point in this proof did we use the condition $\rho|_{\mathcal{N}} < \lambda \sigma|_{\mathcal{N}}$ so (ii) is true more generally than stated. However we will next use the weaker condition with the constraint $\rho|_{\mathcal{N}} < \lambda \sigma|_{\mathcal{N}}$ to prove (iii).

\vspace{.4cm}
\noindent $(ii) \implies (iii)'_a$. We have from Lemma~\ref{lem:preJLMS} for any two states $\rho < \lambda \sigma \in \mathcal{N}_\star$
\begin{equation}
\label{proven}
\lim_{N \rightarrow \infty} S( \rho \circ \alpha_N |\sigma \circ \alpha_N  ; \mathcal{L}_N) = S(\rho|\sigma; \mathcal{N} )
\end{equation}
where $\alpha_N$ is the $\varPhi_N$-dual map for $\beta_N'$, where $\varPhi_N$ was given in the preamble of the current theorem. 

\eqref{proven} allows us to apply the Petz sufficiency for nets \cite{ohya2004quantum,petz1994discrimination} (net here simply refers to a generalized sequence, not the net of algebras in AQFT). We follow this paper closely, although ultimately they have a different setup since they study ``asymptotic sufficiency'' for as little as two states. Here we get more stringent results since we have asymptotic sufficiency for all the states on the code subspace. 

\vspace{.2cm}
\noindent{\underline{Setup:}}
\vspace{.2cm}

For ease of notation we will write $J_N \equiv J_{\varPhi_N;\mathcal{L}_N}$ and $J = J_{\psi;\mathcal{N}}$.

We  take $\sigma = \omega_\eta|_{\mathcal{N}}$
and consider $\rho = \omega_\psi |_{\mathcal{N}}$ 
for some cyclic and separating vector $\psi$. The condition $\rho < \lambda \sigma$ implies that $\psi = n' \eta$
for some $n' \in \mathcal{N}'$. $\psi$ is cyclic and separating implies that if $n' \xi =0$ then $\xi =0$. Thus $n'$ is invertible. 
We additionally assume that $\psi$ is in the natural cone of $\eta$ for $\mathcal{N}$. We can of course do this while maintaining the form $\psi = n' \eta$ - for example we could apply an appropriate unitary in $\mathcal{N}'$ to take us to the natural cone.  

That is
$ \psi \in \mathscr{P}_{\eta;\mathcal{N}}$. %which is in the canonical cone for a unique unitary $\phi'$. 
Define the vectors $H_N, \widehat{\Psi}_N \in \mathcal{K}_N$:
\begin{equation}
\omega_{H_N}|_{\mathcal{L}_N} =\sigma \circ \alpha_N \qquad \omega_{\widehat{\Psi}_N}|_{\mathcal{L}_N} = \rho \circ \alpha_N
\end{equation}
where $H_N,\widehat{\Psi}_N \in \mathscr{K}_N$ are chosen in the natural cone associated to the cyclic and separating vector $\varPhi_N$ for $\mathcal{L}_N$.

Note that $V_N \psi$ need not be in this same natural cone as $\widehat{\Psi}_N$,
even if it induces approximately the same  state on $\mathcal{L}_N$.
%Instead define $\Psi_N \equiv \beta'_N(n') H_N$.
However there is a partial isometry $\Theta'_N \in \mathcal{L}_N'$ that takes it to the cone $\Theta'_N V_N \psi \in \mathscr{P}_{\varPhi_N,\mathcal{L}_N}$.  Where:
\begin{equation}
(\Theta'_N)^\dagger \Theta'_N
= \pi_{\mathcal{L}_N'}( V_N \psi)
\qquad \Theta'_N (\Theta'_N)^\dagger
= J_N \pi_{\mathcal{L}_N}( V_N \psi) J_N
\end{equation}
Define $\Psi_N = (\Theta'_N)^\dagger \widehat{\Psi}_N$ and note that:
\begin{equation}
\omega_{\widehat{\Psi}_N}(1-\Theta'_N (\Theta'_N)^\dagger)
= \rho \circ \alpha_N( 1- \pi_{\mathcal{L}_N}( V_N \psi))  \rightarrow 0
\end{equation}
where we used \eqref{alphaV} of the complement version of Theorem~\ref{thm:main}.
Hence:
\begin{equation}
\lim_N (1-\Theta'_N (\Theta'_N)^\dagger) \widehat{\Psi}_N = 0
\end{equation}
And this implies that: 
\begin{equation}
\lim_N \omega_{\Psi_N}|_{\mathcal{L}_N} - \rho \circ \alpha_N = \lim_N \omega_{\Psi_N}|_{\mathcal{L}_N} - \omega_{V_N \psi} = 0
\end{equation}
in the linear functional norm.  
One then shows that:
\begin{equation}
\label{equalH}
\lim_N \| \varPhi_N - H_N \| =0 
\qquad \lim_N \| V_N \psi - \Psi_N \| =0
\end{equation}
(recall also that $\lim_N \| V_N \eta - \varPhi_N \| =0$ by definition.)
We prove the above equation, in each case, using the fact that the two vectors inside the norm difference induce linear functionals on $\mathcal{L}_N$ that uniformly approach each other due to \eqref{alphaV} of the complement version of Theorem~\ref{thm:main} (with an appropriate choice of the sequence $m_N$.) The choice of natural cone discussed above then guarantees the vectors approach each other \cite{araki1974some}.\footnote{Explicitly $\widehat{\Psi}_N - \Theta_N' V_N \psi \rightarrow 0$ by this argument. Applying $(\Theta'_N)^\dagger$ gives the second term in \eqref{equalH}.}

Define the contractions $W_N^{\eta,\psi}$ via:
\begin{equation}
W_N^\psi \ell \left| \widehat{\Psi}_N \right> = \alpha_N(\ell) \left| \psi \right>\,
\qquad W_N^\eta \ell \left| H_N \right> = \alpha_N(\ell) \left| \eta \right>\qquad \forall \ell \in \mathcal{L}_N 
\end{equation}
and one can show that $(W_N^\psi)^\dagger W_N^\psi \left| \widehat{\Psi}_N\right> =  \left| \widehat{\Psi}_N\right>$
and also $(W_N^\eta)^\dagger W_N^\eta \left| H_N\right> =  \left| \H_N\right>$. 
Consider any bounded sequence $\ell_N \in \mathcal{L}_N$ then:
\begin{equation}
\label{aweaklimit}
w-\lim_N (W_N^\eta -V_N^\dagger ) \ell_N \left| H_N \right>
= w-\lim_N ( \alpha_N(\ell_N) \left| \eta \right>
- V_N^\dagger \ell_N V_N \left| \eta \right> )
=0
\end{equation}
where we use the weak Hilbert space limit.
The first equality is from \eqref{equalH}
and the second use (the complement version of) \eqref{alphaV}.

\vspace{.2cm}
\noindent{\underline{Asymptotic equality of Connes cocycle:}}
\vspace{.2cm}

Monotonicity of relative entropy under $\alpha_N$ can be expressed as the expectation in $\psi$ of an integral over a difference of two resolvents that are provably positive for each $t  \geq 0$:
\begin{equation}
\label{diffres}
(t+ \Delta)^{-1}  -W_N^\psi (t+ \Delta_N )^{-1} (W_N^\psi)^\dagger 
\end{equation}
where we defined the relative modular operators $\Delta_N = \Delta_{H_N, \widehat{\Psi}_N;\mathcal{L}_N}$ and $\Delta = \Delta_{\eta, \psi;\mathcal{N}}$. Due to \eqref{proven} and since \eqref{diffres} is positive and bounded for each $t > 0$ the expectation in $\psi$ must vanish in the limit. Positivity further gives: 
\begin{equation}
\lim_N W_N^\psi (t+ \Delta_N )^{-1} \left| \widehat{\Psi}_N \right> = (t+ \Delta)^{-1} \left| \psi \right>
\end{equation}
in the Hilbert space norm. One can now pick an appropriate integral over $t$ to extract the limit on relative modular flow:
\begin{equation}
\lim_N W_N^\psi \Delta_N^{is} \left| \widehat{\Psi}_N \right> = \Delta^{is} \left| \psi \right> \qquad \forall s \in \mathbb{R}
\end{equation}
Define the Connes cocycle $u_s^N = \Delta_N^{is} \Delta_{\widehat{\Psi}_N;\mathcal{L}_N}^{-is} \in \mathcal{L}_N$ and $u_s = \Delta^{is} \Delta_{\psi;\mathcal{N}}^{-is}$
so that:
\begin{equation}
\label{defineabove}
\lim_N \alpha_N(u_s^N) \left| \psi \right> = u_s \left| \psi \right> \qquad \forall s \in \mathbb{R}
\end{equation}
and so we can use the cyclicity property of $\psi$ for $\mathcal{N}'$ to extract:
\begin{equation}
%\label{strongusn}
\label{solim5}
so-\lim_N \alpha_N(u_s^N) = u_s 
\end{equation}
which is a result proven in \cite{petz1994discrimination}, where further details on the proof of \eqref{solim5} can be found.

\vspace{.2cm}
\noindent{\underline{Complex interpolation argument:}}
\vspace{.2cm}

Define the holomorphic vector:
\begin{equation}
\left| \Gamma_N(z) \right> =\Delta_{\eta,\psi}^{z} W_N^\psi \Delta_{H_N,\widehat{\Psi}_N}^{-z} \left| \widehat{\Psi}_N \right>
\end{equation}
where this is the same interpolation vector used in \cite{junge2018universal,Faulkner:2020iou,Faulkner:2020kit}.
In this part we will generally include the vector labels on the modular and relative modular operators, although we will drop the algebra label leaving only two possibilities $\mathcal{L}_N$ and $\mathcal{N}$ and these will be clear from the vector labels. We will also sometimes use the commutant algebra for the modular operators which we then label with a prime. 

This vector has some nice properties. It is holomorphic in the open strip $0 < {\rm Re}(z) < 1/2$ and strongly continuous on the closure of this strip. These results follow since: 
\begin{equation}
\label{itisbounded}
\Delta_{\eta,\psi}^{z} W_N^\psi \Delta_{H_N,\widehat{\Psi}_N}^{-z}
\end{equation}
is a bounded operator in this strip with norm $\leq 1$, see for example \cite{Faulkner:2020kit}.  We can evaluate it on the top of the strip:
\begin{equation}
\left| \Gamma_N(1/2+it) \right> = \Delta_{\eta,\psi}^{it}  (J W_N^\eta J_N)  \Delta_{H_N,\widehat{\Psi}_N}^{-it} \left| \widehat{\Psi}_N \right>
\label{Gahalf}
\end{equation}
and on the bottom of the strip we find: 
\begin{equation}
\left| \Gamma_N(it)\right> = \Delta_{\psi}^{it} u_{-t}^\dagger \alpha_N( u^N_{-t}) \left| \psi \right>
\end{equation}
where we have used the cocycles defined above \eqref{defineabove}. We see that:
\begin{equation}
\lim_N \left| \Gamma_N(it)\right> = \left| \psi \right> \qquad t \in \mathbb{R}
\end{equation}
Our goal now is to establish that this simple limit still applies in the holomorphic strip and in particular on the top edge $z = 1/2 + it$. This later task will be somewhat non-trivial. 
We first aim to compute the overlap $f_N(z) = \left( \psi, \Gamma_N(z) \right)$. This inherits the analyticity and continuity properties of $\Gamma_N(z)$. Furthermore since by assumption we have $\psi = n' \eta$ for some $n' \in \mathcal{N}'$, $\Delta'_{\psi,\eta}\equiv\Delta_{\psi,\eta;\N'}$ then:
\begin{equation}
f_N(z) = \left( (\Delta'_{\psi,\eta})^{-\bar{z}} n' \left| \eta \right>, W_N^\psi \Delta_{H_N,\widehat{\Psi}_N}^{-z} \left| \widehat{\Psi}_N \right> \right)
\end{equation}
for which the two vectors in the inner product are in the domains of the unbounded operator when $-1/2\leq {\rm Re} z \leq 0$, so indeed we can extend analyticity of $f_N(z)$ to $-1/2 < {\rm Re} z < 1/2$ with continuity on the closure. 
We know that $|f_N(z)| \leq 1$ in the top strip
$0 \leq {\rm Re} z \leq 1/2$ by the boundedness of the operator \eqref{itisbounded}. We have the estimate on the bottom strip:
\begin{equation}
|f_N(z) | \leq \max_{0\leq \theta \leq 1/2} \| (\Delta'_{\psi,\eta})^{\theta} n' \left| \eta \right> \| \| \Delta_{H_N,\widehat{\Psi}_N}^{\theta} \left| \widehat{\Psi}_N \right> \|\,,\qquad \theta = {\rm Re} z
\end{equation} 
which achieves a maximum since the supremum is over a compact space of a continuous function. 
We can estimate the function along the bottom edge $| f_N(-1/2+it) |^2 \leq \left< \psi \right| n' (n')^\dagger \left| \psi \right> \leq \| n' \|^2$ for all $t \in \mathbb{R}$ 
and on the top $| f_N(z) | \leq 1 \leq \| n' \|$ follows
from the boundedness of $\Gamma_N$ there. By the Phragmen–Lindelof theorem this implies
that $f_N(z)$ is uniformly (in $N$) bounded in the complex strip $-1/2 < {\rm Re} z < 1/2$. Hence by the Vitali-Porter theorem we have:
\begin{equation}
\lim_N f_N(z) = \left< \psi \right| \left. \psi \right> = 1
\end{equation}
uniformly in $z$ for all compact subsets $z \in S$ of the strip $S \subset \{ z: -1/2 < {\rm Re} z < 1/2\}$. This notably does not include the top of the strip. 
To deal with this we consider a new function:
\begin{equation}
g_N(z) = \left( \left| \psi \right>,
 J\big( \Delta_{\psi,\eta}^{-(\bar{z}-1/2)} W^\eta_N \Delta_{\widehat{\Psi}_N,H_N}^{(\bar{z}-1/2)} \ell \left| H_N \right> \big) \right)
\end{equation}
where $\ell = J_N \Theta'_N \beta'_N(n') J_N$. %($\beta_N'$ exist based on \ref{lem:preJLMS}).
Boundedness of the operator:
\begin{equation}
\Delta_{\psi,\eta}^{-(\bar{z}-1/2)} W^\eta_N \Delta_{\widehat{\Psi}_N,H_N}^{(\bar{z}-1/2)}
\end{equation}
for $0 \leq {\rm Re} z \leq  1/2$ guarantees analyticity inside the strip $0 < {\rm Re} z <  1/2$ and continuity to the edge. We can also manipulate the equation for $z = 1/2 + i t$ to find:
\begin{equation}
g_N(z) = \left( \Delta_{\widehat{\Psi}_N,H_N}^{\bar{z}-1/2}  \ell \left| H_N \right> ,
(W_N^\eta)^\dagger (\Delta_{\eta,\hat{\psi}}')^{(z-1/2)} \big| \hat{\psi} \big> \right)
\end{equation}
from which it is clear we can further analytically continue this function to the  strip $1/2 \leq {\rm Re} z < 1$.
We can evaluate this function on the edge $z=0 + it $ to find:
\begin{equation}
g_N(i t) = \left< \psi  \right| \Delta_{\eta,\psi}^{it}  W_N^\psi \Delta_{H_N,\widehat{\Psi}_N}^{-it} \Theta_N' \beta'(n') \left| H_N \right>
\end{equation}
We use these forms to prove that:
\begin{equation}
\lim_N (g_N(it) - f_N( it) ) =0
\qquad \lim_N (g_N(1/2+it) - f_N(1/2+ it) ) =0
\end{equation}
We can use these limits to conclude that $\lim_N(g_N(z) - f_N(z)) =0$ for all $0 < {\rm Re} z < 1/2$. See Figure~\ref{fig:f-g}. But since we know that the limit on $f_N(z)$ exists and equals $1$ we conclude that:
\begin{equation}
\lim_N( g_N(z) -1) = 0\qquad 0 < {\rm Re} z < 1/2
\end{equation}
which, using the Vitali-Porter theorem implies that the limit vanishes uniformly on compacts subsets of the holomorphic domain:
\begin{equation}
\lim_N( g_N(z) -1) =0 \qquad 0 < {\rm Re} z < 1
\end{equation}
Hence equality of $g_N$ and $f_N$ along $z=1/2+it$ gives us the desired limit:
\begin{equation}
\lim_N f_N(1/2 + it) = \lim_N g_N(1/2+it) = 1
\end{equation}
Hence, setting $t=0$ in $f_N(1/2+it)$ and using \eqref{Gahalf} we have:
\begin{equation}
\lim_N\left< \psi \right| (J W_N^\eta J_N)  \left| \widehat{\Psi}_N \right> = 1
\end{equation}
Approximating $\lim_N(\left| \Psi_N \right> - \beta_N'(n') \left| H_N \right>)=0$ with \eqref{equalH} and then using
the fact that $J_N \Theta'_N \beta_N'(n') J_N \in \mathcal{L}_N$ we can use the weak limit in \eqref{aweaklimit}
to show that:
\begin{equation}
\lim_N\left< \psi \right|  J V_N^\dagger  J_N \Theta'_N \beta_N'(n')J_N \left| \H_N \right> = 1
\end{equation}
which we then turn back into, again using \eqref{equalH}
\begin{equation}
\label{crossJ}
\lim_N\left< \psi \right|  J V_N^\dagger  J_N \Theta'_N V_N  \left| \psi \right>  = 1
\end{equation}
Or:
\begin{equation}
\label{thetaphi}
\lim_N \|   \Theta'_N V_N  \left| \psi \right>  - V_N   \left| \psi \right> \| = \lim_N\|  J_N \Theta'_N V_N  \left| \psi \right>  - V_N J \left| \psi \right> \| =0 
\end{equation}
by direct computation with \eqref{crossJ}. Plugging the first limit into the second we have:
\begin{equation}
0 =\lim_N  \|  J_N  V_N  \left| \psi \right>  - V_N J  \left| \psi \right> \| = \lim_N \|  J_N  V_N  \big| \psi \big>  - V_N J  \big| \psi \big>\|
\label{limitta}
\end{equation}
for all states $\psi$ in the natural cone 
and of the form $\psi = n' \eta$ for some $n'$.
The elements in $\mathcal{N}'_\eta \subset \mathcal{N}'$
that are entire analytic for the modular automorphism group $\sigma_{\eta;\mathcal{N}'}^s(n_0')$ form a strongly$^\star$ dense sub-algebra \cite{hiai2020concise}. Now $\Delta_{\eta}^{1/4} (n_0')^\dagger n_0' \left| \eta\right>= n' \left| \eta \right>$ for all $n_0' \in \mathcal{N}'_\eta$ is a dense subset of the positive cone $\mathscr{P}_{\eta;\mathcal{N}}$. The entirety allows us to write these states in the form $n' \left| \eta\right>$ which are the states that we made use of in the complex interpolation of $\Gamma_N$. 

\begin{figure}[h!]
\centering
\includegraphics[scale=.4]{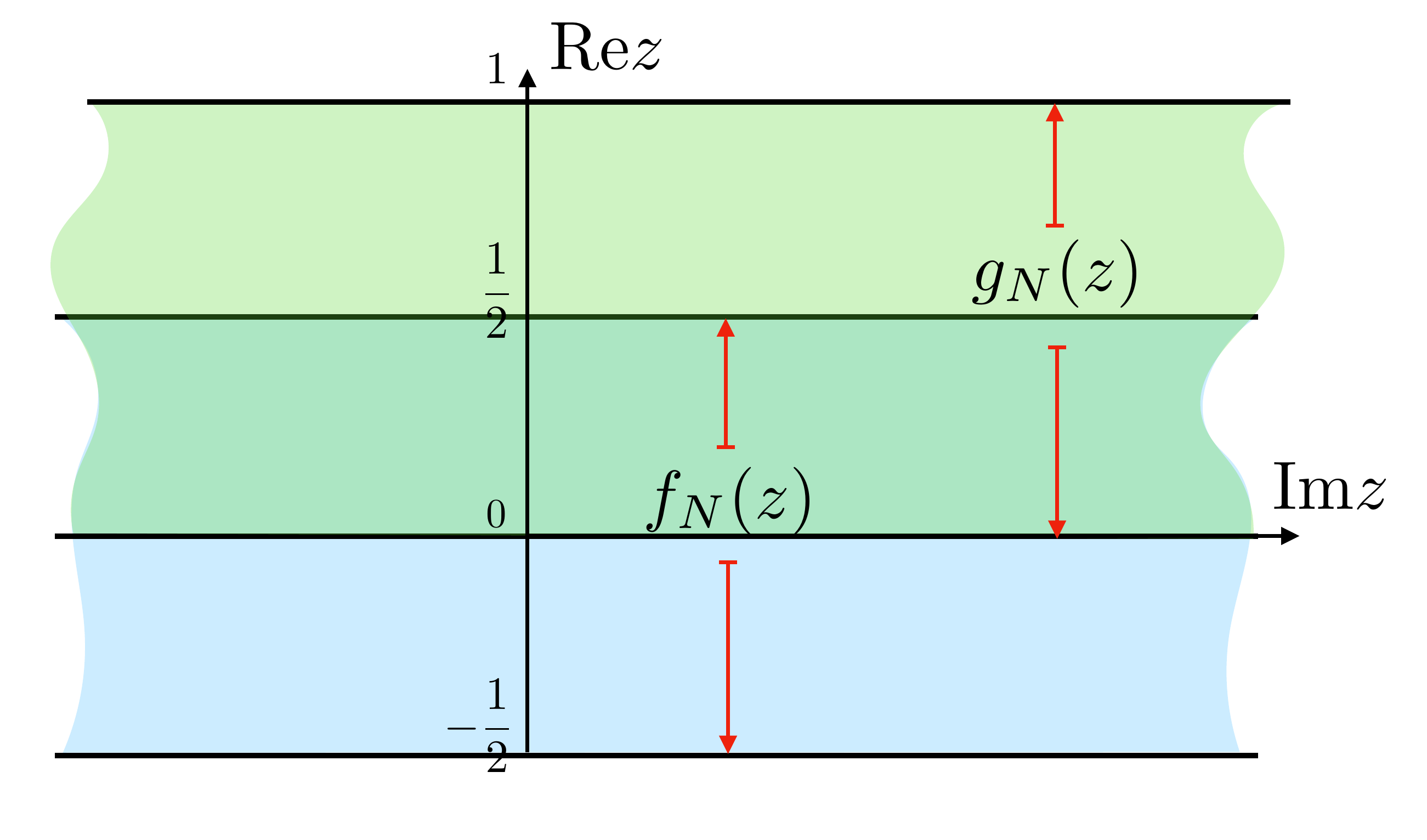}
\caption{The complex strips used in the proof for $(ii) \implies (iii)_a'$. The two functions are different but approach each other as $N \rightarrow \infty$.
Details of the proof, including using extended analytic strips with the same limiting behavior, have some intriguing similarities with \cite{Ceyhan:2018zfg}.
\label{fig:f-g}}
\end{figure}

Since every vector in $\mathscr{H}$ can be written as a sum of four elements in the positive cone $(\hat{\psi}_1 - \hat{\psi}_2) + i (\hat{\psi}_3 - \hat{\psi}_4)$ \cite{araki1974some} there is a dense subspace
of $\mathscr{H}$ where we use entire elements as above for each of the $\hat{\psi}_i$ in the natural cone. 
Thus using density (and, as usual uniform boundedness) we can turn \eqref{limitta} into the strong operator statement:
\begin{equation}
\label{modcon}
so-\lim_{N} \left( J_N V_N - V_N J \right) =0
\end{equation}

Set:
\begin{equation}
\beta_N( \cdot) \equiv J_N \beta_N'( J \cdot J) J_N
\end{equation}
It is easy to check that:
\begin{equation}
\label{newbeta}
\lim_N \beta_N(n) V_N - V_N n = 0
\end{equation}
for all $n \in \mathcal{N}$. We used \eqref{modcon}. %from $(iii)_a'$. 
[At this point we have actually shown that $(ii) \implies (iii)_c'$ and $(ii) \implies (i)$. However it is convenient to continue with $(ii) \implies (iii)_a'$ and prove the other parts follow from $(iii)_a'$.] 

Via \eqref{solim5} the cocycles map through the code
after using \eqref{alphaV} for the commutant.
\begin{equation}
wo-\lim_N V_N^\dagger \Delta_{H_N,\widehat{\Psi}_N}^{is} \Delta_{\widehat{\Psi}_N}^{-is} V_N = \Delta_{\eta,\psi}^{is} \Delta_\psi^{-is}
\end{equation}
Since these are unitaries we can turn this into:
\begin{equation}
so-\lim_N  \Delta_{H_N,\widehat{\Psi}_N}^{is} \Delta_{\widehat{\Psi}_N}^{-is} V_N - V_N \Delta_{\eta,\psi}^{is} \Delta_\psi^{-is}  =0
\end{equation}
by direct computation. We can pass the modular conjugation operator with \eqref{modcon} through these equations to find:
\begin{equation}
so-\lim_N  \Delta_{H_N}^{is} \Delta_{H_N,\widehat{\Psi}_N}^{-is} V_N - V_N \Delta_{\eta}^{is} \Delta_{\eta,\psi}^{-is}  =0
\end{equation}
Let us act this operator on $ \big| \psi \big>$
which gives:
\begin{equation}
\label{reltwo}
\lim_N (\Delta_{\widehat{\Psi}_N, H_N}^{is} V_N  -V_N \Delta_{\psi,\eta}^{is}) \big|\psi\big> = 0
\end{equation}
after using \eqref{thetaphi} and \eqref{equalH} repeatedly. 
Compute the inner product of \eqref{reltwo} with:
\begin{equation}
\beta_N'(n_2')^\dagger \beta_N'(n_1')^\dagger \Delta_{\widehat{\Psi}_N, H_N}^{is} V_N \big|\psi\big>
\end{equation}
for any two $n_1',n_2' \in \mathcal{N}'$ and take the limit leaving using \eqref{reltwo} on the second term:
\begin{align}
 \lim_N \big<\psi\big| V_N^\dagger \Delta_{H_N}^{is} \beta_N'(n_1') \beta_N'(n_2') \Delta_{H_N}^{-is} V_N
\big|\psi\big> =  \big<\psi\big| \Delta_{\eta}^{is} n_1' n_2' \Delta_{\eta}^{-is}
\big|\psi\big> 
\label{lastalign}
\end{align}
We can pass a unitary $u\in \mathcal{N}$ on the backside of the operators in this matrix element by inserting 
\begin{equation}
so-\lim_N \beta_N(u^\dagger) \beta_N(u) V_N - V_N =0
\end{equation}
on the right hand side and using \eqref{newbeta}. This allows us to move $\psi \rightarrow u \psi$ away from the natural cone. That is we were originally working with states in the natural cone of the form: $\psi = n' \eta = J n' J \eta$. Thus $\psi$ is also in the natural cone of the commutant with $\psi = n \eta$ for $n\in \mathcal{N}$. Picking a general $n_2 \in \mathcal{N}$
there is some unitary $u \in \mathcal{N}$ so that $u^\dagger n_2 \eta$ is in the natural cone. So we can use $u$ as above to work with $\psi = n_2 \eta$ for any $n_2 \in \mathcal{N}$ and such that:
\begin{equation}
 \lim_N \big<\psi\big| X_N 
\big|\psi\big> =0 \qquad X_N \equiv \left(V_N^\dagger \Delta_{H_N}^{is} \beta_N'(n_1') \beta_N'(n_2') \Delta_{H_N}^{-is} V_N - \Delta_{\eta}^{is} n_1' n_2' \Delta_{\eta}^{-is} \right)
\end{equation}
Consider:
\begin{align}
& 2\left<\psi_2 \right|X_N \left| \psi_1 \right> \\&= ( \left| \psi_1 \right> +  \left| \psi_2 \right>, X_N \left| \psi_1 \right> +  \left| \psi_2 \right>) 
- \left< \psi_1 \right| X_N \left| \psi_1 \right> -  \left< \psi_2 \right| X_N \left| \psi_2 \right> 
\\ & + i ( \left| \psi_1 \right> + i  \left| \psi_2 \right>, X_N \left| \psi_1 \right> +  i \left| \psi_2 \right>) 
- i \left< \psi_1 \right| X_N \left| \psi_1 \right> - i \left< \psi_2 \right| X_N \left| \psi_2 \right> \rightarrow 0
\end{align}
for all $\psi_{1,2} = n_{1,2} \eta$ for $n_{1,2} \in \mathcal{N}$. Since this is a dense subspace we can use uniform boundedness to conclude that $wo-\lim_N X_N =0$:
\begin{equation}
\label{toreturn}
wo-\lim_N V_N^\dagger \Delta_{H_N}^{is} \beta_N'(n_1') \beta_N'(n_2') \Delta_{H_N}^{-is} V_N - \Delta_{\eta}^{is} n_1' n_2' \Delta_{\eta}^{-is}=0
\end{equation}
which then implies:
\begin{equation}
\label{modflowproof2}
so-\lim_N  \Delta_{H_N}^{is} \beta_N'(n') \Delta_{H_N}^{-is} V_N - V_N \Delta_{\eta}^{is} n' \Delta_{\eta}^{-is}=0
\end{equation}
for all $n' \in \mathcal{N}$ and $s \in \mathbb{R}$.
Act with this equation on $\eta$ and use cyclicity and uniform boundedness to show that:
\begin{equation}
so-\lim_N  \Delta_{H_N}^{is}  V_N - V_N \Delta_{\eta}^{is}=0
\end{equation}
and by the continuity Lemma~\ref{lem:contmod} $(i) \implies (ii)$ we can change this to:
\begin{equation}
   so-\lim_N \Delta_{\varPhi_N}^{is}  V_N - V_N \Delta_{\eta}^{is}=0
\end{equation}
for all $s\in \mathbb{R}$ for $\varPhi_N$ in the theorem statement. 

\vspace{.4cm}
\noindent $(iii)_a \implies (iii)_b$ is clear. 

\vspace{.4cm}
\noindent $(iii)_b \implies (iii)_c$.

Fix some $\forall\, n'_{1,2}\in \mathcal{N}'$.
		Define the sequence of analytic functions with respect to a cyclic and separating vector $\eta\in \Hil$ for $\N$ where also $\varPhi_N$ is cyclic and separating for $\mathcal{L}_N$:
	\begin{align} %\label{F(s)}
	   &F_N(s)=%\langle\eta|V_N^\dagger 
    \left< \varPhi_N \right| \beta_N'(n_1') \Delta_{\varPhi_N;\LL_{N}}^{is} \beta_N'(n_2') %V_N|\eta\rangle
    \left| \varPhi_N \right>,\qquad 0\leq \mathrm{Im}s\leq 1
	\end{align}
which by standard modular theory are uniformly bounded analytic functions on the interior strip and continuous in the closure. The boundary values on the top boundary of the strip ${\rm Im} s=1$  is determined by the KMS condition and reads:
	\begin{align} %\label{F(s)}
	   &F_N(i+t)=%\langle\eta|V_N^\dagger
    \left< \varPhi_N \right| \beta_N'(n_2') \Delta_{\varPhi_N;\LL_{N}}^{-it} \beta_N'(n_1') %V_N|\eta\rangle
    \left| \varPhi_N \right>,\qquad t \in \mathbb{R}
	\end{align}
The limit of $\lim_N F_N$ exists at the bottom of the strip $s = t \in \mathbb{R}$ and at the top $s= it$ for $t \in \mathbb{R}$ (pointwise). This is established by replacing the matrix elements of $\varPhi_N$ with $V_N \eta$ and using $(iii)_b$. Hence the limit exists in the middle of the strip (for example we can write an integral equation for $F_N(s)$ in terms
of its boundary values and the limit of these integrals converges by dominated convergence.)
Thus we have:
\begin{equation}
\lim_N F_N(s) = F(s) = \langle\eta|n_1' \Delta_{\eta;\N}^{is} n_2'|\eta\rangle
\end{equation}
uniform on compact subsets of the interior of the strip and pointwise on the edge. The function $F(s)$ has the same analytic properties as $F_N(s)$.

	Setting $s=i/2$ and approximating $V_N \eta$ by $\varPhi_N$ we find:
	\begin{align}
	   \lim_{N\rightarrow\infty} \langle\eta|n_1' V_N^\dagger J_{\varPhi_N;\LL_{N}}V_N (n_2')^\dagger|\eta\rangle = \langle\eta|n_1'  J_{\eta;\N} (n_2')^\dagger|\eta\rangle 
	\end{align}
	considered the cyclicity of $\eta$ for $\N'$, and uniform boundedness
	one has
	\begin{align}
	   wo-\lim_{N\rightarrow\infty}V_N^\dagger J_{\varPhi_N;\LL_{N}}V_N- J_{\eta,\N} =0
	\end{align}
	which is equvalent to
	\begin{align}\label{JV-VJ}
	    so-\lim_{N\rightarrow\infty} J_{\varPhi_N;\LL_{N}}V_N-V_N J_{\eta,\N} =0
	\end{align}
 as required.

\vspace{.4cm}
\noindent $(iii)_c \implies (i)$
Define:	
	\begin{align}
	   \beta_N=j_{\varPhi_N;\LL_{N}}\circ\beta'_N\circ j_{\eta,\N},
	\end{align}
where $j(\cdot) = J \cdot J$ etc. 
We can directly confirm that:
\begin{equation}
so-\lim_N \beta_N(n) V_N - V_N n =0 \qquad \forall\, n \in \mathcal{N}
\end{equation}
using $(iii)_c$. 
Hence $\N\rightarrowtail \LL_{\N}$ as required. 

The postamble of the theorem has already been established at intermediate steps above. For the statement about relative entropy see Lemma~\ref{lem:preJLMS}. Also for the construction in $(ii) \implies (iii)$ we could have used
any $\psi$ cyclic and separating for $\mathcal{N}$ represented on the code by $\Psi_N$ cyclic and separating for $\mathcal{L}_N$.
\end{proof}

We needed:

\begin{lemma}[Continuity of modular operator]
\label{lem:contmod}
Under the conditions of the preamble to Theorem~\ref{thm:JLMS} with $\mathcal{N} \rightarrowtail \mathcal{L}_N$ and:
\begin{equation}
\label{contJ}
so-\lim_N J_{\varPhi_N;\mathcal{L}_N}
V_N - V_N J_{\eta;\mathcal{N}} =0
\end{equation}
For any sequence of vectors $H_N %\in \mathscr{P}_{V_N\eta;\mathcal{L}_N}
$, cyclic and separating for $\mathcal{L}_N$, such that $H_N - \varPhi_N \rightarrow 0 $ then the following statements are equivalent: 
\begin{itemize}
\item[(i)]
\begin{equation}
%\label{contmod}
so-\lim_{N \rightarrow \infty} (\Delta_{ H_N; \mathcal{L}_N}^{is}  V_N - V_N \Delta_{\eta;\mathcal{N}}^{is}) =0 \,, \qquad \forall s \in \mathbb{R}
\label{modi}
\end{equation}
\item[(ii)] \begin{equation}
so-\lim_{N \rightarrow \infty} (\Delta_{ \varPhi_N; \mathcal{L}_N}^{is} V_N - V_N \Delta_{\eta;\mathcal{N}}^{is}) =0  \,, \qquad \forall s \in \mathbb{R}
\label{modii}
\end{equation}
\end{itemize}
\end{lemma}
\begin{proof}

We cannot directly use existing results in the literature, see for example \cite{araki1977relative}, since the modular operators above do not converge to a fixed operator. However we can still follow very closely \cite{araki1977relative}.  The transformation between strong resolvent convergence and modular operator convergence are used back and forth, the discussion of equivalence between strong resolvent convergence and dynamical convergences in chapter ten of \cite{alma99954910089005899} proved to be very useful. In the modular operators below we will drop the algebra label which is either $\mathcal{L}_N$ or $\N$ - and which one it is should be clear from the state labels.

Suppose that $\Theta_N' H_N \equiv \hat{H}_N \in \mathscr{P}^\natural_{\varPhi_N;\mathcal{L}_N}$ %then $J_{V_N \eta} H_N = J_{V_N \eta}  (\Theta_N')^\dagger J_{V_N \eta}  \Theta_N' H_N $ 
then:
\begin{equation}
\| \Theta_N' H_N -\varPhi_N \|
\leq \| (\omega_{\Theta_N' H_N} - \omega_{\varPhi_N})|_{\mathcal{L}_N} \|
\leq \| H_N - \varPhi_N\|^{1/2}
\end{equation}
where we used the equivalence of the linear functional distance and the distance for representatives on the natural cone \cite{araki1974some}.
Using the fact that both $H_N,\varPhi_N$ approximate $V_N \eta$ we have:
\begin{equation}
\label{thetaN}
\Theta_N' V_N \eta - V_N  \eta \rightarrow 0 \qquad (\Theta_N' )^\dagger V_N \eta - V_N  \eta \rightarrow 0
\end{equation}
Applying $\beta_N(n)$ and using density and uniform boundedness:
\begin{equation}
s^{\star}o- \Theta_N' V_N  - V_N =0
\end{equation}
We can use this to show that
if this Lemma is true with  $\hat{H}_N \in \mathscr{P}^\natural_{\varPhi_N;\mathcal{L}_N}$ then it is also true without this condition. For this we use
$\Delta_{H_N}^{is}  = (\Theta_N')^\dagger \Delta_{\hat{H}_N}^{is} \Theta_N'$.

Hence we can proceed by assuming that $H_N \in \mathscr{P}^\natural_{\varPhi_N;\mathcal{L}_N}$. Our first task is to use the limiting behavior of $H_N$ along with \eqref{contJ} to prove the limiting equality of the resolvent operators. 
Based on properties of Tomita operator, for all bounded sequence of $m_N \in \LL_N$
	\begin{eqnarray}
	    &\norm{\Delta_{H_N}^{1/2}m_N \left| H_N \right>- \Delta_{\varPhi_N}^{1/2}m_N \left| \varPhi_N \right>}=\norm{ J_{\varPhi_N} \Delta_{H_N}^{1/2}m_N H_N-J_{\varPhi_N} \Delta_{\varPhi_N}^{1/2}m_N \varPhi_N}\non\\
	    &\leq\| m^\dagger_N\|\norm{H_N-\varPhi_N}\rightarrow0,\qquad 
	\end{eqnarray}
	hence
	\begin{align}\label{norm0.1}
	    &\norm{\left((i+\Delta_{\varPhi_N}^{1/2})^{-1}-(i+\Delta_{H_N}^{1/2})^{-1}\right)(i+\Delta^{1/2}_{H_N})\beta_N(n)H_N}\non\\
	    =&\left\|(i+\Delta_{\varPhi_N}^{1/2})^{-1}\left(\Delta^{1/2}_{H_N}\beta_N(n)H_N-\Delta_{\varPhi_N}^{1/2}\beta_N(n) \varPhi_N + i \beta_N(n)(H_N-\varPhi_N)\right) \right. \non\\ &\qquad \qquad  \left. -\beta_N(n)(H_N-\varPhi_N)\right\|\non\\
	    \leq&  \norm{\left(\Delta^{1/2}_{H_N}\beta_N(n)H_N-\Delta_{\varPhi_N}^{1/2}\beta_N(n)\varPhi_N \right)}+2\norm{\beta_N(n)}\norm{H_N-\varPhi_N}
	\end{align}
 Implying the limit of the left hand side of the above equation vanishes. 

Use $\lim_N\norm{(\beta_N(n^\dagger)H_N-V_Nn^\dagger \eta)}\rightarrow 0$, and also \eqref{contJ},
%$so-V_NJ_{\eta}-J_{V_N\eta}V_N=0$
then one has
	\begin{eqnarray}
	0&= \lim_{N} (J_{\Omega_N}^2\Delta^{1/2}_{H_N}\beta_N(n)H_N-V_NJ_{\eta}^2 \Delta_{\eta}^{1/2}n\eta) \non \\ &=
   \lim_{N} (\Delta^{1/2}_{H_N}\beta_N(n)H_N - V_N \Delta^{1/2}_{\eta}n\eta) 
	\end{eqnarray}
	 so that the vanishing of the left hand side of (\ref{norm0.1}) becomes:
	\begin{eqnarray}
	   \lim_{N\rightarrow \infty}\norm{\left((i+\Delta_{H_N}^{1/2})^{-1}-(i+\Delta_{\varPhi_N}^{1/2})^{-1}\right)V_N (i+\Delta^{1/2}_{\eta})n\eta} =0,
	\end{eqnarray}
Since $\N \eta$ is a core for  closed  unbounded operator
$\Delta_{ \eta}^{1/2}$ for any vector $\psi \in D(\Delta_{\eta}^{1/2})$ in the domain there is a sequence of $n_j \in \mathcal{N}$ such that:
\begin{equation}
\| n_j \left| \eta \right> - \left|\psi \right>\|^2 + \| \Delta_{\eta}^{1/2}(n_j \left| \eta \right> - \left|\psi \right>)\|^2 \rightarrow 0
\end{equation}
Given any $\xi \in \mathscr{H}$ define:
$\psi =(i+ \Delta_{\eta}^{1/2})^{-1} \xi$
which is then clearly in $\psi \in D(\Delta_{ \eta}^{1/2})$ so that the corresponding sequence satisfies:
\begin{equation}
\lim_j (i + \Delta_{\eta}^{1/2}) n_j \left| \eta \right> 
= \left| \xi \right> 
\end{equation}
Thus $(i + \Delta_{\eta}^{1/2})\mathcal{N}$ is dense in $\mathscr{H}$.  Then
	\begin{align}
 \label{iresolv}
	   so-\lim_{N\rightarrow\infty} \left((i+\Delta_{H_N}^{1/2})^{-1}-(i+\Delta_{\varPhi_N}^{1/2})^{-1}\right)V_N=0
	\end{align}
A similar argument gives the ``complex conjugrate'' $i \rightarrow -i$ statement.
This is the desired resolvent convergence statement. 

\noindent $(i) \implies (ii)$. Next we use the assumed convergence of modular flows in \eqref{modi} to pass to resolvent convergence in order to apply \eqref{iresolv}.

We integrate under this limit as follows. For any $\psi \in \mathscr{H}$
\begin{equation}
\lim_N \left\| \int_{-\infty}^{\infty} \frac{ ds}{\cosh{\pi s}} (\Delta_{H_N}^{is}  V_N - V_N \Delta_{\eta}^{is}) \left| \psi \right> \right\|  
\leq \lim_N \int_{-\infty}^\infty \frac{ds}{{\cosh{\pi s}}} \left\|(\Delta_{H_N}^{is}  V_N - V_N \Delta_{\eta}^{is}) \left| \psi \right> \right\|  =0
\end{equation}
where we pass the limit inside using dominated convergence with $\| (\Delta_{H_N}^{is}  V_N - V_N \Delta_{\eta}^{is}) \left| \psi \right> \| < 2C$. 
This then implies:
\begin{equation}
\label{resolvh}
so-\lim_N  (1+ \Delta_{H_N})^{-1} \Delta_{H_N}^{1/2} V_N - V_N (1+\Delta_{\eta})^{-1} \Delta_{\eta}^{1/2}=0
\end{equation}
Acting on $ n \left| \eta \right>$ for any $n \in \mathcal{N}$ and using $V_N n \eta - \beta_N(n) H_N \rightarrow 0$ on the first term twice and $J_N V_N - V_N J \rightarrow 0$  we can turn this into:
\begin{equation}
\lim_N \left( (1+ \Delta_{H_N})^{-1} V_N  - V_N (1+\Delta_{\eta})^{-1}\right) J n^\dagger \left| \eta \right>=0
\end{equation}
one more application of cyclicity for $n^\dagger \eta$ and uniform boundedness we have resolvent convergence:
\begin{equation}
\label{resolvc}
so-\lim_N  (1+ \Delta_{H_N})^{-1} V_N - V_N (1+\Delta_{\eta})^{-1}=0
\end{equation}
 adding/subtracting  the two limits $i \times$ \eqref{resolvc}
and \eqref{resolvh} we get resolvent convergence for $\Delta_{H_N}^{1/2}$: 
\begin{align}
\label{resolvch}
&so-\lim_N  ( i+ (\Delta_{H_N})^{1/2})^{-1} V_N - V_N ( i+(\Delta_{\eta})^{1/2})^{-1}=0
\end{align}
which, in conjunction with \eqref{iresolv} 
 and taking joint gives:
\begin{align}
\label{resolvch2}
&so-\lim_N  (\pm i+ (\Delta_{\Omega_N})^{1/2})^{-1} V_N - V_N (\pm i+(\Delta_{\eta})^{1/2})^{-1}=0
\end{align}
given that for any $\psi \in \mathscr{H}$:
	\begin{align}
	    &\lim_N\norm{\left((1+\Delta_{\Omega_N})^{-1} V_N-V_N (1+\Delta_{\eta})^{-1}\right)|\psi\rangle}\non\\
	    \leq&\lim_N\norm{(i+\Delta_{\Omega_N}^{1/2})^{-1}\left((-i+\Delta_{\Omega_N}^{1/2})^{-1} V_N- V_N (-i+\Delta_{\eta}^{1/2})^{-1}\right)|\psi \rangle}\non\\
	    &+\lim_N\norm{\left((i+\Delta_{\Omega_N}^{1/2})^{-1} V_N- V_N (i+\Delta_{\eta}^{1/2})^{-1}\right)(-i+\Delta_{\eta}^{1/2})^{-1}|\psi\rangle}\non
	\end{align}
	which implies that
\begin{align}
\label{topoly}
    so-\lim_N  (1+ (\Delta_{\Omega_N}))^{-1} V_N - V_N (1+(\Delta_{\eta}))^{-1}=0
\end{align}

The final task is to reverse this process to pass back to the modular flow statement of the theorem.  Now \eqref{topoly} is true for any polynomial of this resolvent since:
\begin{equation}
so-\lim_N  (1+ (\Delta_{\varPhi_N}))^{-k} V_N - V_N (1+(\Delta_{\eta}))^{-k}=0
\end{equation}
for all $k\in \mathbb{Z}_{\geq 0}$. Hence we can use the Stone–Weierstrass Theorem to approximate any continuous function $f$ on the compact interval $[0,1]$. So for all $\epsilon >0 $ there is some finite sequence $\{p_k\}_{k=0,\ldots K}$ such that:
\begin{equation}
\| f( (1+ (\Delta_{\varPhi_N}))^{-1} ) - \sum_k^K p_k (1+ (\Delta_{\varPhi_N}))^{-k} \| \leq \epsilon
\end{equation}
from which it is possible to show that:
\begin{equation}
so-\lim_N  f((1+ (\Delta_{\varPhi_N}))^{-1}) V_N - V_N f((1+(\Delta_{\eta}))^{-1})=0
\end{equation}
Consider a set of smooth window functions $0 \leq w_\Lambda(x)\leq 1$ on the closed interval $[0,1]$ parameterized by $\Lambda >0$.
We take these to be monotonically increasing as $\Lambda \rightarrow \infty$ and converging pointwise to $1$ on the open interval $(0,1)$. We pick them to be (smoothly) vanishing at $w_\Lambda(0)= w_\Lambda(1) =0$ for all fixed $\Lambda$. 
Define the functions $f_\Lambda(x) = (-1+1/x)^{i s} w_\Lambda(x)$ for $0< x <1$ and $f_\Lambda(0)=f_\Lambda(1)=0$.
These are continuous on the interval $[0,1]$ (they limit to zero at $x=0,1$ since $(-1+1/x)^{i s}$ is bounded.)  Define $R_N = ( 1 + \Delta_{\varPhi_N})^{-1}$ and $R = (1+ \Delta_\eta)^{-1}$. Fix some normalized $\psi \in \mathscr{H}$.
Given the above construction we note that:
\begin{equation}
\lim_{\Lambda \rightarrow \infty} \| (1- w_\Lambda(R) ) \left| \psi \right> \| =0
\end{equation}
by dominated convergence in the spectral representation of the resolvent. Hence for all $\epsilon >0$ there is some $\Lambda(\epsilon)$
with $ \| (1- w_\Lambda(R) ) \left| \psi \right> \| < \epsilon$ for all $\Lambda > \Lambda(\epsilon)$.  

Then:
\begin{align}
& \| (\Delta_{\varPhi_N}^{is}  V_N - V_N \Delta_{\eta}^{is}) \left| \psi \right> \|
\leq  \| (f_\Lambda(R_N)  V_N - V_N f_\Lambda(R)) \left| \psi \right> \|
\\& \qquad \qquad  +\| \Delta_{\varPhi_N}^{is} (1-w_\Lambda(R_N))  V_N \left| \psi \right> \|  + \|  V_N \Delta_{\eta}^{is} (1-w_\Lambda(R)) \left| \psi \right> \|
\end{align}
Taking the limit:
\begin{align}
& \limsup_N \| (\Delta_{\varPhi_N}^{is}  V_N - V_N \Delta_{\eta}^{is}) \left| \psi \right> \|
\leq    2 C \|  (1-w_\Lambda(R)) \left| \psi \right> \| \leq 2 C \epsilon
\end{align}
for all $\Lambda > \Lambda(\epsilon)$. Since $\epsilon$ was arbitrary to begin with we can take the limit $\epsilon \rightarrow 0$. We arrive at the desired result \eqref{modii}.

$(ii) \implies (i)$ Uses a very similar analysis as above with $H_N  \leftrightarrow \varPhi_N$. 
\end{proof}

\section{Discussion}
\label{sec:d}

In this paper we have put forward a new mathematical framework 
for studying AdS/CFT.  
This framework involves a confluence of ideas known to underpin holography: including
large-$N$ limits of matrix theories \cite{Banks:1996vh}, quantum error correcting codes \cite{Almheiri:2014lwa}, the paradigm of bulk reconstruction \cite{Hamilton:2006az,Bousso:2012sj,Czech:2012bh}, and the algebraic approach to QFTs \cite{Duetsch:2002hc,Casini:2019kex}. 
The asymptotically isometric codes are approximate quantum error correcting codes that confront the large-$N$ limit directly. We showed that these
codes have many properties in relation with gravitational theories.
We also very nearly derived these codes from the expected behavior of large-$N$ correlation functions in matrix like theories such
as $\mathcal{N}=4\,SU(N)$ SYM. The main unproven assumption on this derivation was a plausible strengthening of the mode of convergence of large-$N$
correlation function, see Definition~\ref{def:fc}. We gave some arguments for this assumption in Appendix~\ref{app:uniform}, based on an assumed uniform type of convergence on fixed operator systems that then generate the full algebra. These operator systems should presumably be constructed using restrictions on the the amplitude and frequency of the operator smearing functions. 
We expect we can show these operator systems enjoy uniform convergence based on some further input on the limiting behavior of correlation functions in large $N$ theories.

A feature of these codes is that the code subspace is the same for each fixed $N$ theory. In fact the code subspace can be thought of as the $N=\infty$ theory and the physical
Hilbert spaces are the QFTs at fixed integer $N$. 
This is the only possibility that  really makes sense for the physical Hilbert space.
Some recent discussion of such large-$N$ limits has discussed a rather different bulk to boundary dictionary \cite{Leutheusser:2021frk,Leutheusser:2021qhd} that in our story would be a tautology.\footnote{One might complain that exact dualities, such as might be expected for AdS/CFT, are exactly this kind of tautology.  In particular the algebraic approach tends to be blind to which duality frame you work in, so perhaps the tautological nature is natural here. Certainly we agree that at fixed $N$ if we can make sense of string theory in asymptotically AdS spaces as a quantum theory in its own right then the duality will be an exact duality. However even with this in hand it might not be obvious how to arrive at a local bulk description in certain limits. After all if the duality is tautological then we already know how to define string theory in asymptotically AdS spacetime - it is non-abelian gauge theory without gravity. Hence one of the more interesting questions is how such an effective theory of gravity can consistently arise from such a quantum stringy description. And in particular how bulk locality can emerge from this. We expect subregion-subregion duality to be one of the most constraining local properties that arises in the limit of small $G_N$ and this is exactly what our codes are studying. } That is, in those papers the bulk and boundary are effectively the same $N=\infty$ theory - perhaps with some formal $1/N$ power series corrections. Such $1/N$ correction will  not change the fact that a well defined error correcting code should be a mapping from this bulk theory to a fixed finite $N$ physical theory.
Running with this line of thought then naturally leads to the asymptotic isometric codes. It is also not so surprising that we need to have a code of this nature, with the microscopic fixed $N$ theories made explicit, in order to make sense of entanglement wedge reconstruction. Otherwise, working purely from a bulk perspective it should not be possible to discover the true meaning of the entanglement wedge - instead one is limited by the causal structure of the bulk theory. \footnote{The rules of euclidean quantum gravity do give a meaning to the entanglement wedge via the replica trick and subsequent developments. Error correcting codes seek a Hilbert space interpretation of such rules and so these cannot be an input.}

One goal that we have not quite achieved is the ability to compute entropies on these codes. Most boundary and corresponding bulk algebras of interest will be type-III$_1$ algebras and so do not have well defined entropies.
The crossed product construction of \cite{Witten:2021unn}, gives a type-II$_\infty$ factor for one half of the thermofield double state and this allows one to define entropies, up to an overall additive shift.  The microscopic entropies at fixed $N$ are also well defined because the algebra is type-I$_\infty$ in this case. In particular we then expect that our codes can be used to prove the entropy duality in that case. We expect to derive a formula where the the finite $N$ entropies limit (after an appropriate $N$ dependent subtraction) to the type-II$_\infty$ entropy. 

Aside from computing entropies, the most important next step is understanding the different possible sectors that arise in the large-$N$ limit. We summarize some work in progress and future work along these lines next.

\subsection{Different sectors}

As alluded to in this paper it is not possible that the asymptotically isometric codes we have constructed could encode for all possible bulk physics in AdS/CFT.
In particular the construction has been limited to perturbative quantum fluctuations about some classical geometry. Back reaction effects could possibly be included using
$1/N$ corrections, allowing for sources with $N$ dependence to shift the background and hence change the leading algebra. This seems to call for some new understanding of such formal perturbative algebras. Instead we might simply work with a new code for this shifted background and attempt to sew these codes together. At the end of the day one might hope for a code that is somehow fibered over the space of classical solutions, varying as the bulk solution varies. This also seems to require some new paradigm. Certain similarities to deformation quantization \cite{rieffel1989deformation,kontsevich2003deformation,binz2004field}\footnote{We think Antony Speranza for pointing this out and for various discussions on this subject.} are apparent where a similar dichotomy between formal $\hbar$ expansions and strict deformation quantization where such series expansions are not required. Before even studying this however, it is possible to focus on some special codes and learn as much from these as possible. 

In a paper under preparation \cite{toappear} we study vacuum codes and thermal codes in some detail. The main results we would like to advertise are (see Section~\ref{sec:examples} for a prelimniary definition of the codes below):
\begin{enumerate}
\item For a vacuum code, which is defined to have a representation of the conformal symmetries with positive energy, we prove that $\mathcal{C}(\diamondsuit) = \mathcal{E}(\diamondsuit)$, implying that the causal wedges are automatically Haag dual for double cone regions. Under the assumption that a code has a representative vector $\eta$ for the vacuum sequence $\Omega_N$, we prove that $\mathcal{C}(\diamondsuit) = \mathcal{E}(\diamondsuit)\,, \forall \diamondsuit$ is a necessary and sufficient condition for a code to be a vacuum code.  

\item For thermal codes, assuming that there exists a time translation symmetry acting on the code subspace that maps equivariantly to the microscopic time translation symmetry, we can prove that the code time translation satisfies the KMS condition and that $\mathcal{E}(K_0) = \mathcal{C}(K_0)$ where $K_0$ is $S^{d-1} \times \mathbb{R}$ (if the code is actually a thermofield double code then $K_0$ is either the left or right cylinder.)

\item Time shifted thermal codes are disjoint to the regular thermal codes if we are above the deconfinement temperature. And as long as we work in the microcanonical  thermal code with fixed area then we can build a larger code patching together all the time shifted codes. This extended code realizes the crossed product construction of \cite{Witten:2021unn} as the bulk algebra. Below the deconfinement temperature these codes are not disjoint and the time translation symmetry is represented as an inner action with a Hamiltonian that is affiliated to $\mathcal{C}(K_0)$. 

\end{enumerate}

Further studies are necessary. So called split codes, defined in the exact setting in \cite{Faulkner:2020hzi}, should facilitate studies of phase transitions between different codes. It should be possible to give a condition for a code
to be asymptotically vacuum, analogous to the asymptotically AdS spacetimes relevant to AdS/CFT. We would then make this a basic requirement for any code. 

\subsection{Bulk locality versus stringy codes}
\label{sec:bulklocality}

One important open question is to understand the conditions under which the codes we constructed are stringy and how to resolve such stringy physics. 
This question is related to the emergence of a local bulk theory. Some ideas/questions for attacking this problem include, studying the saturation of the modular chaos bound \cite{Faulkner:2018faa}, confronting the question of the split property \cite{Doplicher:1984zz} for the bulk algebras, and giving a better understanding for when the causal wedges can be extended to Haag dual nets.  The later question is potentially related to stringy physics since  there seems to be some difficulty in defining stringy notions of an entanglement wedges \cite{Chandrasekaran:2021tkb,Chandrasekaran:2022qmq}.

We end with a conjecture that ties these loose ends together. Suppose that a code is asymptotically vacuum and the single trace fields include a boundary stress tensor.
We firstly conjecture that the split property for $\mathcal{C}$ is tied to the existence of a local quantum field theory. The motivation for this is as follows. Without the
split the thermodynamics of the theory cannot be that of a local QFT (on some potentially higher dimensional space). On the other hand, with the split one can attempt to reconstruct
a local quantum like theory: we firstly conjecture that the modular chaos bound needs be saturated \cite{Faulkner:2018faa}, since graviton exchange leads to saturating chaos like growth yet an infinite tower of higher spin fields is necessary to change this growth to sub-maximal \cite{Camanho:2014apa,Caron-Huot:2017vep}, and so anything other than saturation would lead to a contradiction with the assumed  split property. Graviton exchange is dual to the stress tensor, hence the initial assumption on the existence of a single trace field for the stress tensor. 
It is then reasonable to further conjecture that theories saturating the modular chaos bound are local quantum theories since we can use this saturation to define local Lorentz frames in the bulk \cite{Faulkner:2018faa}.  
Hence the conjecture: split property for $\mathcal{C}$ $\iff$ local quantum field theory. 

To connect this to the question of Haag dual extensions we consider the deformation results in \cite{Lewkowycz:2018sgn}. The claim we can extract from this paper is that a theory with a local gravitational dual (this seems to be an implicit assumption in \cite{Lewkowycz:2018sgn} since they were using the covariant phase space formalism) then the equality of bulk and boundary modular Hamiltonians is maintained under deformations. These deformations include both state deformations, following ideas in \cite{Faulkner:2014jva,Faulkner:2017tkh}, and shape deformations $O$, following for example \cite{Faulkner:2015csl,Faulkner:2016mzt,Balakrishnan:2016ttg}. Here, compared to the relaxed assumptions in \cite{Lewkowycz:2018sgn}, we are assuming the equations of motion are satisfied and that the bulk algebra that defines the bulk modular Hamiltonian is associated to the extremal surface. Usually the equations of motion are assumed as
part of the formalism (the Weyl operators in our approach are based on a symplectic form that comes from a phase space of on-shell solutions). Certainly we are allowed to assume the extremality condition since the equality of modular Hamiltonians for any algebra is a necessary condition for an entanglement wedge, so if it is true for some wedge then we are
done. See Theorem~\ref{thm:JLMS}$(iii)_a$. Presumably this is also a sufficient condition if were to do similar perturbative computations of  the relative modular operator, or even just relative entropy as in Theorem~\ref{thm:JLMS}$(ii)$. We expect that these arguments then imply that $\mathcal{E}(O') = \mathcal{E}(O)'$ is maintained under deformation - so starting with the vacuum code where we can independently prove Haag duality for double cones \cite{toappear} we arrive at the conjecture:
\begin{align}
\nonumber
\text{split property for}\,\, \mathcal{C} & \iff \text{modular chaos saturation} \\&  \iff  \text{local quantum field theory}  \iff \mathcal{E}(O') = \mathcal{E}(O)' \,\, \forall O \sim \diamondsuit
\label{eq:conjecture}
\end{align}
where $O \sim \diamondsuit$ denotes all regions smoothly deformable to a diamond. 

To develop this we would then need to work out what it means to deform between different asymptotic codes along with the question of how to define an asymptotically vacuum code in the first place. We should at the very least include some conditions that allow us deform such codes to the actual vacuum code. 
Of course there were many hand-wavy uncertain steps here, so we emphasize this is a conjecture that could be wrong. Or perhaps a conjecture that needs some refinement. 
Finding simple counterexamples would be useful way to make such refinements, codes constructed using scaling limits near phase transition (see the next subsection) seem likely to violate some of these conjectures and necessitate refinement. Also by-passing the middle condition of \eqref{eq:conjecture} (which is anyway vague) might be a useful first step. 
We point out that \cite{Lewkowycz:2018sgn} argues for the absence of certain enhancements in the modular flow correlation functions based on expansions near the extremal surface and the covariant phase space formalism. This is
analogous to saturation of the modular chaos bound where new enhancements, violating $\lim_N \Delta_{\mathcal{M}_N(O)}^{is} V_N - V_N \Delta_{\mathcal{E}(O)}^{is}  \neq 0 $, might come from a tower of higher spin fields but these would then violate the split. For the remaining implication in \eqref{eq:conjecture} (without the middle condition) one would have to argue that these enhancements are necessarily there for a split violating code - in particular the fields that violate the split cannot simply decouple from the deformed modular operators. Presumably this decoupling from the modular operator is not allowed due to cyclicity type arguments. 
%The split property is presumably another way to do this, in particular we again expect saturation of the modular chaos bound is ne

Ultimately one might hope that
this conjecture underpins Jacobson's program to relate  thermodynamics of local Rindler Horizons to the Einstein's equations \cite{Jacobson:1995ab,Lashkari:2013koa,Faulkner:2013ica,Bueno:2016gnv,Lewkowycz:2018sgn,Jacobson:2015hqa}. Here perhaps Haag duality can be thought of as some local equilibrium assumption (via the KMS condition on the modular automorphism group) that then gives rise to Einstein's equations. 

\subsection{Phase transitions}

Phase transitions between different entanglement wedges have given important input on our understanding of quantum information aspects of quantum gravity. More recently they have been used to study the Page curve of black hole evaporation in realistic gravitational theories. Thus understanding how to describe phase transitions in QEC codes is an important issue. In our asymptotic codes the physics of phase transitions will require more than one code to study. Each code being related to a different phase. This is familiar from the study of (first order) phase transitions in the thermodynamic limit of statistical mechanical systems. 

In the evaporating black hole case there is a phase transition between the entanglement wedge of the radiation as a function of how much radiation is collected. At the page time the entanglement between interior modes and the Hawking radiation is screened by the appearance of an island in the entanglement wedge, via an application of the island rule which determined the entanglement wedge in gravity.
Naively this entanglement persists in the effective bulk theory, despite the entanglement seemingly violating the rank condition on entanglement between a the black hole microstates and the radiation. 
In this case non-isometric codes are particularly natural  since the
bulk state that is used as an input to the island formula would only be possible if the quantum fields inside the black hole have more states than is suggested by the black hole area formula. For an isometric code this
would lead to a paradox for the rank condition, hence non-isometric codes become natural. We then claim such codes do not lead to large tensions with the unitarity of the bulk gravitational description as
long as we use the code to study states in the boundary theory that survive the large-$N$ limit.
More specifically, to describe the Page curve situation using asymptotic codes we need two codes for before and after the Page time. We then suppose that states in one code (say post page) can be written as states in the other code (pre page) via the application of bulk $N$ dependent unitaries. In this case the unitaries would entangle the Hawking radiation from the interior and lead to tensions with unitarity discussed above. However since these $N$ dependent states need not survive the large-$N$ limit there is no immediate issue.

Along these lines we propose the Hawking paradox arises due to a non-commutativity of the large-$N$ limit and the $\mathcal{O}(N^\#)$ semiclassical time evolution required to see the paradox. 
The importance of non-isometric codes for resolving the Hawking paradox was discussed with increasingly levels of urgency in \cite{Penington:2019npb,Akers:2021fut,Akers:2022qdl}. In our paper we do not attempt to define a code that simultaneously works for states that are pre and post the page time, and so our discussion is not directly connected to these works. One might conjecture that they are tied together at finite $N$, although we do not make this precise. Another possibility is to study a scaling limit where one sits exactly on top of the Page transition. These are important avenues for future study. 
% a regime we ignore in this paper. 

\appendix

	\section{Modes of convergence}

\label{app:convergence}

Here we summarize some of the standard mathematical properties of topological vector spaces that are required. See text book discussions in \cite{kadison1997fundamentals,bratteli2012operator}. 

 Given a Hilbert space $\Hil$ there is natural topology on the vector space determined by the norm from the inner product of the Hilbert space. 
Convergence of a sequence of vectors:
\begin{equation}
\lim_N \left| \xi_N \right> = \left| \xi \right> \quad \iff \quad \lim_N \| \left| \xi_N \right> - \left| \xi \right> \| = 0
\end{equation}
Weak convergence applies to linear functionals on $\mathscr{H}$:
\begin{equation}
w-\lim_N \left| \xi_N \right> = \left| \xi \right> \quad \iff \lim_N \left< \eta \right| \left( \left| \xi_N \right> - \left| \xi \right>  \right) = 0 \qquad \forall \left< \eta \right| \in \mathscr{H}^\star
\end{equation}
This notion of weak and norm convergence applies to any Banach space. 

%	\textbf{$C^*$-algebra} $\M$ is defined by
%	\begin{itemize}
%	    \item[(1)] is complete with respect to $\norm{\cdot}$ (in Banach space $\norm{xy}\leq\norm{x}\norm{y}$)
%	    \item[(2)] is closed under involution (taking adjoint) $\norm{x^\dagger}=\norm{x}$, with $\norm{xx^\dagger}=\norm{x}^2$
%	\end{itemize}

We now consider the algebra of bounded linear operators $\B(\Hil)$ as a topological vector space. 
There are various topologies we can consider.
The first is the uniform topology
defined by the norm $\| \cdot \|$:
	\begin{eqnarray}	 \norm{x}=\sup_{\xi\in\Hil}\frac{\norm{x\xi}}{\norm{\xi}}, \qquad \forall~\xi\in \Hil
	\end{eqnarray}
		for $x\in\B(\Hil)$.  One can check that $\norm{x^\dagger x}=\norm{x x^\dagger}=\norm{x}^2$. 
  A sequence of operators $x_N$ converges
  to an operator $x$ in the uniform topology iff $\lim_N \| x_N -x \| =0$. We denoate this $x_N \rightarrow x$.
The set of operators $\mathcal{B}(\Hil)$ is complete with respect to this norm making $\mathcal{B}(\Hil)$ a simple example of a C$^\star$ algebra. A subset $\mathcal{W} \subset \mathcal{B}(\Hil)$ is closed in the uniform topology iff all Cauchy sequences $\{ x_N\}_N \subset \mathcal{W}$ converge to an operators in $\mathcal{W}$.

Other topologies include the locally convex topologies  generated by the family of seminorms:
	\begin{itemize}
	\item \textit{weak operator topology} (wot)
	\begin{align}
	 x\mapsto\lvert( \xi,x\eta)\rvert,\qquad~\forall~\xi,\eta\in\Hil
		\end{align}
		\item \textit{strong operator topology} (sot)
		\begin{align}
	 x\mapsto\norm{ x\xi},\qquad~\forall~\xi\in\Hil
		\end{align}
% 		\item \textit{strong operator topology} (SOT) is a pointwise convergence on $\B(\Hil)$ generated by basis of neighbourhoods of $x\in\B(\Hil)$ of the form
% 		\begin{align}
% 	 N(x;\xi_1,\cdots,\xi_n;\epsilon):=\{y\in\B(\Hil)|\norm{(x-y)\xi_i}<\epsilon;~\xi_i\in\Hil;~i=1,\dots,n\},
% 		\end{align}
% 		with $\epsilon>0$.
		\item \textit{strong$^*$ operator topology} (s$^\star$ot)
		\begin{align}
	 x\mapsto\left(\norm{x\xi}^2+\norm{x^\dagger\xi}^2\right)^{1/2},\qquad\forall~\xi\in\Hil
		\end{align}
	\end{itemize}
These different topologies are defined by the fact that they are the weakest topology that makes respectively each of the above maps $\mathcal{B}(\Hil) \rightarrow \mathbb{R}_{\geq 0}$ continuous. 
Convergent sequences in each topology will respectively be denoted:
\begin{align}
wo-\lim x_N &= x \iff \lim_N \left< \xi \right| x_N \left| \eta \right> = 
\left< \xi \right| x \left| \eta \right>
 \qquad \forall \xi,\eta \in \mathscr{H} \\
 so-\lim x_N &= x \iff \lim_N  x_N \left| \xi \right> = x \left| \xi \right>
 \qquad \forall \xi \in \mathscr{H} \\
  s^\star o-\lim x_N &= x \iff \lim_N  x_N \left| \xi \right> = x \left| \xi \right> \,\, {\rm and}\,\,  \lim_N  x_N^\dagger \left| \xi \right> = x^\dagger \left| \xi \right>
 \qquad \forall \xi \in \mathscr{H}
\end{align}
Sequential completeness using the convergence of sequences with respect to the above semi-norms is not sufficient to define closed subsets of $\mathcal{B}(\Hil)$, in contrast to the uniform case above. Instead we would need the notion of a convergent net. The details of this will not be needed in this paper. 

 Unital $*$-subalgebras of $\mathcal{N} \subset \mathcal{B}(\mathscr{H})$ that are closed with respect to any of the above three topologies are von Neumann algebras.  They are equal to their double commutant $\mathcal{N} = \mathcal{N}''$. The set of all operators  $\mathcal{B}(\mathscr{H})$ is the simplest such example. 

There are a few more operator topologies that we will occasionally need. 
These will be based on the linear functionals $\sigma$ on $\mathcal{B}(\mathscr{H})$ given by $\sigma(\cdot) =  {\rm Tr}_{\mathscr{H}}T \cdot$ for the trace class operators $T$. 
This is a Banach space denoted $\mathcal{B}(\mathscr{H})_\star$
with respect to the linear functional norm:
\begin{equation}
\| \sigma \| = \sup_{x \in \mathcal{B}(\mathscr{H})}\frac{ |\sigma(x)|}{\| x\|}
\end{equation}
This space is not the same as the (uniform) continuous dual of $\mathcal{B}(\mathscr{H})$ which is generally a larger space. 

We can give a general description of such a $\sigma \in \mathcal{B}(\mathscr{H})_\star$ in terms of two infinite sequences of vectors $\eta_i$ and $\xi_i$ with:
\begin{equation}
\sigma(x) = \sum_i \left< \eta_i \right|  x \left| \xi_i \right>
\qquad \sum_i \| \xi_i \|^2, \sum_i \| \eta_i \|^2 < \infty
\end{equation}
Positive linear functionals determined by trace class operators $\rho \in \mathcal{B}(\mathscr{H})_\star^+$are called normal states on $\mathcal{B}(\mathscr{H})$ and take the general form:
\begin{equation}
\rho(x) = \sum_i \left< \xi_i \right|  x \left| \xi_i \right>
\qquad \sum_i \| \xi_i \|^2 < \infty
\end{equation}

 A new set of locally convex topologies on $\mathcal{B}(\mathscr{H})$  is determined by the following families of seminorms:
		\begin{itemize}
		    \item 	\textit{ultraweak topology}\label{ultraweak} is determined by
	\begin{align}
	    x\mapsto |\sigma(x)|
     %(x)=\lvert\sum_{i=1}^{\infty}\langle x\xi_i,\eta_i\rangle\rvert,\qquad \sum_{i=1}^{\infty}\norm{\xi_i}^2<\infty,~ \sum_{i=1}^{\infty}\norm{\eta_i}^2<\infty
	\end{align}
 for all $\sigma \in \mathcal{B}(\mathscr{H})_\star$.
	\item  \textit{ultrastrong topology} is determined  by
	  \begin{eqnarray}
	    x\mapsto \rho(x^\dagger x)^{1/2}
     %=\left(\sum_{i=1}^{\infty} \norm{ x\xi_i}^2\right)^{1/2},\qquad \sum_{i=1}^{\infty} \norm{ \xi_i}^2<\infty
	  \end{eqnarray}
	  for all positive $\rho\in \mathcal{B}(\mathscr{H})_*^+$.
	  \item 	  \textit{ultrastrong$^*$ topology} is determined by
	  \begin{align}
	     x\mapsto \left(\rho(x^\dagger x)+\rho( x x^\dagger)\right)^{1/2},\qquad\forall~\rho\in X_*^+
	  \end{align}
	  for all positive $\rho\in \mathcal{B}(\mathscr{H})_*^+$.
		\end{itemize}
  where again the topologies are the weakest topology such that the respective maps above are continuous.
We define convergence of sequences of operators with respect to these topologies in the standard way, although we do not introduce a notation for this since we will not need to use it as much. 

The pre-dual $\mathcal{N}_\star$ of a von Neumann algebra $\mathcal{N} \subset \mathcal{B}(\mathscr{H})$ is the Banach space of ultraweakly continuous linear functionals. The normal states $\mathcal{N}_\star^+$ are the positive versions of these linear functionals. 
It is called the pre-dual because the (norm) continuous dual of $\mathcal{N}_\star$ is the original von Neumann algebra $(\mathcal{N}_\star)^\star
= \mathcal{N}$. The weak topology on $\mathcal{N}_\star$ is the topology of pointwise convergence. That is if a sequence in the predual $\rho_N \in \mathcal{N}_\star$ converges to $\rho \in \mathcal{N}$ weakly we denote this:
\begin{equation}
w-\lim \rho_N = \rho \iff  \lim_N \rho_N(x) = \rho(x)\, \qquad \forall x\in \mathcal{N}
\end{equation}

	 We have the following relation between these various topologies with the following ordering from coarse to fine \cite{bratteli2012operator}:
	\begin{align}
	&\mathrm{ultraweak}~~\prec ~\mathrm{ultrastrong}\prec ~~\mathrm{ultrastrong}^\star~\prec\mathrm{uniform}\non\\
	&~~~~~~~\curlyvee~~~~~~~~~~~~\curlyvee~~~~~~~~~~~~~~~~~~\curlyvee\non\\
	&   ~~~~~~\mathrm{wot}~~~~\prec~~~~\mathrm{sot}~~~~~~~~\prec ~~~~~\mathrm{s}^\star\mathrm{ot}
	\end{align}
	where "$\prec$" means that the right hand side is finer than the left hand side, i.e.  if a set is closed in WOT, it is closed in SOT, etc, which indicates that if a sequence converges in the SOT it converges in the WOT, etc.

Generally a sequence that converges in the ``ultra'' version of some operator topology also converges in the non-ultra version as long as the sequence of operators is uniformly bounded. 
We will need this in the weak case and so we give a proof:
\begin{lemma}
\label{ultratonot}
Consider a sequence of uniformly bounded operators $x_N \in \mathcal{N}$ converging $x_N \rightarrow x$ to $x \in \mathcal{N}$ in the weak operator topology. Then for all $\sigma \in \mathcal{N}_\star$ the following limit exists:
\begin{equation}
\lim_N \sigma(x_N) = \sigma(x)
\label{toconcls}
\end{equation}
\end{lemma}
\begin{proof}
Since we can express $\sigma = \sum_i \left<\eta_i\right| \cdot \left|\xi_i\right>$ with $\sum_i \| \xi_i \|^2, \sum_i \| \eta_i \|^2 < \infty$ we have:
\begin{equation}
|\sigma(x_N -x)|
\leq \sum_i | \left< \eta_i \right| (x_N - x) \left| \xi_i \right> |
\leq (\| x \| + \sup_N \| x_N \|) 
\sum_i \| \xi_i \| \| \eta_i \|
< \infty
\end{equation}
where we used the assumed uniform bound on $x_N$ and the Cauchy-Schwarz
inequality on the infinite sum.
Hence $(\| x \| + \sup_N \| x_N \|)  \| \xi_i \| \| \eta_i \| \geq | \left< \eta_i \right| (x_N - x) \left| \xi_i \right> |$ is a dominating sequence whose sum converges. Hence by the dominated convergence theorem applied to sequences we can take the limit $\lim_N$ inside the sum and conclude \eqref{toconcls}.
\end{proof}

	\begin{theorem*}{Banach-Steinhaus theorem}[uniformly boundedness principle]
 \label{ubp}
	
	Let $A$ be Banach spaces, $B$ be a norm space, and $\{T_N\}_{N\in\mathbb{N}}\subset\B(A,B)$ are a sequence of bounded linear operators. Then the following statements are equivalent:
	\begin{itemize}
	    \item[(i)] (pointwise boundedness on $A$) for all $\xi\in A$, there exist $C_\xi\geq0$, s.t.
	    \begin{align}
	   \sup_{N \in\mathbb{N}} \norm{T_N \xi} <C_\xi
	    \end{align}
	    \item[(ii)] (uniformly boundedness) there exist $C\geq0$, s.t.
	    \begin{align}
	   \sup_{N\in\mathbb{N}} \norm{T_N}<C
	    \end{align}
	\end{itemize}
	\end{theorem*}

 \begin{corollary}
Suppose that $\lim_N T_N \xi$ exists for $\xi$ in a dense subspace of $A$ then the following are equivalent: 
\label{cor-ubp}
\begin{itemize}
\item[(i)] The operators are uniformly bounded $\| T_N \| < C$ for all $N$.
\item[(ii)] $\lim_N T_N \xi$ exists for
all $\xi \in A$.
\item[(iii)] For all $\xi\in A$, the map $T\in\B(A,B)$ defined by 
\begin{align}
    T\xi=\lim_{N\in\mathbb{N}} T_N\xi
\end{align}
is a bounded linear map.
\end{itemize}
 \end{corollary}

In the main text we used this in various forms.
For Lemma~\ref{lem:C} we used a weak form which
applies the above theorem twice. Consider the theorem with $A = \mathscr{H}$ and $B = \mathbb{C}$ allowing us to discuss convergences of the linear functionals $T_N = \left< \psi_1 \right| V_N^\dagger V_N$ for any fixed $\psi_1 \in \mathscr{H}$. That is weak convergence guarantees, via Corrolary~\ref{cor-ubp}, $\| V_N^\dagger V_N \left| \psi_1 \right> \| < C_{\psi_1}$ for all $N$. Another application of Banach-Steinhaus Theorem~\ref{ubp}
with $A = B = \mathscr{H}$ and $T_N = V_N^\dagger V_N$ gives $V_N^\dagger V_N < C$ for some $C$. 

Since $V_N : \mathscr{H} \rightarrow \mathscr{K}_N$ with different Hilbert spaces as the image we need a slightly genaralized version of strong operator convergence.
In particular we can only define convergence to $0$ in this case:
\begin{equation}
so-\lim_N T_N = 0 \qquad \iff \qquad \lim_N \| T_N \left| \xi \right> \| =0 \,\, \forall \xi \in \mathscr{H}
\end{equation}
for sequences $T_N \in \mathcal{B}(\mathscr{H}, \mathscr{K}_N)$.

The following lemma makes clear the relation between types  of convergence for an asymptotically isometric code:
\begin{lemma} 
\label{lem:seq}
Given a sequence of bounded operators $V_N \in \mathcal{B}( \mathscr{H}, \mathscr{K}_N)$ such that  $wo-\lim_{N \rightarrow \infty}$ $ V_N^\dagger V_N = 1$, the following three conditions are equivalent for a finite set of sequences of uniformly bounded operators $\{ m_N^1, m_N^2, \ldots m_N^k \}$
with associated limits $\{ n^1, n^2 \ldots n^k \}$:
\begin{itemize}
\item[(1)] Strong convergence:
\begin{equation}
so-\lim_{N \rightarrow \infty} (m_N^i V_N - V_N n^i) = 0 \qquad \forall i=1,2 \ldots k
\end{equation}
\item[(2)] Simultaneous strong convergence (for all $1 \leq \ell \leq k$):
\begin{equation}
so-\lim_{N \rightarrow \infty}  \left( m_N^{i_1}  m_N^{i_2}  \ldots m_N^{i_\ell} V_N - V_N n^{i_1} n^{i_2}  \ldots n^{i_\ell}\right)  = 0 \qquad \forall  \{ i_1, i_2, \ldots i_\ell \} \subset \{1, 2, \ldots k \}
\end{equation}
\item[(3)] Simultaneous weak convergence (for all $1 \leq \ell \leq k$):
\begin{equation}
wo-\lim_{N \rightarrow \infty} V_N^\dagger m_N^{i_1}  m_N^{i_2}  \ldots m_N^{i_\ell} V_N = n^{i_1} n^{i_2} \ldots n^{i_\ell}  \qquad \forall  \{ i_1, i_2, \ldots i_\ell \} \subset \{1, 2, \ldots k \}
\end{equation}
\item[(4)] Pairwise weak convergence:
\begin{equation}
wo-\lim_{N \rightarrow \infty} V_N^\dagger m_N^i m_N^j V_N = n^i n^j  \qquad \forall i,j=1,2 \ldots k
\end{equation}
\end{itemize}
\end{lemma}
\begin{proof}
Clear. 
\end{proof}

\section{Weyl algebra $C^\star$ norm}

\label{app:weyl}

 There are two natural norms \cite{Slawny:1972iq} defined on $\mathcal{W}$:
\begin{equation}
\| \sum_i \lambda_i w_i \|_B = \sum_i | \lambda_i|,\qquad \lambda_i\in\mathbb{C}
\end{equation}
which can be used to close $\overline{\mathcal{W}}^B$ to a unital Banach algebra. This norm does not satisfy the $C^\star$ algebra condition that $\| a^\dagger a \| = \| a\|^2$. 
However by considering positive linear functionals on $\overline{\mathcal{W}}^B$ we can define a new norm:
\begin{equation}
\|  \sum_i \lambda_i w_i\| = \sup_\rho \left( \sum_{ij}\lambda_j^* \lambda_i  \rho( w_j^\star w_i) \right)^{1/2}
\end{equation} 
where $\rho: \overline{\mathcal{W}}^B \rightarrow \mathbb{C}$ and $\rho( a^\dagger a) \geq 0$ and $\rho(1) = 1$. This defines a $C^\star$ algebra norm on $\mathcal{W}$ which can be used to complete this algebra
to a unital $C^\star$ algebra $\overline{\mathcal{W}}^{\| \cdot \|}$. We will continue to denote this algebra as $\mathcal{W}$ for simplicity. 

Note that generally $\| a \| \leq \| a \|_B $. 
An example of a linear functional is:
\begin{equation}
\tau( \lambda_0 + \sum_i \lambda_i w_i) = \lambda_0
\end{equation}
where $w_i$ are not the identity. Hence:
\begin{equation}
\sum_i | \lambda_i | \geq \| \sum_i \lambda_i w_i \| \geq \left( \sum_i | \lambda_i |^2  \right)^{1/2}
\end{equation}
for $w_i$ unique (possibly the identity). The first inequality is the previously discussed $\| a \| \leq \| a \|_B $.

However it is possible to show
 that:
\begin{equation}
\| \lambda_0 +  \sum_i \lambda_i w_i \| = \| \lambda_0 +  \sum_i \lambda_i w_i \|_B = \sum_i |\lambda_i |
\end{equation}
for finite subsets $\{ w_i \}_i \subset \mathcal{W}_u$
such that all  $h_{i}$ are linearly independent.

\section{Single particle Hilbert space}
\label{app:oneparticle}

In this appendix we summarize the so called standard subspace that can be used to define free states of bosonic fields reduced to sub-regions. Our discussion follows \cite{Longo:2021rag}.
We have a real vector space $H$ an inner product $\alpha$ which we assume is non-degenerate 
and we assume $H$ is complete with respect to $\alpha$. $\beta$ is a symplectic form on $H$ that we take to be non-degenerate
and compatible with $\alpha$ \eqref{bdab}. 

Then \eqref{bdab} implies the existence of an operator $D$ on the real Hilbert  space $\H$
such that $\beta(h_1,h_2) = \alpha(D h_1,h_2)$ by the Riesz representation theorem. In particular it is easy
to see that $D^\dagger = -D$ and $\| D \| \leq 1$ where $^\dagger$ and operator norm are defined in the usual way with respect to $\alpha$. 
The polar decomposition gives $D = V | D|$ where $V^\dagger V = 1$, $V^2 = -1$ and $[D, V] =0$. 
The kernel of $D$ is empty due to the non-degenerate condition on $\beta$. 

Suppose that for all $h_1$ there is some $h$ such that:
\begin{equation}
\label{satdom}
\beta(h, h_1)^2 = \alpha(h,h) \alpha(h_1,h_1) \qquad \forall h_1 \in H
\end{equation}
Then $\alpha(h,h) \alpha(h_1,h_1)  = \alpha(Dh_1, h)^2 \leq \alpha( Dh_1, D h_1) \alpha(h,h) \leq \alpha(h,h) \alpha(h_1,h_1) $.
 Implying that  $\alpha(D h_1, D h_1) = \alpha(h_1,h_1)$ which implies that $\| D h_1 \| = \| h_1 \|$. But $\| (D^2+1) h_1\|^2
 = \|  D^2 h_1 \| + 2 {\rm Re} \alpha( h_1, D^2 h_1) + \| h_1 \|^2
 = \|  D^2 h_1 \| - \| h_1 \|^2$.  But $\| D^2 h_1 \| \leq  \| D h_1 \|  = \| h_1 \|^2$ since $\| D\|^2 \leq 1$. 
 Hence $\| (D^2+1) h_1\|^2 \leq 0$ implying equality $(D^2+1)h_1 =0$ for all $h_1$. That is $D^2=-1$. 
 Conversely if $D^2 = -1$ then  consider  $h = D h_1$. This gives \eqref{satdom}. It turns out this is the condition
 for irreducibility of the Weyl algebra. 
 
Now we double the Hilbert space $H_2 = H \oplus H$ with inner product extended in the usual way. 
Define a dilation of $D$ on this real Hilbert space:
\begin{equation}
\iota = \begin{pmatrix} D  &  V (1+D^2)^{1/2}\\    V (1+D^2)^{1/2} & - D \end{pmatrix}
\qquad \pi = {\rm ker}(1+D^2)  \begin{pmatrix} 0 & 0 \\ 0& 1 \end{pmatrix}
\end{equation}
Then $\iota \pi = \pi \iota$, $\pi^2 = \pi =\pi^\dagger$ and $\iota^2 = -1$. 
Thus $\iota$ defines a complex structure on $\mathscr{H}_1 = (1-\pi) H_2$ where the complex inner product is:
\begin{equation}
\left(  \psi_1,  \psi_2 \right)_{\mathscr{H}_1} = \left(  \psi_1,  \psi_2 \right)_{H_2}
+ i  \left( \iota \psi_1,  \psi_2 \right)_{H_2}
\end{equation}
where $\psi_i \in (1-\pi) H_2$ and where recall that the $H_2$ inner product always gives a real number. 
Note that multiplying by a complex number on this new Hilbert space is achieved via $c \psi
= {\rm Re} c \psi + {\rm Im} c \iota \psi$ and this achieves the desired antilinearity in the inner product. 
We then define:
\begin{equation}
\kappa(h) = (1-\pi) h \oplus 0 = h \oplus 0
\end{equation}
from which we find that:
\begin{equation}
{\rm Im} \left( \kappa(h_1) , \kappa(h_2) \right)_{\mathscr{H}_1}
= ( D h_1 \oplus  V(1+D^2)^{1/2} h_1, h_2 \oplus 0)_{H_2} = \alpha(Dh_1, h_2) =
\beta(h_1,h_2)
\end{equation}
and:
\begin{equation}
{\rm Re} \left( \kappa(h_1) , \kappa(h_2) \right)_{\mathscr{H}_1}
= \alpha(h_1,h_2)
\end{equation}
as desired. 

Some physics comments are in order. The subspace $  {\rm ker}(1+D^2)   \oplus 0 \subset \mathscr{H}_1$ gives rise, via the Fock space construction,
 to an irreducible representation of the Weyl unitaries.  That is the commutant vanishes and the fock vacuum is a pure state. This is effectively why we projected
 out the second half of $H_2$ on this subspace - there is no doubling from entanglement with modes on the commutant. And $D$  already defines
 a complex structure on $ {\rm ker}(1+D^2) =  H \oplus 0$ since $D^2 = -1$ here. The complex structure then determines what we call a creation operator
 vs an annihilation operator. If $ {\rm ker}(1+D^2) = 0$ then the subspace $\kappa(H) \subset \mathscr{H}_1$ is called standard. In particular it leads
 to a cyclic and separating fock vacuum. The doubling then represents the commutant modes that are entangled with $\mathcal{C}$. 
 In particular the complex structure now gives the Minkowski like creation and annihilation operators, which are linear combinations of the Rindler like creation and annihilation operators.
 
 The case where $ {\rm ker}(1+D^2) = 0$ the doubled\footnote{Doubled here means in the direct sum sense and not direct product. It will however roughly correspond to a tensor product for the resulting Fock space.} modes are not projected out. Then one might think of the two subspaces $H \oplus H$
 as follows. Suppose that  we compare to a (singular) pure state determined by $\tilde{\alpha}$, compatible with $\beta$, such that $\tilde{D}^2 =-1$.
 Then the complex structure for this state, on the doubled Hilbert space, provides two Rindler like modes with a factorized pure state. Then the
 different choice of complex structure $\iota$ associated to the original state $\alpha$ corresponds to a Bogoliubov transformation, which results in the Hartle-Hawking/Minkowski modes which now gives rise to an entangled state between the two Rindler modes. States on the factorized Fock space will be singular (in the sense of being in a disjoint folium) on this new Fock space. 
 
 It is an exercise to show (see \cite{Longo:2021rag} for details):
 \begin{equation}
 \label{notneczero}
{\rm ker}(D) = \kappa(H) \cap \iota \kappa(H)' = 0 \qquad  {\rm ker}(1+D^2) =  \kappa(H) \cap \iota \kappa(H) 
 \end{equation}
 where  $S' = (\iota S)^{\perp_{\mathbb{R}}}$ for a real closed subspace $S \subset \mathscr{H}_1$. One can also check that ${\rm span}\{ \kappa(H) \cup \iota \kappa(H) \}$
 is dense in $\mathscr{H}_1$. More generally if we also allow for degenerate $\beta$ (but not degenerate $\alpha$) then the first equation in \eqref{notneczero}
 is not necessarily equal to $0$.  The case where both of these kernels is empty gives rise to the notion of a factorial standard subspace:
\begin{definition}
Given some a real vector space $H$ with symmetric positive form $\alpha$ and symplectic form $\beta$, not necessarily complete. We
say that $(H,\alpha,\beta)$ is a \emph{factorial standard subspace} if the $\alpha$-norm completion denoted $(\bar{H},\bar{\alpha},\bar{\beta})$ is such that:
\begin{equation}
\kappa(\bar{H}) \cap  \kappa(\bar{H})' = 0\,\qquad \kappa(\bar{H}) \cap \iota \kappa(\bar{H}) = 0 \, \qquad \kappa(\bar{H}) \vee \iota \kappa(\bar{H}) =  \mathscr{H}_1
\end{equation}
where $\kappa$ was defined above and $S_1 \vee S_2
 = \overline{{\rm span}(S_1, S_2)}$ for a real closed subspace $S_{1,2} \subset \mathscr{H}_1$.
\end{definition}
In this case the Fock representation, or GNS representation associated to $\alpha$, is then cyclic and separating for the Fock vacuum or GNS vector. 

The Fock space:
\begin{equation}
\mathscr{H}_F = \exp(\mathscr{H}_1) = \bigoplus_{n=0}^\infty Sym(\mathscr{H}_1^{\otimes n}) 
\end{equation}
has creation and annihilation operators defined as:
\begin{equation}
a(\mathfrak{h})^\dagger \Psi_n =\sqrt{n+1}( Sym( \mathfrak{h}) \otimes \Psi_n ) \qquad a(\mathfrak{h})  \Psi_n
= \sqrt{n} \sum_{k=1}^{n} (\mathfrak{h}, g_k) \Psi_{n-1}\backslash g_k
\end{equation}
for all $\Psi_n = Sym(g_1 \otimes g_2 \ldots g_n) \in Sym(\mathscr{H}_1^{\otimes n})  $ and $\Psi_n\backslash g_k = Sym(g_1 \otimes \ldots g_{k-1} \otimes g_{k+1} \otimes \ldots g_n)Sym(\mathscr{H}_1^{\otimes (n-1)}) $
In particular the Weyl unitaries are represented as:
\begin{equation}
v(h) = \exp( i (a(\mathfrak{h}) +a(\mathfrak{h})^\dagger)) 
\end{equation}
and these give a representation of the Weyl algebra for $\mathfrak{h} \in \mathscr{H}_1$ considered as a real
Hilbert space with symplectic form ${\rm Im}(\mathfrak{h}_1,\mathfrak{h}_2)_{\mathscr{H}_1}$ and state:
\begin{equation}
\left< 0 \right| v(\mathfrak{h}) \left| 0 \right> = \exp( - 1/2 \| \mathfrak{h} \|_{\mathscr{H}_1}^2 ) 
\end{equation}
where $0 \in \mathscr{H}_F$ is the Fock vacuum. 
Such a Fock rerpesentation is irreducible with pure state $\left| 0 \right>$. 
In particular ${\rm Im}(\mathfrak{h}_1,\mathfrak{h}_2)_{\mathscr{H}_1}^2 \leq \| \mathfrak{h}_1 \|^2_{\mathscr{H}_1} \| \mathfrak{h}_2 \|_{\mathscr{H}_1}^2$ 
can always be saturated by picking $\mathfrak{h}_1 = i \mathfrak{h}_2$, so the discussion around \eqref{satdom} implies irreducibility. 

The real linear subspace $\kappa(\bar{H})$ then allows us to define a representation of the Weyl algebra for $\bar{H}$ on the Fock space
$\exp(\mathscr{H}_1)$.  This representation will not be irreducible if $\kappa(\bar{H}) \cap \iota \kappa(\bar{H}) \neq \mathscr{H}_1$. 
The Weyl unitaries are:
\begin{equation}
\pi_F( w(h)) = \exp( i (a( \kappa(h)) +a(\kappa(h))^\dagger)) \, \qquad h \in \bar{H}
\end{equation}
and satisfies all the desired properties.

\section{More on the center}

\label{app:center}

In this appendix we give some extra information about the possible superselection sectors for the Weyl algebra with a non-trivial center. That is with a degenerate form $\beta$. 
See \cite{binz2004construction,kadison1997fundamentals}. 

In the presence of a center we now do need to confront different possible folia of states. The center is abelian so we can understand this at the level of abelian $C^\star$-algebras and von Neumann algebras.  The Weyl algebra associated to the central operator $U$ defined in \eqref{defU}, is a unital commutative C$^\star$ algebra. Hence it must correspond to continuous
functions on some compact Hausdorff space \cite{kadison1997fundamentals}. This algebra is not the space of continuous functions on $U \in \mathbb{R}$, since this is non-compact.
Rather the $C^\star$ algebra is the space of almost periodic functions $AP(\mathbb{R})$ which are simply the uniform completion of the space of finite linear combinations of characters $\sum_s c_s e^{i s U}$. These can alternatively be described as the space of continuous functions on the
Bohr compactification $\mathbb{R}_b$  of $\mathbb{R}$ which is a compact Hausdorff space \cite{binz2004construction}. 
The positive linear functionals on $AP(\mathbb{R})$ are then determined by Radon measures on $\mathbb{R}_b$. A subset of such linear functionals, termed regular, correspond to any finite positive measure on $\mathbb{R}$ \cite{hewitt1953linear}. These will be the states of interest. There is a one-one correspondence between the bounded continuous ``positive definite'' functions\footnote{A function is ``positive definite'' if the matrix $f_{ij} = f(s_i-s_j)$ is positive semi-definite for any finite subset $\{ s_i \}_i \subset \mathbb{R}$.}:
\begin{equation}
\sigma_\mu( e^{ is U}) = f(s) = \int d U \mu(U) e^{is U}
\end{equation}
and such finite positive measures on $U$. Note that the abelian von Neumann algebra for $\pi_{\sigma_\mu}$ is given by $\pi_{\sigma_\mu}(AP(\mathbb{R}))'' = L^\infty(\mathbb{R}, \mu)$. States on this folium correspond to other measures $\nu$ on $\mathbb{R}$ that are absolutely continuous with respect to $\mu$: $\nu \ll \mu$. 

At finite $N$ the function $\sigma_N( e^{is (X^{(N)}- N_T^{-1/2} \bar{E}_N )})$ of $s$ is continuous and positive definite for CFTs
with reasonable operator spectrum. In order establish this for the limit on $N$ we need some strengthened continuity condition on the finite $N$ spectrum and also some specific details of the sequence of states $\sigma_N$. Rather than attempt to derive this we simply assume continuity of the limiting function. In fact the assumption on the existence of $\alpha$ Definition~\ref{def:largeN}, which determines the limiting correlation function in \eqref{assumpo}, implies the continuity of $\sigma(e^{is U})$. And there are only two options: 1. the Gaussian measure given in \eqref{gaussianmu} which leads to a faithful representation and 2. the $\delta$ function measure \eqref{deltamu} where $\alpha( h^{\mu\nu} =  \xi^{(\nu} \delta_\Sigma^{\mu)} , h_2 ) = 0 \,, \forall h_2 \in \H$, is also degenerate.

\section{Uniform operator system convergence}

\label{app:uniform}

We assume that, comparing to the main text, we have just arrived at the discussion above Definition~\ref{def:fc} in Section~\ref{sec:faith}. We will continue from there with a stronger
convergence assumption that will eventually lead to Definition~\ref{def:fc}.
Let us summarize where we are at:

As usual we have a sequence of QFTs with nets $\mathcal{M}_N$. We have an associated complete single trace algebra $(\H, \mu_N)$
and a large-$N$ sector determined by a sequence of pure state $\psi_N \in \mathscr{K}_N$.
We use these to construct a limiting algebra $\mathcal{W}$, a $\| \cdot \|$-norm dense Banach $^\star$ sub-algebra $\mathcal{W}_B \subset \mathcal{W}$, a
limiting state $\sigma$ and reconstruction maps $\mathcal{\gamma}_N : \mathcal{W}_B \rightarrow \mathcal{M}_N$ that
are unital pointwise bounded $\| \gamma_N(w) \| < \| w \|_B$ where $\| \cdot \|_B$ is the Banach norm for $\mathcal{W}_B$ as discussed in Appendix~\ref{app:weyl}. 
Then $\mathscr{H}$ is the GNS Hilbert space for $\sigma = \omega_{\eta} \circ \pi_\sigma$ where $\pi_\sigma : \mathcal{W}
\rightarrow \mathcal{C} = \pi_\sigma(\mathcal{W})''$. As usual we can derive the asymptotic homomorphism property
for $\gamma_N$:
\begin{equation}
\lim_N \omega_{\psi_N}( \gamma_N(w_1) \ldots \gamma_N( w_q) )\rightarrow \sigma(w_1 \ldots w_q), \qquad \{w_1, \ldots, w_q \} \subset \mathcal{W}_B
\end{equation}
Here we would like to explore a slightly stronger convergence requirement:

\begin{definition}
\label{uosc}
We say that the maps $\gamma_N$ have the \emph{uniform operator system convergence property} if there is a sequence
of operators systems $\mathcal{S}_1 \subset \ldots \mathcal{S}_k \subset \mathcal{S}_{k+1} \ldots \subset \mathcal{W}_B$
for $k \in \mathbb{N}$ that generates $\mathcal{W}_B = \overline{ \cup_k \mathcal{S}_k}^{\| \cdot \|_B}$, and such that for all $\epsilon > 0$ and all  $k \in \mathbb{N}$ there exists some $N(\epsilon;k)$
such that:
\begin{equation}
\label{nhom}
\left| \omega_{\psi_N}( \gamma_N(w_1)  \gamma_N(w_2)) - \sigma(w_1 w_2) \right|^2
< \epsilon \sigma( | w_1^\dagger |^2 ) \sigma( | w_2|^2 )
\end{equation}
for all $N > N(\epsilon;k)$ and all $ w_1, w_2  \in \mathcal{S}_k$.  
\end{definition}

An operator system is a vector subspace of operators that is not closed under multiplication. See \cite{paulsen2002completely}. 

For simplicity of discussion we work with the faithful case in this Appendix. That is $\sigma$ is assumed faithful for $\mathcal{W}$,
so that $\eta$ is cyclic and separating for $\mathcal{C}$. We recal the definition of $V_N$:
\begin{equation}
V_N  \left| [w] \right> = \gamma_N(w) \left| \psi_N \right> \,, \qquad w \in \mathcal{W}_B
\end{equation}
which makes this a densely defined linear operator. Hence the adjoint exists as a closed operator $V_N^\dagger$ on some domain.

 As in the main text we use \eqref{nhom} (the non-uniform
version suffices) to show that:
\begin{equation}
\lim_N \left( V_N \left| [w']\right> , V_N \left| [w'] \right>\right) \,,\qquad \lim_N ( \gamma_N(w) V_N - V_N \pi_\sigma(w) ) \left| [w'] \right> =0
\end{equation}
for all $w,w' \in \mathcal{W}_B$. 

We prove:

\begin{lemma}
\label{lem:uosc}
With the above setup including Definition~\ref{uosc} (we only really need $q=2$) there exists a uniformly bounded sequence of operators $\tilde{V}_N \in \mathcal{B}(\mathscr{H}, \mathscr{K}_N)$ such that:
\begin{align} 
\label{vntvn}
\lim_N (V_N - \tilde{V}_N ) \left| [w] \right> = 0 \qquad \forall w \in \mathcal{W}_B \\
\label{fpart}
so-\lim_N \tilde{V}_N^\dagger \tilde{V}_N = 1\,, \qquad so-\lim_N \gamma_N(w) \tilde{V}_N - \tilde{V}_N \pi_\sigma(w) = 0 \qquad \forall w \in \mathcal{W}_B
\end{align}
And in this way we have constructed an asymptotically isometric code, Definition~\ref{def:code}. 
\end{lemma}
\begin{proof}
Define the nested projections $\pi_k = \overline{\pi_\sigma(\mathcal{S}_k)\left| \eta \right>}$ onto these closed subspaces.  So that $\pi_k \left| [w] \right> = \left| [w] \right>$ 
for all $w \in \mathcal{S}_k$. 
Fix some $k \in \mathbb{N}$ and write \eqref{nhom} as:
\begin{align}
&| \left< [w_1] \right| \pi_k V_N^\dagger V_N \pi_k \left| [w_2] \right> 
- \left< [w_1] \right|  \pi_k \left| [w_2] \right> |^2 \\
&
\qquad \leq \epsilon \|  \pi_k \left| [w_1] \right> \|^2 \|  \pi_k \left| [w_2] \right> \|^2
\end{align}
and since $\left| [w_1] \right> $ are dense on $\pi_k \mathscr{H}$ we can take the supremum over these states to compute the norm:
\begin{equation}
\label{2case}
\| \pi_k V_N^\dagger  V_N \pi_k \| - 1 \leq \| \pi_k \left( V_N^\dagger  V_N -  1 \right) \pi_k \|
\leq \epsilon 
\end{equation}
which also implies that $V_N \pi_k$ is a bounded operator for sufficiently large $N$, namely $N> N(\epsilon;k)$. 

We are free to pick $N(\epsilon;k)$ monotonically increasing with $k$ and decreasing with $\epsilon$. 
Now consider a function $k(N) \in \mathbb{N}$ 
such that $N(1/k(N');k(N')) < N'$ and $k(N) \rightarrow \infty$ as $N \rightarrow \infty$. It is always possible to find such a function. It has the desired property that for for all $\delta> 0$ there exists an $N_\star$ such that $k(N) > 1/\delta$ and hence $\| \pi_{k(N)} (V_N^\dagger V_N -1) \pi_{k(N)} \| < \delta $ for all $N > N_\star$. We also have
the uniform bound on the norm $\| \pi_{k(N)} V_N^\dagger V_N \pi_{k(N)}\| < 1 + 1/k(N)$. 
Define:
\begin{equation}
\tilde{V}_N \equiv V_N \pi_{k(N)}
\end{equation} 
We have that:
\begin{equation}
\lim_N \| \tilde{V}_N^\dagger \tilde{V}_N - \pi_{k(N)}\| = 0
\end{equation}
Now for all $w \in \cup_k \mathcal{S}_k$ there is some $N_\star$ for which $(\pi_{k(N)}-1) \left| [w] \right> =0 $ for all $N > N_\star$. Hence:
\begin{equation}
\lim_N (V_N - \tilde{V}_N ) \left| [w] \right> = 0 
\end{equation}
for all $w \in \cup_k \mathcal{S}_k$.
For fixed $w_0 \in \cup_k \mathcal{S}_k$ considering the sequence of functions:
\begin{equation}
f_N(w) = \left(  (V_N - \tilde{V}_N ) \left| [w_0] \right>, (V_N - \tilde{V}_N ) \left| [w] \right> \right)
\end{equation}
then we have $\lim_N f_N(w) = 0$
for $w \in \cup_k \mathcal{S}_k$. Since these are linear functionals on a dense subspace of the Banach space $\mathcal{W}_B$ we can extend this limit to all of $w\in \mathcal{W}_B$ by uniform boundedness (Banach-Steinhaus, see Appendix~\ref{app:convergence}). A similar analysis for varying $w_0$ gives \eqref{vntvn}.

We also have strong convergence on a dense subspace $w \in \cup_k \mathcal{S}_k$:
\begin{equation}
\lim_N \| \tilde{V}_N^\dagger \tilde{V}_N \left| [w] \right> - \left| [w] \right> \| = 0
\end{equation}
Hence by uniform boundedness we simply have strong convergence as claimed. \end{proof}

\section{Additivity and the causal wedge}

\label{app:causal}

Consider working with a CFT that has a known gravitational dual and the code is dual
to quantum fields on some background asymptotically AdS manifold $M$. We assume such a space is globally hyperbolic in the AdS sense. 
Let us define the bulk net of quantum fields as $\widetilde{\mathcal{C}}(D)$, which we assume are
von Neumann algebras labelled by causally complete bulk regions $D \subset M$. The bulk Hilbert space will be unitarily equivalent to the Hilbert space $\mathscr{H}$ that we
reconstruct with the single trace operators. Hence we can take $\widetilde{\mathcal{C}}(D)$ to act on $\mathscr{H}$. 
For a sufficiently small diamonds on the boundary $\vardiamondsuit_i$, the state used to construct $\mathcal{C}$ approaches that of vacuum and so we can identify $\mathcal{C}(\diamondsuit_i) = \widetilde{\mathcal{C}}(\vardiamondsuit_i)$ with the algebra of operators in the corresponding bulk causal diamond.

For some boundary region $O$ the geometric causal wedge is usually defined as $J^{+}(O)  \cap J^{-}(O) $ where  $J^{\pm}(O)$ are the bulk chronological past and future of the boundary region $O$. 
Unfortunately this is often (generically) not a domain of dependence\footnote{We thank Netta Engelhardt and Veronika Hubeny for explaining this to us.}, so it is a rather unnatural region from a bulk quantum fields point of view. 
Instead 
we work with the following causal domain:
\begin{equation}
C_D(O) = O^{\perp \perp} \subset M
\end{equation}
where $S^{\perp} = M \backslash ({\rm cl} (J^+(S) \cup J^-(S)) )$ which is the causal complement operation on $S$ defined in terms of the bulk causal structure.
This is equivalent to the prime $'$ notation except with respect to the bulk causal structure on $M$ instead of the boundary causal structure. 
Generally $J^{+}(O)  \cap J^{-}(O ) \subset C_D(O)$ and $C_D(O)$ is the smallest domain of dependence that still contains $J^{+}(O)  \cap J^{-}(O)$. 
Hence it is very similar to the causal wedge and
so we will often refer to it as such. 

Note that $(\diamondsuit_i)^{\perp\perp} = \vardiamondsuit_i$ for sufficiently small (but finite) diamonds:  $ \vardiamondsuit_i$
is the causal wedge of $O$ but the difference with  $(\diamondsuit_i)^{\perp\perp}$ only arises after caustics occur when shooting light-rays from the boundary. 
But we can always find a sufficiently small diamond with no such caustics. 

Our goal is to give evidence for the claim $\mathcal{C}(O) = \widetilde{\mathcal{C}}(O^{\perp\perp})$. 
Consider a boundary cover $\cup_i \diamondsuit_i = O$ for small diamonds so that we can easily identify $\mathcal{C}(\diamondsuit_i) =\widetilde{\mathcal{C}}(\vardiamondsuit_i)$ as discussed above.
Then: $C_D(O) = (\cup_i \vardiamondsuit_i)^{\perp \perp}$. 
Suppose that:
\begin{equation}
\label{claimE}
\widetilde{\mathcal{C}}(C_D(O)) = \bigvee_i \widetilde{\mathcal{C}}(\vardiamondsuit_i)
\end{equation}
Then $\widetilde{\mathcal{C}}(O^{\perp\perp}) = \widetilde{\mathcal{C}}(C_D(O)) = \bigvee_i \mathcal{C}(\diamondsuit_i) = \mathcal{C}(O)$ by the additivity property of $\mathcal{C}$. 

So we would establish the claim if we could show \eqref{claimE}. Note that $\cup_i \vardiamondsuit_i \neq C_D(O)$, so this result does not follow from the form of additivity 
that we have worked with in this paper (now of course applied to the bulk.)
In fact it is known that much stronger forms of additivity apply to many theories.\footnote{The strongest possible condition is as follows \cite{Pontello:2020csg}. Both von Neumann algebras
and the sets of causally complete regions form a poset (partially ordered set) under inclusion. The corresponding $\vee, \wedge = \cap$ operations then give this
poset a lattice structure. The lattices have an involution given by the prime operation - the algebra commutant or $\perp$.
We can define the $\vee$ operation as $D_1 \vee D_2 = (D_1^\perp \cap D_2^\perp)^{\perp}$. Then a possible, but very strong statement of additivity follows
by assuming the map $D \rightarrow \widetilde{\mathcal{C}}(D)$ is a lattice homomorphism. There are known counter examples to this strong form arising from free theories, so we will not rely on this property.} For example for theories satisfying the timeslice axiom (not generalized free fields), one only needs additivity for regions covering a spatial time slice. We expect the bulk theory to be a well behaved QFT, so this extended additivity is reasonable. However here we need something different since the $\vardiamondsuit_i$ cover a time-like boundary slice. In fact Borchers timelike tube theorem \cite{borchers1961vollstandigkeit} is exactly such an example - allowing for the causal extension of a time-like cylindrical region of Minkowksi space. Araki found generalizations of Borchers setup \cite{araki1963generalization}. Based on this we expect \eqref{claimE} is always true. 

This discussion should be viewed as replacing the common refrain about the existence of a non-standard Green's function, inside the causal wedge, that evolves data in a space-like direction from $O$. The plausible existence of such a Green's function is usually used to arrive at the causal wedge duality. The relation to Borchers timelike tube theorem was first pointed out in \cite{Duetsch:2002hc} and discussed further in \cite{Levine:2020upy}. The claim that the dual causal region is $D_C(O)$ rather than $J^{+}(O)  \cap J^{-}(O )$ seems to be unappreciated. Although, see footnote~1 of \cite{Benedetti:2022aiw} which came out while we were finishing this paper.

 We might also bypass the additivity argument altogether, and instead follow the Green's function deformation argument given in Section~\ref{sec:sta}. In fact a rigorous version
 of this argument would be tantamount to proving a more general version of Borcher's theorem. We would work at the level of the bulk single particle Hilbert space $\widetilde{\mathscr{H}}_1$, and the standard subspace - the conjecture would be that solutions to the bulk equations of motion for boundary sources in $O$, as elements of the real space $\widetilde{H} \subset \widetilde{\mathscr{H}}_1$, are dense in the standard subspace $\widetilde{H}(O^{\perp \perp})$ for the causal domain.
 
 \begin{figure}[h!]
\centering
\includegraphics[scale=.4]{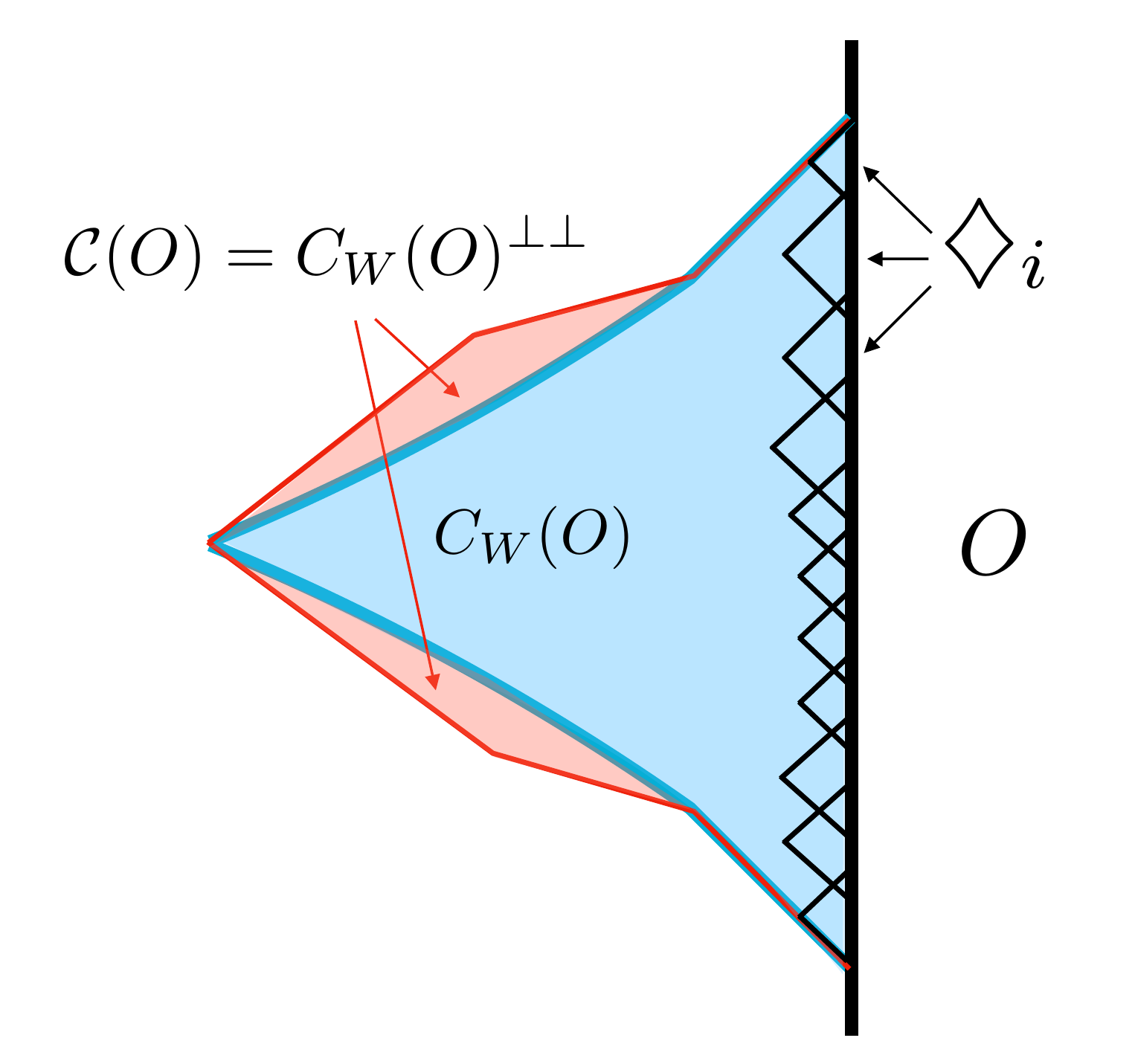}
\caption{Additivity property for the causal domain.}
\end{figure}

\acknowledgments

We thank Chris Akers, Netta Engelhardt, Animik Ghosh, Veronika Hubeny, Hong Liu and Antony Speranza for helpful discussions.
This work is partially supported by the Air Force Office of Scientific Research under award number FA9550-19-1-0360 and the Department of Energy under award number DE-SC0019183.

\bibliography{entanglement}

\end{document}